\documentclass{tlp}
\usepackage{amsmath,amssymb}
\usepackage{stmaryrd}
\usepackage[all]{xy}
\usepackage{fancybox}
\usepackage{graphicx}
\usepackage[breaklinks]{hyperref}

\newtheorem{theorem}{Theorem}[section]

\newtheorem{proposition}[theorem]{Proposition}
\newtheorem{example}[theorem]{Example}
\newtheorem{definition}[theorem]{Definition}
\newtheorem{corollary}[theorem]{Corollary}

\let\oldproof=\proof
\renewenvironment{proof}{\intheoremtrue\oldproof}{\hspace*{1em}\hbox{\proofbox}\endtrivlist}

\newcommand{\vcenterbox}[1]{\raisebox{0.5\depth}{#1}}
\newcommand{\exproofbox}{\hspace*{1em}\hbox{\proofbox}}

\newcommand{\dist}{2cm}
\newcommand{\dista}{1.5cm}
\newcommand{\ovalee}[1]{\ovalbox{%
\parbox[c][.4cm][c]{.7cm}{\begin{center}$#1$\end{center}}}}
\newcommand{\ovalemidi}[2]{\ovalbox{%
\parbox[c][.4cm][c]{.85cm}{\begin{center}$#1 (#2)$\end{center}}}}
\newcommand{\oneedge}[2]{\xymatrix{#1  \ar@{-}[r]^{1} & #2}}
\newcommand{\twoedge}[2]{\xymatrix{#1  \ar@{=}[r]^{2} & #2}}

\newcommand{\fun}{\mathop{\rightarrow}}

\newcommand{\card}[1]{|#1|}
\newcommand{\fwp}{\wp_f}
\newcommand{\Nat}{\mathbb{N}}
\newcommand{\Natzero}{\mathbb{N^+}}

\newcommand{\mwp}{\wp_m}
\newcommand{\supp}[1]{\llfloor #1 \rrfloor}
\newcommand{\emptymulti}{\multil \multir}
\newcommand{\multil}{\{\!\!\{}
\newcommand{\multir}{\}\!\!\}}
\newcommand{\multisum}{\uplus}
\newcommand{\bigmultisum}{\biguplus}
\newcommand{\calS}{\mathcal{R}}

\newcommand{\lub}{\sqcup}

\newcommand{\biglub}{\bigsqcup}

\newcommand{\var}{\mathcal{V}}
\newcommand{\Pos}{\mathrm{\Xi}}
\newcommand{\subst}{\mathit{Subst}}
\newcommand{\Isubst}{\mathit{ISubst}}
\newcommand{\Ren}{\mathit{Ren}}
\newcommand{\vars}{\mathit{vars}}
\newcommand{\occ}{\mathit{occ}}
\newcommand{\rng}{\mathrm{rng}}
\newcommand{\dom}{\mathrm{dom}}
\newcommand{\mgu}{\mathrm{mgu}}
\newcommand{\eq}{\mathrm{Eq}}

\newcommand{\Sharing}{\mathtt{Sharing}}
\newcommand{\ASub}{\mathtt{ASub}}
\newcommand{\Lin}{\mathtt{Lin}}
\newcommand{\Free}{\mathtt{Free}}
\newcommand{\Con}{\mathtt{Con}}
\newcommand{\Def}{\mathtt{Def}}
\newcommand{\Prop}{\mathtt{Pos}}
\newcommand{\PSD}{{\mathit{SS}^\rho}}

\newcommand{\Linp}{\mathtt{ShLin}^{\omega}}
\newcommand{\lp}{\mathrm{\omega}}
\newcommand{\indeg}{\mathit{indeg}}
\newcommand{\outdeg}{\mathit{outdeg}}
\newcommand{\res}{\mathit{res}}
\newcommand{\relev}{\mathit{rel}}
\newcommand{\mG}{{\mathcal G}}
\newcommand{\src}{\mathrm{src}}
\newcommand{\tgt}{\mathrm{tgt}}

\newcommand{\ShLinp}{\mathtt{ShLin}^2}
\newcommand{\Andysh}{\mathit{Sg}^2}
\newcommand{\an}{\mathrm{2}}
\newcommand{\andybin}{\uplus}
\newcommand{\Andybin}{\biguplus}
\newcommand{\downclo}{{\mathop{\downarrow}}}

\newcommand{\ShLin}{\mathtt{ShLin}}
\newcommand{\shl}{\mathit{sl}}
\newcommand{\bin}{\mathrm{bin}}

\newcommand{\PROLOG}{\texttt{PROLOG}}
\newcommand{\difflist}{\mathit{difflist}}

\newcommand{\wrt}{w.r.t.~}
\newcommand{\ie}{i.e., }
\newcommand{\eg}{e.g.\ }
\newcommand{\ra}{\rightarrow}

\newcommand{\Ra}{\Rightarrow}
\newcommand{\seml}{\llbracket}
\newcommand{\semr}{\rrbracket}

\newcommand{\ext}{\mathrm{ext}}

\allowdisplaybreaks

\title{On the interaction between sharing and linearity}

\author[G. Amato and F. Scozzari]
{GIANLUCA AMATO  and FRANCESCA SCOZZARI\\
Dipartimento di Science, Universit\`a ``G. d'Annunzio'' di Chieti-Pescara, Pescara, Italy.\\
\email{amato@sci.unich.it, scozzari@sci.unich.it}}

\submitted {26 March 2008}
\revised {3 July 2008, 15 July 2009}
\accepted{27 July 2009}

\begin{document}

\maketitle

\label{firstpage}

\begin{abstract}
In the analysis of logic programs, abstract domains for detecting sharing and
linearity information are widely used. Devising abstract unification algorithms
for such domains has proved to be rather hard.
At the moment, the available algorithms are correct but not optimal, \ie they cannot fully exploit
the information conveyed by the abstract domains. In this paper, we define
a new (infinite) domain $\Linp$ which can be thought of as a general
framework from which other domains can be easily derived by abstraction.
$\Linp$ makes the interaction between sharing and linearity explicit. We provide
a constructive characterization of the
optimal abstract unification operator on $\Linp$ and we lift it to two
well-known abstractions of $\Linp$. Namely, to the classical
$\Sharing \times \Lin$ abstract domain and to the more precise $\ShLinp$ abstract domain
by Andy King.  In the case of single binding substitutions, we obtain optimal abstract unification
algorithms for such domains.
\end{abstract}


\section{Introduction}

In the analysis of logic programs, the theory of abstract
interpretation \cite{CousotC79,CousotC92lp} has been widely used to
design new analyses and to improve existing ones.
Given a concrete semantics working over a concrete
domain, an abstract interpretation formalizes an  analysis by
providing an abstract domain and an abstract semantics (working on the
abstract domain), and relating them to their concrete counterparts.
An abstract domain is a collection of abstract objects which
encode the property to analyze.  The
concrete and abstract domains are related by means of abstraction and
concretization maps, which allow each concrete object to be abstracted into an abstract object which describes it.
The abstract semantics, in most cases,
is given by a set of abstract operators on the abstract domain, which are the counterparts of the concrete ones.  For example, in the
case of logic programs, one can individuate in the concrete semantics
the main operations (unification, projection, union), and an abstract
semantics can be specified by giving the abstract unification,
abstract projection and abstract union operations. The theory of
abstract interpretation assures us that, for any concrete operator,
there exists a best abstract operator, called the optimal operator.
It computes the most precise result among all possible correct
operators, on a given abstract domain.  Designing the optimal abstract
counterpart of each concrete operator is often a very difficult task.
In fact, even if the definition of the optimal operator for any
abstract domain is known from the theory of abstract interpretation
(as a composition of the concrete operator and the abstraction map),
the hard task is to provide an explicit definition of the abstract operators and
to devise  algorithms on the abstract domain which
compute them.

\subsubsection{The context}
 The property of sharing has been the subject of many papers
\cite{JacobsL92,HW92,MuthukumarH92,CodishSS99,BagnaraHZ-ta-TCS}, 
from the both theoretical and practical point of view. Typical applications
of sharing analysis are in the fields of optimization of unification
\cite{S86} and parallelization of logic programs
\cite{HermenegildoR95}.
The goal of (set) sharing analysis is to detect sets of variables
which share a common variable in the answer substitutions.  For
instance, consider the substitution $\{x/f(u,v), y/g(u,u,u), z/v\}$.
We say that $x$ and $y$ share the variable $u$, while $x$ and $z$
share the variable $v$, and no single variable is shared by $x,y$ and $z$.  Many domains concerning sharing properties also consider
linearity in order to improve the precision of the
analysis. We say that a term is linear if it does not contain
multiple occurrences of the same variable. For instance, the term
$f(x,f(y,z))$ is linear, while $f(x,f(y,x))$ is not, since $x$ occurs twice.

\subsubsection{The problem}
It is now widely recognized that the original domain proposed for
sharing analysis, namely, $\Sharing$ \cite{Langen90,JacobsL92},
 is not very precise, so that it is often combined
with other domains for handling freeness, linearity, groundness or
structural information (see \citeNP{BagnaraZH05TPLP} for a comparative
evaluation). In particular, adding some kind of linearity information
seems to be very profitable, both for the gain in precision and speed
which can be obtained, and for the fact that it can be easily and
elegantly embedded inside the sharing groups (see \citeNP{King94}).  In
the literature, many authors have proposed abstract unification
operators (\eg \citeNP{CDY91,HW92,MuthukumarH92,King94}) for domains of
sharing properties, encoding different amounts of linearity
information.  However, optimal operators for combined analysis of
sharing and linearity have never been devised, neither for the domain
$\ShLinp$ \cite{King94}, nor for the more broadly adopted
$\Sharing \times \Lin$ \cite{HW92,MuthukumarH92}.

With the lack of optimal operators, the analysis loses precision and might even be slower. The latter
is typical of sharing analysis, where abstract domains are usually
defined in such a way that, the less information we have, the more complex the
abstract objects are.  This is not the case for other kinds of
analyses, such as groundness analysis, where the complexity of
abstract objects may grow accordingly to the amount of groundness
information they encode.

The lack of optimal operators is due to the fact that the role played
by linearity in the unification process has never been fully
clarified. The traditional domains which combine sharing and linearity
information are too abstract to capture in a clean way the effect of
repeated occurrences of a variable in a term and most of the effects
of (non-)linearity are obscured by the abstraction process.

\subsubsection{The solution}
We propose an abstract domain $\Linp$ which is able to encode the
\emph{amount} of non-linearity, \ie which keeps track of the exact
number of occurrences of the same variable in a term.  Consider again
the substitution $\theta=\{x/f(u,v), y/g(u,u,u), z/v\}$.  Intuitively,
to each variable $w$ in the range of the substitution, we associate the
multiset of domain variables which are bound to a term where $w$ occurs,
and call it an $\omega$-sharing group.  For instance, we associate, to the variable
$u$, the $\omega$-sharing group $\{x,y,y,y\}$, to denote that $u$
appears once in $\theta(x)$ and three times in $\theta(y)$.  To the
variable $v$, we associate the $\omega$-sharing group $\{x,z\}$, to
denote that $v$ appears once in $\theta(x)$ and once in $\theta(z)$.  Then
we consider the collection of all the multisets so obtained $\{
\{x,y,y,y\}, \{x,z\} \}$, which describes both the sharing property
and the exact amount of non-linearity in the given substitution.  The
domain we obtain is conceptually simple, but cannot be
directly used for static analysis, without a widening operator \cite{CousotC92wide}, since it contains infinite ascending chains.
However, in this domain the role played by (non-)linearity is
manifest, and we can provide a constructive characterization of the optimal abstract unification operator.  The cornerstone of the abstract unification is the
concept of \emph{sharing graph} which plays the same role as
alternating paths \cite{S86,King00jlp} for pair-sharing
analysis.
We use sharing graphs to combine different $\omega$-sharing groups during unification.  The use of sharing
graphs offers a new perspective for looking at variables in the
process of unification, and simplifies the proofs of correctness and
optimality of the abstract operators.

We prove that sharing graphs yield an optimal abstract unification operator for single binding substitutions. We also provide a purely
algebraic characterization of the unification process, which should help in
implementing the domain through widening operators and in
devising abstract operators for further abstractions of $\Linp$.

% In order to prove optimality, we also introduce a parallel unification
% operator, able to compute the abstract unification over $\Linp$ by
% considering all the bindings at the same time.  Usually, both concrete and
% abstract unification are computed by considering one binding at a time. For
% instance, the unification of a substitution $\theta$ with
% $\{x_1/t_1,x_2/t_2,\ldots,x_n/t_n\}$ is performed by first computing
% the unification of $\theta$ with $\{x_1/t_1\}$, and then unifying the
% result with $\{x_2/t_2,\ldots,x_n/t_n\}$. We show that our parallel
% unification operator and the standard (sequential) one do coincide
% over $\Linp$.  Surprisingly, we will show that this is not the case for many domains
% in the literature for sharing analysis, including the reduced product
% $\Sharing \times \Lin$.

\subsubsection{The applications}
We consider two well-known domains for sharing properties,
namely, the reduced product \cite{CousotC79} $\Sharing \times \Lin$ and the more
precise domain $\ShLinp$ by Andy King, and show that they can be
immediately obtained as abstractions of $\Linp$.  By exploiting the
unification operator on $\Linp$, we provide the optimal abstract unification operators, in the case of
single binding substitutions, for both domains. We show that we gain in precision \wrt any previous attempt to design an abstract unification operator on these domains. This is the first time that abstract unification  has been provided optimal for a domain including sharing and linearity information.

Surprisingly, the optimal abstract operators are able to improve not only aliasing and linearity information, but also groundness. We show that, in certain cases, we improve over $\Prop$ \cite{ArmstrongMSS94}. This is mainly due to the fact that our operators exploits the occur-check condition. As far as we know, there is no abstract unification operator in the literature, for a domain dealing with sharing, freeness and linearity, which is more precise than $\Def$ for groundness.

Unification for multi-binding substitutions is usually computed by considering one binding at a time. For instance, the unification of a substitution $\theta$ with $\{x_1/t_1,\linebreak x_2/t_2,\ldots,x_n/t_n\}$ is performed by first computing the unification of $\theta$ with $\{x_1/t_1\}$, and then unifying the result with $\{x_2/t_2,\ldots,x_n/t_n\}$. Actually, computing abstract unification one binding at a time is optimal in $\Linp$ \cite{AmatoS05tr-2}. We show that this is not the case for $\ShLinp$ and $\Sharing \times \Lin$.  This means that the classical schema of computing unification iteratively on the number of bindings cannot be used when looking for optimality with multi-binding substitutions, at least with these two domains.

\subsubsection{Structure of the article}
In Section \ref{sec:notation} we recall some basic notions and the notations about substitutions, multisets and abstract interpretation. In Section \ref{sec:existential} we briefly recall the domain of existential substitutions and its operators, which will be used throughout the article.  In Section \ref{sec:shlinomega} we define the domain $\Linp$, together with the  unification operator, we show the optimality result and give an alternative algebraic characterization of the unification operator. In Section \ref{sec:practical}  we exploit our results to devise the optimal unification operators for $\ShLinp$ and $\Sharing \times \Lin$, in the case of single binding substitutions. Section \ref{sec:practice} gives some evidence that there are practical advantages in using the optimal unification operators for $\ShLinp$ and $\Sharing \times \Lin$. In Section \ref{sec:related} we compare our domains and operators with those known in the literature.  We conclude with some open questions for future work.
The proofs of the main results of the paper are in Appendix \ref{sec:mainproof}, and the proofs of the results in Section \ref{sec:practical} are in Appendix \ref{sec:proofs}.

The paper is a substantial expansion of \cite{AmatoS03lopstr}, which introduces preliminary results of optimality for domains involving sharing and linearity properties.

%are mostly long and tedious case-based analysis. Therefore, we have placed them in Appendix \ref{sec:proofs}.

\section{Notation}
\label{sec:notation}
Given a set $A$, let $\wp(A)$ be the powerset of $A$ and $\fwp(A)$ be
the set of finite subsets of $A$. Given two posets $(A,\leq_A)$ and
$(B,\leq_B)$, we denote by $A \fun B$ the poset of monotonic functions
from $A$ to $B$ ordered pointwise. We use $\leq_{A \ra B}$ to denote
the order relation over $A \fun B$.  When an order for $A$ or $B$ is
not specified, we assume the least informative order ($x \leq y \iff
x=y$). We also use $A \uplus B$ to denote disjoint union and
$\card{A}$ for the cardinality of the set $A$.

\subsection{Terms and substitutions}

In the following, we fix a first order signature and a denumerable set
of variables $\var$. Given a term or other syntactic object $o$, we
denote by $\vars(o)$ the set of variables occurring in $o$ and by
$\occ(v,o)$ the number of occurrences of $v$ in $o$. When it does not
cause ambiguities, we abuse the notation and prefer to use $o$ itself
in the place of $\vars(o)$. For example, if $t$ is a term and $x \in
\var$, then $x \in t$ should be read as $x \in \vars(t)$.

We denote by $\epsilon$ the empty substitution, by $\{x_1/t_1,
\ldots, x_n/t_n\}$ a substitution $\theta$ with $\theta(x_i)=t_i\neq
x_i$, by $\dom(\theta)=\{x\in\var \mid \theta(x)\neq x\}$ and
$\rng(\theta)=\cup_{x\in\dom(\theta)}\vars(\theta(x))$ the domain and
range of $\theta$ respectively. Let $\vars(\theta)$ be the set $\dom(\theta)\cup
\rng(\theta)$ and, given $U\in \fwp(\var)$, let $\theta|_{U}$ be the
projection of $\theta$ over $U$, \ie the unique substitution such that
$\theta|_U(x)=\theta(x)$ if $x\in U$ and $\theta|_U(x)=x$
otherwise.
%% We also write $\theta|_{-U}$ to denote the restriction of
%% $\theta$ over all variables but those in $U$, \ie
%% $\theta|_{-U}=\theta|_{\dom(\theta) \setminus U}$.
Given $\theta_1$ and $\theta_2$ two substitutions with disjoint
domains, we denote by $\theta_1 \uplus \theta_2$ the substitution
$\theta$ such that $\dom(\theta)=\dom(\theta_1) \cup \dom(\theta_2)$
and $\theta(x)=\theta_i(x)$ if $x \in \dom(\theta_i)$, for each $i \in
\{1,2\}$.  The application of a substitution $\theta$ to a term $t$ is
written as $t \theta$ or $\theta(t)$. Given two substitutions $\theta$
and $\delta$, their composition, denoted by $\theta\circ \delta$, is
given by $(\theta \circ \delta)(x)=\theta(\delta(x))$. A substitution $\theta$ is idempotent when $\theta \circ \theta=\theta$ or, equivalently, when $\dom(\theta)\cap\rng(\theta)=\emptyset$.
A substitution
$\rho$ is called renaming if it is a bijection from $\var$ to $\var$
(this is equivalent to saying that there exists a substitution
$\rho^{-1}$ such that $\rho \circ \rho^{-1} = \rho^{-1} \circ \rho =
\epsilon$). Instantiation induces a preorder on substitutions:
$\theta$ is more general than $\delta$, denoted by $\delta
\leq\theta$, if there exists $\sigma$ such that $\sigma \circ \theta =
\delta$. If $\approx$ is the equivalence relation induced by $\leq$,
we say that $\sigma$ and $\theta$ are equal up to renaming when
$\sigma \approx \theta$.
The sets of substitutions, idempotent
substitutions and renamings are denoted by $\subst$, $\Isubst$ and
$\Ren$ respectively. Given a set of equations $E$, we write
$\sigma=\mgu(E)$ to denote that $\sigma$ is a most general unifier of
$E$.
%% Since $\sigma$ is only defined up to renamings, we only use this
%% notation in a context where it is not important the actual unifier
%% which is chosen.
Any idempotent substitution $\sigma$ is a most general unifier of the
corresponding set of equations $\eq(\sigma)=\{ x = \sigma(x) \mid
x\in\dom(\sigma)\}$. In the following, we will abuse the notation and
denote by $\mgu(\sigma_1,\ldots,\sigma_n)$ the
substitution $\mgu(\eq(\sigma_1) \cup \ldots \cup \eq(\sigma_n))$, when it exists.
In spite of a single binding substitution $\{x/t\}$ we often use just the
\emph{binding} $x/t$. In the rest of the paper we assume that a binding $x/t$
is idempotent, namely,that $x \notin \vars(t)$.

A \emph{position} is a sequence of positive natural numbers. We
denote with $\Pos$ the set of all positions and with $\Natzero$  the set
of all positive natural numbers. Given a term $t$ and a position
$\xi$, we define $t(\xi)$ inductively as follows:
\[
\begin{split}
  t(\epsilon) &= t \qquad \text{(where $\epsilon$
    denotes the empty sequence)}\\
  t(i \cdot \xi')&= \begin{cases} t_i(\xi') & \text{if $t \equiv
      f(t_1, \ldots, t_n)$ and $i \leq n$;}\\
    \text{undefined} & \text{otherwise.}
  \end{cases}
\end{split}
\]
For any variable $x$, an \emph{occurrence} of $x$ in
$t$ is a position $\xi$ such that $t(\xi)=x$.

In the rest of the paper, we use: $U$, $V$, $W$ to denote finite sets
of variables; $h,k,u,v,w,x,y,z$ for variables; $t$ for terms; $f, r, s$ for term symbols;
$a, b$ for constants; $\eta,\theta,\sigma,\delta$
for substitutions; $\rho$ for renamings.
%  The same holds for
%  derivatives of these symbols, obtained by adding subscripts,
%  superscripts or both.

\subsection{Multisets}

A \emph{multiset} is a set where repetitions are allowed. We denote by
 $\multil x_1, \ldots, x_m \multir$  a multiset, where $x_1,
\ldots, x_m$ is a sequence with (possible) repetitions.
We denote by $\emptymulti$ the empty multiset.
We will often use the polynomial notation $v_1^{i_1} \ldots
v_n^{i_n}$, where $v_1, \ldots, v_n$ is a sequence without
repetitions, to denote a multiset $A$ whose element $v_j$ appears $i_j$ times.
The set $\{v_j \mid i_j > 0\}$  is called the \emph{support} of $A$ and is denoted by $\supp{A}$.
We also use the functional notation $A: \{v_1, \ldots, v_n\} \fun \Nat$, where $A(v_j)=i_j$.

In this paper, we only consider multisets whose support is \emph{finite}.
We denote with $\mwp(X)$ the set of all
the multisets whose support is \emph{any finite subset} of $X$.
For example, both $a^2c^4$ and $a^1b^2c^3$ are elements of
$\mwp(\{a,b,c\})$. The cardinality of a multiset is $|A|= \sum_{v \in \supp{A}}
A(v)$.

The new fundamental operation for multisets
is the \emph{sum}, defined as
\[
  A \multisum B = \lambda v \in \supp{A} \cup \supp{B}. A(v)+B(v)
  \enspace .
\]
Multiset sum is associative, commutative and $\emptymulti$ is the
neutral element. Note that we also use $\multisum$ to denote disjoint
union for standard sets. The context will allow us to identify the
proper semantics of $\uplus$.

Given a multiset $A$ and $X \subseteq \supp{A}$,
the \emph{restriction} of $A$ over $X$, denoted by
$A|_{ X}$, is the only multiset $B$ such that $\supp{B}=X$ and $B(v)=A(v)$
for each $v \in X$.
Finally, if $A \in \mwp(X)$, $E[x]$ is an integer expression
and $x \in X$, we define
\[
  \sum_{x \in A} E[x] = \sum_{x \in \supp{A}} A(x) \cdot E[x] \enspace .
\]
For example, given a multiset $A=\multil 5,5,6,8,8,8 \multir$ then $\sum_{x \in A} x^2=2*5^2+6^2+3*8^2=278$.
%%  If $X$ is a multiset of multisets, i.e.  $X=\multil X_1,\ldots,
%% X_n \multir$, we also write $\multisum X = \multisum_{i \in I} X_i$ or,
%% with the same meaning, $X_1 \multisum \ldots \multisum X_n$. If a set $S$ is
%% used in a context where a multiset is expected, it stands for the
%% multiset with support $S$ and such that $S(x)=1$ for each $x \in S$.

\subsection{Abstract interpretation}

Given two sets $C$ and $A$ of concrete and abstract objects
respectively, an \emph{abstract interpretation} \cite{CousotC92fr} is
given by an approximation relation $\rightslice \subseteq A \times C$.
When $a \rightslice c$ holds, this means that $a$ is a correct
abstraction of $c$. In particular, we are interested in the case when
$(A,\leq_A)$ is a poset and $a \leq_A a'$ means that $a$ is more
precise than $a'$.  In this case we require that, if $a \rightslice c$
and $a \leq_A a'$, then $a' \rightslice c$, too. In more detail, we
require what \citeN{CousotC92fr} call the \emph{existence of a best
  abstract approximation assumption}, \ie the existence of a map
$\alpha: C \ra A$ such that for all $a \in A, c \in C$, it holds that
$a \rightslice c \iff \alpha(c) \leq_A a$. The map $\alpha$ is called
the \emph{abstraction function} and maps each $c$ to its best
approximation in $A$.

Given a (possibly partial) function $f: C \ra C$, we say that
$\tilde{f}: A \ra A$ is a correct abstraction of $f$, and write
$\tilde{f} \rightslice f$, whenever
\[
a \rightslice c  \Rightarrow \tilde{f}(a)
  \rightslice f(c) \enspace ,  
\]
assuming that $\tilde{f}(a) \rightslice f(c)$ is true whenever $f(c)$
is not defined.  We say that $\tilde{f}: A \ra A$ is the
\emph{optimal} abstraction of $f$ when it is the best correct
approximation of $f$, \ie when $\tilde{f} \rightslice f$ and
\[
\forall f':A \ra A.\  f' \rightslice f
  \Rightarrow \tilde{f} \leq_{A \ra A} f' \enspace .  
\]
In some cases, we prefer to deal with a stronger framework, in which
the domain $C$ is also endowed with a partial order $\leq_C$ and
$\alpha: C \ra A$ is a left adjoint to $\gamma: A \ra C$, \ie
\[
  \forall c \in C. \forall a \in A. \alpha(c) \leq_A a \iff c \leq_C
  \gamma(a) \enspace .
\]
The pair $\langle \alpha, \gamma \rangle$ is called a \emph{Galois
  connection}.  In particular, we will only consider the case of
\emph{Galois insertions}, which are Galois connections such that
$\alpha \circ \gamma$ is the identity map. If $\langle \alpha, \gamma
\rangle$ is a Galois insertion and $f: C \ra C$ is a monotone map, the
optimal abstraction $\tilde{f}$ always exists and it is definable as
$\tilde{f}= \alpha \circ f \circ \gamma$.

\section{The domain of existential substitutions}
\label{sec:existential}

The choice of the concrete domain depends on the observable properties
we want to analyze.  Most of the semantics suited for the analysis of logic
programs are based on computed answer substitutions, and most
of the domains are expressed as abstractions of sets of substitutions.
In general, we are not really interested in the substitutions, but in
their quotient-set \wrt an appropriate equivalence relation. Let us
consider a one-clause program $\mathtt{p(x,x)}$, the goal $p(x,y)$,
and the following answer substitutions: $\theta_1=\{y/x\}$,
$\theta_2=\{x/y\}$, $\theta_3=\{x/u, y/u\}$ and
$\theta_4=\{x/v,y/v\}$. Although $\theta_1$ and $\theta_2$ are equal
up to renaming, the same does not hold for $\theta_3$ and $\theta_4$.
Nonetheless, they essentially represent the same answer, since $u$ and
$v$ are just two different variables we chose when renaming apart the
clause $\mathtt{p(x,x)}$ from the goal $p(x,y)$, and therefore are not
relevant to the user.  On the other hand, if $\theta_3$ and $\theta_4$
are answer substitutions for the goal $q(x,y,u)$, then they correspond
to computed answers $q(u,u,u)$ and $q(v,v,u)$ and therefore are
fundamentally different.  As a consequence, the equivalence relation
we need to consider must be coarser then renaming, and must take into
account the set of variables of interest, \ie the set of variables
which appear in the goal.  For these reasons, we think that the best
solution is to use a domain of equivalence classes of substitutions.
Among the various domains proposed in the literature
(\eg \citeNP{JacobsL92,MSJ94,LeviS03}), we adopt the domain of existential
substitutions \cite{AmatoS09sharing}, since it is explicitly defined as
a quotient of a set of substitutions, \wrt a suitable equivalence relation.
Moreover, the domain is equipped with all the necessary operators for
defining a denotational semantics, namely, projection, renaming and
unification.  We briefly recall the basic definitions of the domain
and the unification operator.

Given $\theta_1, \theta_2 \in \subst$ and $U \in \fwp(\var)$, the
preorder $\preceq_U$ is defined as follows:
\[
\theta_1 \preceq_U \theta_2 \iff \exists  \delta \in \subst.
  \forall x \in U.\ \theta_1(x)=\delta(\theta_2(x)) \enspace.  
\]

The notation $\theta_1 \preceq_U \theta_2$ states that $\theta_1$ is an
instance of $\theta_2$ \wrt the variables in
$U$. The equivalence relation induced by the preorder $\preceq_U$
is given by:
\[
  \theta_1 \sim_U \theta_2 \iff \exists \rho \in \Ren.
  \forall x \in U.\ \theta_1(x)=\rho(\theta_2(x)) \enspace .
\]
This relation precisely captures the extended notion of renaming which
is needed to work with computed answer substitutions.

\begin{example}
  It is easy to check that $\{x/w, y/u\} \sim_{\{x,y\}} \epsilon$ by
  choosing the renaming $\rho=\{x/w, \allowbreak w/x, y/u, u/y\}$. Note that $\sim_U$ is
  coarser than the standard equivalence relation $\approx$: there is
  no renaming $\rho$ such that $\epsilon = \rho \circ \{x/w, y/u\}$.
  As it happens for $\preceq$, if we enlarge the set of variables of
  interest, not all equivalences between substitutions are preserved:
  for instance, $\{x/w, y/u\} \not\sim_{\{w,x,y\}} \epsilon$.
  \exproofbox
\end{example}

Let $\Isubst_{\sim_U}$ be the quotient set of $\Isubst$ \wrt
$\sim_U$. The domain $\Isubst_\sim$ of \emph{existential
  substitutions} is defined as the disjoint union of all the
$\Isubst_{\sim_U}$ for $U\in \fwp(\var)$, namely:
\[
\Isubst_\sim = \biguplus_{U\in \fwp(\var)} \Isubst_{\sim_U} \enspace .
\]
In the following we write $[\theta]_{U}$ for the equivalence class of
$\theta$ \wrt $\sim_U$.  The partial order $\preceq$ over
$\Isubst_\sim$ is given by:
\[
[\theta]_U \preceq [\theta']_V \iff U \supseteq V \wedge \theta
\preceq_V \theta' \enspace .
\]
Intuitively, $[\theta]_U \preceq [\theta']_V$ means that $\theta$ is
an instance of $\theta'$ \wrt the variables in $V$, provided that
they are all variables of interest of $\theta$.

To ease notation, we often omit braces from the sets of variables of
interest when they are given extensionally. So we write
$[\theta]_{x,y}$ instead of $[\theta]_{\{x,y\}}$ and $\sim_{x,y,z}$
instead of $\sim_{\{x,y,z\}}$.  When the set of variables of interest
is clear from the context or when it is not relevant, it will be omitted.
Finally, we omit the braces which enclose the bindings of a
substitution when the latter occurs inside an equivalence class, \ie we write
$[x/y]_U$ instead of $[\{x/y\}]_U$.

\subsubsection{Unification}

Given $U,V \in \fwp(\var)$, $[\theta_1]_U, [\theta_2]_V \in
\Isubst_{\sim}$, the most general unifier between these two classes is
defined as the mgu of suitably chosen representatives, where variables
not of interest are renamed apart. In formulas:
\begin{equation}
  \label{eq:mgu}
  \mgu([\theta_1]_U,[\theta_2]_V)=
  [\mgu(\theta'_1, \theta'_2)]_{U \cup V} \enspace,
\end{equation}
where $\theta_1 \sim_U \theta'_1 \in \Isubst$, $\theta_2 \sim_V
\theta'_2 \in \Isubst$ and $(U \cup \vars(\theta'_1)) \cap (V \cup
\vars(\theta'_2)) \subseteq U \cap V$. The last condition is needed to
avoid variable clashes between the chosen representatives $\theta'_1$
and $\theta'_2$. Moreover, $\mgu$ is the greatest lower bound of
$\Isubst_{\sim}$ ordered by $\preceq$.

\begin{example}
  Let $\theta_1=\{ x/a, y/r(v_1,v_1,v_2)\}$ and $\theta_2=\{
  y/r(a,v_2,v_1), z/b \}$. Then
  \[
  \mgu([\theta_1]_{x,y},[\theta_2]_{y,z})=[x/a,y/r(a,a,v),z/b]_{x,y,z} \enspace,
  \]
  by choosing $\theta'_1=\theta_1$ and $\theta'_2=\{y/r(a,w,v),z/b\}$.
  In this case we have
  \begin{multline*}
      \{x/a,y/r(a,a,v),z/b\} \sim_{x,y,z} \\
      \mgu(\theta'_1,\theta'_2)=\{ x/a, y/r(a,a,v), z/b, v_1/a, w/a, v_2/v
      \} \enspace . \mathproofbox
  \end{multline*}
\end{example}
A different version of unification is obtained when one of the two
arguments is an existential substitution, and the other one is a
standard substitution. In this case, the latter argument may be viewed
as an existential substitution where all the variables are of
interest:
\begin{equation}
  \label{eq:mixmgu}
  \mgu([\theta]_U,\delta)=
  \mgu([\theta]_U,[\delta]_{\vars(\delta)}) \enspace .
\end{equation}
Note that deriving the general unification in \eqref{eq:mgu} from the
special case in \eqref{eq:mixmgu} is not possible. This is because
there are elements in $\Isubst_\sim$ which cannot be obtained as
$[\delta]_{\vars(\delta)}$ for any $\delta \in \Isubst$ (see
Example \ref{ex:sim}).

This is the form of unification which is better suited for analysis of logic programs, where existential substitutions are the denotations of programs while standard substitutions are the result of unification between goals and heads of clauses. Therefore, the rest of the paper will be concerned with the problem of devising optimal abstract operators corresponding to \eqref{eq:mixmgu}, for three different abstract domains. Of course, unification is not the only operator needed to give semantics to logic programs: we also need projection, renaming and union. However, providing optimal abstract counterparts for these operators is generally a trivial task, and will not be considered here.

We want to conclude the section with a small remark about our choice
of the concrete domain. By adopting existential substitutions and the
corresponding notion of unification, we greatly simplify all the semantic definitions which are heavily based on renaming variables apart. This is because all the details concerning
renamings are shifted towards the inner level of the semantic domain, where
they are more easily managed \cite{JacobsL92,AmatoS09sharing}.

\section{The abstract domain $\Linp$}
\label{sec:shlinomega}

The domain $\Sharing \times \Lin$ is one of the best known domains in the literature which combine sharing and linearity information.
The domain $\Sharing$ records the information of variable aliasing, by abstracting the substitution
$\theta=\{x/f(u,v),\linebreak y/g(u,u,u), z/v\}$ into the set $\{uxy,vxz\}$. The object $uxy$, called a \emph{sharing group}, states that $\theta(u),\theta(x)$ and $\theta(y)$ do share some variable (the variable $u$ in this case). Analogously, the sharing group $vxz$ states that $\theta(v),\theta(x)$ and $\theta(z)$ do share (in this case the variable $v$).
One of the simplest way of adding linearity information is to record, in a separate object, the set of variables $w$ such that $\theta(w)$ is a linear term. In our example, only $\theta(y)$ is not linear. Thus the substitution is abstracted into the pair $(\{uxy,vxz\},\{u,v,x,z\})$.
Another known domain in the literature is $\ASub$ whose main difference \wrt $\Sharing \times \Lin$ is that it only records sharing information between pairs of variables. Thus, in $\ASub$, each sharing group has at most two elements. Developing optimal unification operators for such abstract domains is
a difficult  task. In our opinion, this is because the gap between the substitutions and $\Sharing \times
\Lin$ (or $\ASub$) is too wide and the combined effect of aliasing and
linearity is difficult to grasp.

We solve this problem by defining a new abstract domain $\Linp$ which can be
used to approximate $\Isubst_\sim$.  Since $\Linp$ has infinite ascending chains, in
most cases it cannot be directly used for the analysis. It
should be thought of as a general framework from which other domains
can be easily derived by abstraction.
In this sense, $\Linp$ closes the gap between the concrete domain of
substitutions and the abstractions like $\Sharing \times \Lin$ or
$\ASub$. The structure of $\Linp$ has made it possible to develop clean and optimal
abstract unification operators. From these, optimal operators for the simpler
domains are easy to obtain, at least for single binding substitutions.

The idea underlying $\Linp$ is to count the exact number of
occurrences of the same variable in a term. It extends
the standard domain $\Sharing$ by recording, for each $v \in \mathcal
V$ and $\theta \in \Isubst$, not only the set $\{w \in \var \mid v \in
\theta(w)\}$ but the multiset $\lambda w \in \var. \occ(v,\theta(w))$.

\begin{definition}[$\omega$-Sharing Group]
  An \emph{$\omega$-sharing group} is a multiset of variables, \ie
  an element of $\mwp(\var)$.
\end{definition}

\begin{example}
Given variables $u,v,w,x,y\in \var$, examples of $\omega$-sharing groups are $u^2v^3x^{19}$, $xyz$ and $u^{23}vwx^2y^3$.
\exproofbox
\end{example}

% Given a substitution $\theta$ and a variable $v \in \var$, we denote
% by $\theta^{-1}(v)$ the $\omega$-sharing group $\lambda w \in \var.
% \occ(v,\theta(w))$, which

\begin{definition}
\label{eq:theta-1}
  Given a substitution $\theta$ and a variable $v \in \var$, we define
  \[
    \theta^{-1}(v)=\lambda w.\  \occ(v,\theta(w)) \enspace .
  \]
\end{definition}
Intuitively, $\theta^{-1}(v)$ is an $\omega$-sharing group which maps each variable $w$ to the number of
occurrences of $v$ in $\theta(w)$.

\begin{example}
Given $\theta=\{x/f(u,u,u), y/g(u,v), z/f(u,v,v)\}$, we have that $\theta^{-1}(u)=ux^3yz$, $\theta^{-1}(v)=vyz^2$,  $\theta^{-1}(w) = w$, and $\theta^{-1}(x)=\emptymulti$.
\exproofbox
\end{example}

\begin{definition}[Correct Approximation]
  Given a set of variables $U$ and a set of $\omega$-sharing groups $S$ (\ie $S \subseteq \mwp(U)$), we say that the pair $(S,U)$ correctly
  approximates a substitution $[\theta]_W$ if $U=W$ and for each $v
  \in \var$, $\theta^{-1}(v)|_W \in S$. In the following we denote by $[S]_U$ the pair $(S,U)$ and write $[S]_U \rightslice
  [\theta]_W$ to mean that $[S]_U$ correctly approximates $[\theta]_W$.
  
\end{definition}
Therefore, $[S]_U$ correctly approximates $[\theta]_U$ when $S$
contains at least all the $\omega$-sharing groups which may arise in
$\theta$, restricted to the variables $U$. Note that $[\theta]_U$ is an equivalence class of substitutions, as defined in Section  \ref{sec:existential}, while $[S]_U$ is just a symbol to denote the pair of objects $(S,U)$. We prefer this notation for the sake of uniformity with substitutions.

\begin{theorem}
  \label{th:wdapprox}
  The relation $\rightslice$ is well defined.
\end{theorem}

We can now define the domain $\Linp$ of $\omega$-sharing groups.
\begin{definition}[$\Linp$]
The domain $\Linp$ is defined as
\[
   \Linp =\{ [S]_U \mid U \in \fwp(\var), S \subseteq \mwp(U),
  S \neq \emptyset \Rightarrow \emptymulti \in S\} \enspace ,
\]
and ordered by $[S_1]_{U_1} \leq_\omega [S_2]_{U_2}$ iff $U_1=U_2$ and
$S_1 \subseteq S_2$.
\end{definition}
The order relation corresponds to the
approximation ordering, since bigger (w.r.t $\leq_\omega$) elements
correctly approximate a larger number of substitutions than smaller
elements.  The existence of the empty multiset, when $S$ is not empty, is required in order to obtain a Galois insertion, instead of a Galois connection.
In order to simplify the notation, in the following we
write an object $[\{\emptymulti, B_1,\ldots, B_n\}]_U \in \Linp$ as $[B_1,\ldots,
B_n]_U$ by omitting the braces and the empty multiset.  Moreover, if $X \in \Linp$, we write
$B \in X$ in place of $X=[S]_U \wedge B \in S$.

\begin{definition}[Abstraction for $\Linp$]
We define the abstraction for a substitution $[\theta]_U$ as
\begin{equation*}
  \alpha_\omega([\theta]_U)=[\{ \theta^{-1}(v)|_U \mid v \in \var \}]_U
  \enspace .
\end{equation*}
\end{definition}
This is the least element of $\Linp$ which correctly approximates
$[\theta]_U$. Note that by the proof of Theorem~\ref{th:wdapprox}
it immediately follows that $\alpha_\omega$ is well defined, \ie  it does
not depend from the choice of the representative for $[\theta]_U$.

\begin{example}
  Given $\theta=\{x/r(y,u,u), z/y, v/u\}$ and $U=\{w,x,y,z\}$, we have
  $\theta^{-1}(u)=x^2vu$, $\theta^{-1}(y)= xyz$,
  $\theta^{-1}(z)=\theta^{-1}(v)=\theta^{-1}(x)=\emptymulti$ and
  $\theta^{-1}(s)=s$ for all the other variables (included $w$).
  Projecting over $U$ we obtain
  $\alpha_\omega([\theta]_U)=[x^2,xyz,w]_U$.
  \exproofbox
\end{example}
\begin{example}
  \label{ex:sim}
  As we have said in Section \ref{sec:existential}, we show an element of $\Isubst_\sim$, namely the
  existential substitution $[x/r(v,v)]_{x}$, which cannot
  be obtained as $[\delta]_{\vars(\delta)}$ for any substitution
  $\delta$. In fact, consider any $\omega$-sharing group  $B=\delta^{-1}(u)|_{\vars(\delta)} \in \alpha_{\omega}([\delta]_{\vars(\delta)})$. Then
  either $u \notin \rng(\delta)$ and $B=\multil \multir$ or $u \in
  \rng(\delta)$ and $B(u)=1$. However, $\alpha([x/r(v,v)]_{x})= [x^2]_{x}$ and $x^2$ does
  not contain any variable with multiplicity one. \exproofbox
\end{example}

% \begin{example}
%   \label{ex:sim2}
%   As promised in Section \ref{sec:existential}, we show an element of $\Isubst_\sim$, namely, the
%   existential substitution $[\epsilon]_{x}$, which cannot
%   be obtained as $[\delta]_{\vars(\delta)}$ for any
%   $\delta\in \Isubst$.  Assume, by contradiction, that there exists a substitution $\delta$ such that
% $[\epsilon]_{x}=[\delta]_{\vars(\delta)}$. Then $\vars(\delta)={x}$. By definition, $\vars(\delta)=\dom(\delta)\cup\rng(\delta)$. Thus $x\in \dom(\delta)$ and $\rng(\delta)=\emptyset$, namely, $\delta(x)$ is ground. Therefore, there is no renaming $\rho$ such that $\epsilon(x)=\rho(\delta(x))$. As a conclusion, $[\epsilon]_{x} \neq [\delta]_{\vars(\delta)}$.
%   \proofbox
% \end{example}
% 
% 
%% Therefore, we can lift $\alpha_U$ to obtain the
%% Galois insertion between $\Psub$ and $\Linp$ as follows:
%% $\alpha_\lp(\bot_\rs)=\bot_\lp$, $\alpha_\lp(\top_\rs)=\top_\lp$ and
%% \begin{equation}
%%   \alpha_\lp([\Sigma,U])=\left[\bigcup \{ \alpha_U(\theta) \mid
%%     \theta \in \Sigma \},U\right]
%% \end{equation}
%% %
%% In the following, we will omit to explicitly define abstraction and
%% concretization on the top and bottom elements of the domains.

%% --- Spostare la projection sulla semantica

%% The projection operation is defined pointwise in the obvious way:
%% \begin{equation}
%%   \pi_\lp([S_1,U_1],[S_2,U_2])=[ \{ B|_{U_2} \mid B
%%   \in S_1 \}, U_1 \cap U_2 ]
%% \end{equation}
%% %

\subsection{Multigraphs}
In order to define an abstract unification operator over $\Linp$, we
need to introduce the concept of multigraph. We call (directed) \emph{multigraph} a graph where
multiple distinguished edges are allowed between nodes.  We use the
definition of multigraph which is customary in category theory
\cite{MacLane88}.

\begin{definition}[Multigraph]
  A \emph{multigraph} $G$ is a tuple $\langle N_G, E_G, \src_G, \tgt_G
  \rangle$ where $N_G \neq \emptyset$ and $E_G$ are the sets of
  \emph{nodes} and \emph{edges} respectively, $\src_G: E_G \fun N_G$
  is the \emph{source} function which maps each edge to its start
  node, and $\tgt_G: E_G \fun N_G$ is the \emph{target} function which
  maps each edge to its end node.

  A \emph{labeled} multigraph $G$ is a multigraph equipped with a
  \emph{labelling} function $l_G: N_G \ra L_G$ which maps each node to
  its \emph{label} in the given set $L_G$.
\end{definition}

We write $e: n_1 \ra n_2 \in G$ to denote the edge $e \in E_G$ such
that $\src_G(e)=n_1$ and $\tgt_G(e)=n_2$. We also write $n_1 \ra n_2
\in G$ to denote any edge $e \in E_G$ such that $\src_G(e)=n_1$ and
$\tgt_G(e)=n_2$.  Moreover, with $| n_1 \ra n_2 \in G |$ we denote the
cardinality of the set $\{ e \in E_G \mid \src_G(e)=n_1 \wedge
\tgt_G(e)=n_2 \}$. In the notation above, we omit ``$\in G$'' whenever
the multigraph $G$ is clear from the context.

We call \emph{in-degree} (respectively \emph{out-degree}) of a node
$n$ the cardinality of the set $\{ e \in E_G \mid \tgt(e)=n \}$
(respectively $\{ e \in E_G \mid \src(e)=n \}$).

%% Given a multigraph $G$, a \emph{path} $\pi$ is a \emph{non-empty}
%% sequence of edges $e_1 \ldots e_n$ such that
%% $\tgt_G(e_i)=\src_G(e_{i+1})$ for each $i \in [1,n-1]$. We denote with
%% $\src_G(\pi)$ and $\tgt_G(\pi)$ the \emph{source} and \emph{target}
%% node of the path $\pi$, i.e. $\src(\pi)=\src(e_1)$ and
%% $\tgt(\pi)=\tgt(e_n)$. A multigraph is \emph{connected} if there
%% exists a path between each two distinct nodes.

Given a multigraph $G$, a \emph{path} $\pi$ is a \emph{non-empty}
sequence of nodes $n_1 \ldots n_k$ such that, for each $i \in
\{1,\ldots,k-1\}$, there is either an edge $n_i \ra n_{i+1} \in G$ or
an edge $n_{i+1} \ra n_{i} \in G$. Nodes $n_1$ and $n_k$ are the
\emph{endpoints} of $\pi$, and we say that $\pi$ \emph{connects} $n_1$
and $n_k$. A multigraph is \emph{connected} when all pairs of nodes are
connected by at least one path.

\subsection{Abstract unification}

We need to find the abstract counterpart of $\mgu$ over $\Linp$, \ie
an operation $\mgu_{\omega}$ such that, if $[S]_U \rightslice
[\theta]_U$, then
\begin{equation}
  \label{eq:prop-absmgu}
   \mgu_\omega([S]_U,\delta) \rightslice \mgu([\theta]_U,\delta)
\end{equation}
for each $\delta \in \Isubst$. Note that we are looking for an abstract
counterpart to the mixed unification in \eqref{eq:mixmgu}, where one
of the two arguments is a plain substitution. In particular, we would like to find an operator which is the minimum
element that satisfies the condition in \eqref{eq:prop-absmgu}, \ie the
\emph{optimal} abstract counterpart of $\mgu$.  Observe that, for a fixed
$U$, the set of all the elements $[S]_U\in \Linp$ is a complete lattice w.r.t.~$\leq_\omega$ with the top element
given by $[\mwp(U)]_U$ and the meet operator given by
\[
  \textstyle \bigwedge_\omega \{ [S_i]_U \mid i \in I \} =
  \big[\bigcap_{i \in I} S_i\big]_U \enspace ,
\]
for any family $\{[S_i]_U \mid i \in I\}$ of elements of $\Linp$.
Moreover, the relation $\rightslice$ is meet-preserving on the left,
since if $[S_i]_U \rightslice [\theta]_U$ for each $i \in I$, then
$\bigwedge_\omega \{ [S_i]_U \mid i \in I\} \rightslice [\theta]_U$.
Therefore, we may define the abstract $\mgu$ as follows
\begin{equation*}
  \mgu_\omega([S]_U,\delta)=
  {\bigwedge}_\omega \big\{[S']_{U'} \mid \forall \theta.
  [S]_U \rightslice [\theta]_U \Rightarrow [S']_{U'} \rightslice
  \mgu([\theta]_U, \delta)\big\} \enspace ,
\end{equation*}
where the definitions of $\rightslice$ and $\mgu$ force $U'$ to be $U \cup \vars(\delta)$. Note that this is just a translation of the general definition of
optimal operator in \cite{CousotC92fr} and it satisfies \eqref{eq:prop-absmgu}.

This definition is completely non-constructive. The rest of this section
is devoted to providing a constructive characterization for
$\mgu_\omega([S]_U,\delta)$. We begin to characterize the
operation of abstract unification by means of graph theoretic notions.

\begin{definition}[Multiplicity of $\omega$-sharing groups]
The \emph{multiplicity} of an $\omega$-sharing
group $B$ in a term $t$ is defined as:
\[
  \chi(B,t)=\sum_{v \in B} \occ(v,t) =
  \sum_{v \in \supp{B}} B(v) \cdot \occ(v,t)  \enspace .
\]
\end{definition}
For instance,
$\chi(x^3yz^4,r(x,y,f(x,y,z)))=3 \cdot 2+1\cdot 2+ 4\cdot 1 = 12$. The
meaning of the map $\chi$ is made clear by the following proposition.

\begin{proposition}
  \label{prop:chi}
  Given a substitution $\theta$, a
  variable $v$ and a term $t$, we have that
  $\chi(\theta^{-1}(v),t)=\occ(v,\theta(t))$. Moreover, given a set of variables $U$, when $\vars(t) \subseteq U$, it holds that $\chi(\theta^{-1}(v)|_U,t)=\occ(v,\theta(t))$.
\end{proposition}

\begin{example}
  Let $B=xy^2z^3$ and $\theta=\{ y/r(x,x), z/r(x,x,x) \}$,
  so that $\theta^{-1}(x)=\{xy^2z^3\}$. Given $t \equiv s(x,z)$ we have
  \[
  \occ(x,\theta(t))=\occ(x,s(x,r(x,x,x)))=4 \enspace ,
  \]
  and
  \[
  \chi(B,t)=B(x) \occ(x,t) + B(z) \occ(z,t)=1 \cdot 1 + 3 \cdot 1=4
  \enspace .   \mathproofbox
  \]
\end{example}

If $[S]_U \rightslice [\theta]_U$ and we unify $[\theta]_U$ with
$\delta$, some of the $\omega$-sharing groups in $S$ may be glued
together to obtain a bigger resultant group.
It happens that the gluing of the sharing groups during the
unification of $[\theta]_U$ with a single binding substitution
$\{x/t\}$ may be represented by special labelled multigraphs which we call \emph{sharing graphs}.

%% \begin{definition}[Sharing Graph]
%%   A \emph{sharing graph} is a labeled multigraph whose nodes are
%%   labeled with sharing groups.
%% \end{definition}

\begin{example}
  Let $S=\{x^3,y\}$ and $U=\{x,y\}$. We look for a representation of the unification process between any substitution $\theta$ approximated by $S$ and the binding $x/r(y)$. We show that multigraphs can be easily used for this purpose.
For instance, the substitution $\theta=\{x/r(g(u,u,u))\}$ is approximated by $S$. By unifying $\theta$ with $\{x/r(y)\}$ we obtain $\delta=\{x/r(g(u,u,u)),y/g(u,u,u)\}$. Note that any approximation of $\delta$ on the variables $\{x,y\}$ must include the sharing group $x^3y^3$ generated by the variable $u$. Thus, any correct approximation of the unification must also contain $x^3y^3$.

We want to associate to any $\omega$-sharing group $B$ in $\delta$ a special multigraph which represents the way the $\omega$-sharing groups in $S$ have been merged in order to obtain $B$.
The nodes of this multigraph are the $\omega$-sharing groups in $S$ (possibly repeated any number of times). The following  is a  sharing graph for $x/r(y)$ and $S$:
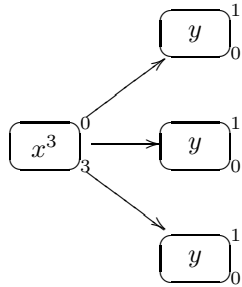
\[
\xybox{
  0 *{\ovalee{x^3}_3^0}="a",
  "a"+<\dist,0cm>*{\ovalee{y}_0^1}="b",
  "b"+<0cm,\dista>*{\ovalee{y}_0^1}="c",
  "b"+<0cm,-\dista>*{\ovalee{y}_0^1}="d",
  {"a" \ar "b"},
  {"a" \ar "c"},
  {"a" \ar "d"}
}
\]
where pedices and apices on a sharing group $B$ are respectively the
values of $\chi(B,x)$ and $\chi(B,r(y))$. For instance, since $\chi(x^3,x)=3$, then we put the pedice $3$ on the node $x^3$ to mean that $x$ is bound to a term containing $3$ occurrences of the same variable.
Symmetrically, since $\chi(x^3,r(y))=0$, then we put the apice $0$ on the node $x^3$.
The in-degree and the out-degree of the nodes reflect the values of apices and pedices.  In this case, we have $3$ out-going edges from $x^3$ and no in-going edges.
Moreover, the multigraph must be connected, in order to guarantee that we can use a single variable to form the sharing group $x^3y^3$.

By summing the labels of all the nodes, namely, $x^3\multisum y \multisum y \multisum y$, we obtain the $\omega$-sharing group $x^3y^3$ which must appear in any correct approximation of the unification.  \exproofbox
\end{example}

Given any labelled multigraph $G$, in the rest of the paper we assume
that the codomain of the labelling function $l_G$ is $\mwp(\var)$, the
set of $\omega$-sharing groups.

\begin{definition}[Sharing Graph]
  A \emph{sharing graph} for the binding $x/t$ and a set of
  $\omega$-sharing groups $S$ is a labelled multigraph $G$ such that
  \begin{enumerate}
  \item $G$ is connected;
  \item for each node $n \in N_G$, $l_G(n) \in S$;
  \item for each node $n \in N_G$, the out-degree of $n$ is equal to
    $\chi(l_G(n),x)$ and the in-degree of $n$ is equal to
    $\chi(l_G(n),t)$.
  \end{enumerate}
The \emph{resultant $\omega$-sharing group} of $G$ is
\[
 \res(G)=\biguplus_{n \in N_G} l_G(n) \enspace .
\]
\end{definition}

\begin{example}\label{ex:shgraph}
  Let $S=\{ux^2,xy,vz,wz,xyz\}$. The following
  is a  sharing graph for $x/r(y,z)$ and $S$:
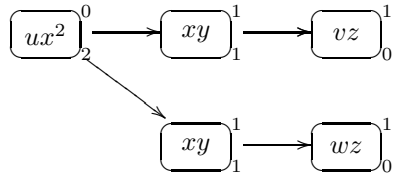
\[
\xybox{
  0 *{\ovalee{ux^2}_2^0}="a",
  "a"+<\dist,0cm>*{\ovalee{xy}_1^1}="b",
  "b"+<\dist,0cm>*{\ovalee{vz}_0^1}="c",
  "b"+<0cm,-\dista>*{\ovalee{xy}_1^1}="d",
  "b"+<\dist,-\dista>*{\ovalee{wz}_0^1}="e",
  {"a" \ar "b"},
  {"a" \ar "d"},
  {"b" \ar "c"},
  {"b" \ar "c"},
  {"d" \ar "e"}
}
\]
where pedices and apices on a sharing group $B$ are respectively the
value of $\chi(B,x)$ and $\chi(B,r(y,z))$.  Therefore the resultant sharing group is $uvwx^4 y^2
z^2$.   \exproofbox
\end{example}

It is worth noting that, given any set of $\omega$-sharing groups $S$ and binding $x/t$, there exist many different sharing graphs for $x/t$ and $S$. Each sharing graph yields a resultant sharing group which must be included in the result of the abstract unification operator. Of course, different sharing graphs may give the same resultant sharing group. The abstract unification operator is defined by collecting all the resultant sharing groups.

\begin{definition}[Single binding unification]
\label{def:pre-mguomega}
Let $U \in \fwp(\var)$, $S$ be a set of $\omega$-sharing groups with $[S]_U \in \Linp$, $x/t$ be a binding,
and $\vars(x/t) \subseteq U$. The set of resultant $\omega$-sharing groups for $x/t$ and $S$ is
\[
  \mgu_\omega(S,x/t)=\{ \res(G) \mid G
  \text{ is a sharing graph for $S$ and $x/t$} \} \enspace .
\]
We lift $\mgu_\omega$ to an operation over $\Linp$.
\[
  \mgu_\omega([S]_U,x/t)=[\mgu(S,x/t)]_U \enspace .
\]
\end{definition}
This is a particular case of the abstract unification
operator, for single binding substitutions and $\vars(x/t) \subseteq U$.

\begin{example}
  \label{ex:shgraphsimple}
  Let $S$ be as in Example \ref{ex:shgraph}. The
  following is a  sharing graph for $x/r(y,y,z)$ and $S$:
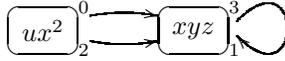
\[
\xybox{
  0 *{\ovalee{ux^2}_2^0}="a",
  "a"+<\dist,0cm>*{\ovalee{xyz}_1^3}="b",
  "b"+<0.5cm,0.1cm>*{}="b1",
  "b"+<0.5cm,-0.1cm>*{}="b2",
  {"a" \ar @/^/  "b"},
  {"a" \ar @/_/  "b"},
  {"b1" \ar @(ur,dr) "b2"},
}
\]
where pedices and apices on a sharing group $B$ are respectively the
value of $\chi(B,x)$ and $\chi(B,r(y,y,z))$.  Therefore $ux^3yz \in
\mgu_\omega(S,x/r(y,y,z))$. Note that this sharing group can
actually be generated by the substitution $\theta=\{x/r(v_1,v_1,v_2),
y/v_2, z/v_2, \linebreak u/v_1, v/a, w/a\}$ where $a$ is a ground term. Let $U=\{u,v,w,x,y,z\}$. It is the
case that $[S]_U \rightslice [\theta]_U$ and
$\mgu([\theta]_U,\{x/r(y,y,z)\})$ performs exactly the variable
aliasings depicted by the sharing graph. Actually $\mgu([\theta]_U,
\{x/r(y,y,z)\}) =[x/r(v_1,v_1,v_1),\linebreak y/v_1,u/v_1,v/a,w/a]_U=[\eta]_U$
and $\eta^{-1}(v_1)|_U = ux^3yz$.   \exproofbox
\end{example}
%

%% A sharing graph represents a possible way to merge together several
%% sharing groups by unifying them with a given binding.

We give here an intuition of the way sharing graphs work.
Assume given a set of $\omega$-sharing groups $[S]_U$ and a binding $x/t$ with
$\vars(x/t) \subseteq U$. We want to compute $[\mgu_\lp(S,x/t)]_U$. To this aim, for any substitution $\theta$ approximated by $[S]_U$, that is, $[S]_U \rightslice [\theta]_U$, we compute $\alpha_\omega(\mgu([\theta]_U,\{x/t\}))$.

For any $B_1,B_2 \in S$, assume that there exist $v_1,v_2 \in \mathcal V$ such that
$B_1=\theta^{-1}(v_1)|_U$ and $B_2=\theta^{-1}(v_2)|_U$.
When unifying $\theta$ with the binding
$x/t$, we use the fact that $\mgu(\eq(\theta) \cup
\{x=t\})=\mgu(\{\theta(x)=\theta(t)\}) \circ \theta$.
By Prop. \ref{prop:chi}, $\theta(x)$ contains $\chi(B_1,x)$
instances of $v_1$ and $\chi(B_2,x)$
instances of $v_2$. Symmetrically,
$\theta(t)$ contains $\chi(B_1,t)$ instances of $v_1$ and $\chi(B_2,t)$ instances of $v_2$.

Assume that $\theta(x)$ and $\theta(t)$ only
differ for the variables occurring in them (and not for the structure of terms).  Then, an arrow from the
sharing group $B_1$ to $B_2$ represents the fact that, in
$\mgu(\{\theta(x)=\theta(t)\})$, one of the copies of $v_1$ is aliased
to one of the copies of $v_2$, \ie that there are corresponding
positions in $\theta(x)$ and $\theta(t)$ where the two terms contain
the variables $v_1$ and $v_2$ respectively. The third condition for
sharing graphs implies that each occurrence of $v_1$ and $v_2$ is aliased to
some other variable.  The first condition (the sharing graph must be connected) ensures that all the variables corresponding to the $\omega$-sharing groups involved in the sharing graph are aliased to each other. In other words, given any two such variables, they are aliased.
Although here we are only considering the case
when $\theta(x)$ and $\theta(t)$ differ for the variables occurring in
them, we will show that it is enough to reach correctness and
optimality. The next example applies this intuition to a concrete case.
\begin{example}
Consider Example \ref{ex:shgraphsimple}, where $\theta=\{x/r(v_1,v_1,v_2),y/v_2,z/v_2,u/v_1,v/a,w/a\}$ and  $U=\{u,v,w,x,y,z\}$.
Let $B_1=ux^2$ and $B_2=xyz$, thus $B_1=\theta^{-1}(v_1)|_U$ and $B_2=\theta^{-1}(v_2)|_U$. When unifying $\theta$ with the binding $x/r(y,y,z)$ we have that $\theta(x)=r(v_1,v_1,v_2)$ and $\theta(r(y,y,z))=r(v_2,v_2,v_2)$.

Note that $\theta(x)$ contains $\chi(ux^2,x)=2$ instances of $v_1$ and $\chi(xyz,x)=1$ instance of $v_2$. Symmetrically, $\theta(r(y,y,z))$ contains $\chi(ux^2,r(y,y,z))=0$ instances of $v_1$ and $\chi(xyz,r(y,y,z))=3$ instances of $v_2$. Moreover, $\theta(x)$ and $\theta(r(y,y,z))$ only differ for the variables occurring in them. Thus, the three edges in the sharing graph of Example \ref{ex:shgraphsimple} correspond to the following aliasings:

\[
\xybox{
  0 *{\theta(x)}="a11",
  "a11"+<1cm,0pt>*{=}="a12",
  "a12"+<1cm,0pt>*{r(v_1,v_1,v_2)}="a13",
  "a11"+<0cm,1cm>*{\theta(r(y,y,z))}="b11",
  "a12"+<0cm,1cm>*{=}="b12",
  "a13"+<0cm,1cm>*{r(v_2,v_2,v_2)}="b13",
   {"a11"+<1.6cm,0.2cm> \ar   "b11"+<1.6cm,-0.2cm>},
   {"a11"+<2.05cm,0.2cm> \ar   "b11"+<2.05cm,-0.2cm>},
   {"a11"+<2.5cm,0.2cm> \ar   "b11"+<2.5cm,-0.2cm>},
}
\]
In particular, the last arrow from $v_2$ to itsself, corresponds to the self-loop in the sharing graph.
\exproofbox
\end{example}

The unification operator $\mgu_\lp([S]_U,x/t)$ can be extended to the case $\vars(x/t)\nsubseteq U$. The idea is to enlarge $S$ by including all the singletons in $\vars(x/t)\setminus U$.

\begin{definition}[Single binding unification with extension]
\label{def:mguomega}
Let $U \in \fwp(\var)$, $S$ be a set of $\omega$-sharing groups with $[S]_U \in \Linp$ and $x/t$ be a binding.
\[
  \mgu_\omega([S]_U,x/t)=\mgu_\omega([S \cup \{ \multil v \multir
  \mid v \in \vars(x/t) \setminus U \} ]_{U \cup \vars(x/t)},x/t)
  \enspace .
\]
\end{definition}

Note that, for a generic abstract domain, the method of extending the
abstract object to include all the variables in the concrete
substitution $\delta$ may result in a non-optimal abstract
unification. For example, this is what happens in the case of the
domain $\Sharing$, as shown in \cite{AmatoS09sharing}. However, we will
prove that, in the case of $\Linp$, the abstract mgu in Definition
\ref{def:mguomega} is optimal.

This operator can be extended to multi-binding substitutions in the obvious way, namely by iterating the single binding operator.
\begin{definition}[Multi-binding unification]
We define $\mgu_\omega([S]_U,\delta)$ with $\delta \in \Isubst$ and $[S]_U \in \Linp$ by induction on the number of bindings:
\[
  \begin{split}
  \mgu_\omega([S]_U,\epsilon)&=[S]_U \enspace ,\\
  \mgu_\omega([S]_U,\{x/t\} \uplus \delta)&=\mgu_\omega(
  \mgu_\omega([S]_U,x/t),\delta)
  \enspace .
  \end{split}
\]
\end{definition}
It is possible to prove that $\mgu_\omega([S]_U,\delta)$ is optimal for multi-binding substitutions \cite{AmatoS05tr-2}. Since optimality of iterative multi-binding unification is not inherited by the abstractions of $\Linp$ (as we show in Section~\ref{sec:multibind}), we will focus on single binding unification. In the rest of the paper, we only consider bindings $x/t$ which are idempotent, namely, such that $x\notin \vars(t)$. 
%
%% Now, we are ready to define the abstract unification in $\Linp$ as:
%% \begin{equation*}
%%   \unif_\lp([S,U_1],\delta,U_2)=[\mgu(S \cup \{ \multil v  \multir
%%   \mid v \in U_2 \setminus U_1 \},\delta),U_1 \cup U_2 ]
%% \end{equation*}
%% provided that $\vars(\delta)\subseteq U_2$.
%% %

\subsection{Correctness of abstract unification}

We now show that $\mgu_\omega([S]_U,\delta)$ is correct \wrt concrete unification. We show correctness for multi-binding substitutions, since it is a trivial extension of the single binding case. In fact, composition of correct operators is still correct.

First of all, we extend the definition of $\theta^{-1}$ to the case when it is applied to a sharing group $B$.
\begin{definition}
Given $\theta \in \Isubst$ and $B$ an $\omega$-sharing group, we define
\begin{equation*}
  \theta^{-1}(B)=\lambda v \in \var. \chi(B,\theta(v)) \enspace .
\end{equation*}
\end{definition}
In order to prove the correctness of abstract unification, we need the
following auxiliary property.
\begin{proposition}
  \label{prop:thetam1}
  Given substitutions $\theta$, $\eta \in \Isubst$ and an $\omega$-sharing group $B$, we have
  \[
    (\eta \circ \theta)^{-1}(B)=\theta^{-1}(\eta^{-1}(B))
    \enspace .
  \]
\end{proposition}

\begin{theorem}[Correctness of $\mgu_\lp$]\label{thm:correctness_mgu_omega}
  The operation $\mgu_\lp$ is correct w.r.t. $\mgu$, \ie
\[
 \forall [S]_U \in \Linp, \delta \in \Isubst.~[S]_U \rightslice [\theta]_{U} \implies
 \mgu_\lp([S]_U,\delta) \rightslice  \mgu([\theta]_{U},\delta) \enspace .
\]

\end{theorem}

\begin{example}
  Let $\theta=\{x/r(s(u,u,u),v,w), y/v', z/w'\}$,
  \allowbreak $\delta=\{x/r(y,y,z)\}$ and $U=\{x,y,z\}$. Therefore $\alpha_\omega([\theta]_U)=
  [x^3,x,y,z]_U$. If we proceed with the concrete unification of
  $[\theta]_U$ with $\delta$, we have
  $\mgu([\theta]_U,\delta)=[\theta']_U$ with
  $\theta'=\mgu(\theta,\delta)=\eta \circ \theta$ and
  $\eta=\mgu(\theta(x)=\theta(r(y,y,z)))$. This gives the following
  results:
  \begin{gather*}
    \eta=\{v'/s(u,u,u), v/s(u,u,u), w'/w\} \enspace ,\\
    \theta'=\{x/r(s(u,u,u),s(u,u,u),w), y/s(u,u,u), z/w, v'/s(u,u,u),
    w'/w\} \enspace ,
  \end{gather*}
  with $[\theta']_U=[\theta]_U$. Now, let $\eta'$ be obtained from $\eta$
  by replacing each occurrence of a variable in $\rng(\eta)$ with a
  different fresh variable, $\beta=\eta' \circ \theta$ and $\rho$ be a
  substitution mapping variables to variables 
  such that
  $\rho(\beta(x))=\theta'(x)$ for each $x \in U$. Note that $\rho$ is not a renaming, since it is not bijective.
   We have:
  \begin{gather*}
    \eta=\{v/s(u_1,u_2,u_3), v'/s(u_4,u_5,u_6), w'/u_7\} \enspace, \\
    \beta=\{x/r(s(u,u,u),s(u_1,u_2,u_3),w), y/s(u_4,u_5,u_6), z/u_7,
    v'/s(u_4,u_5,u_6), w'/u_7 \} \enspace, \\
    \rho=\{u_1/u, u_2/u, u_3/u, u_4/u, u_5/u, u_6/u, u_7/w\} \enspace .
  \end{gather*}
  Following the proof, we build a multigraph $G$ as follows:
  
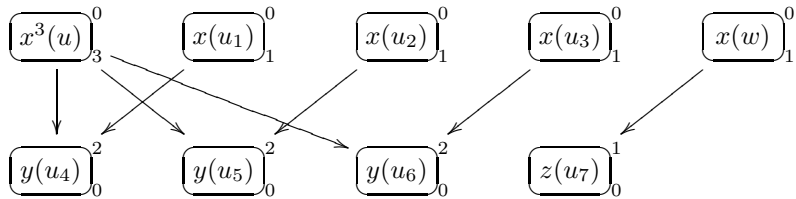
\[
  \xymatrix{
    \ovalemidi{x^3}{u}_3^0 \ar[d] \ar[dr] \ar[drr] &
    \ovalemidi{x}{u_1}_1^0 \ar[dl] & \ovalemidi{x}{u_2}_1^0
    \ar[dl] &
    \ovalemidi{x}{u_3}_1^0  \ar[dl]& \ovalemidi{x}{w}_1^0 \ar[dl]\\
    \ovalemidi{y}{u_4}_0^2 & \ovalemidi{y}{u_5}_0^2 & \ovalemidi{y}{u_6}_0^2 &
    \ovalemidi{z}{u_7}_0^1
  }
  \]
  Note that we have chosen to annotate every sharing group with the corresponding variable in $\vars(\beta(U))$. This is not a
  sharing graph since it is not connected, but if we take
  $Y=\supp{\rho^{-1}(u)}=\{u,u_1,u_2,u_3,u_4,u_5,u_6\}$, the
  restriction of $G$ to the nodes annotated with a variable in $Y$ is
  a sharing graph whose resultant $\omega$-sharing group is $x^6y^3$.
  \exproofbox
\end{example}

\subsection{Optimality of Abstract Unification}

We now prove that $\mgu_\lp$ is not only correct, but also optimal for a single binding substitution, \ie it is the least correct
abstraction.  This means proving that, given a set of $\omega$-sharing groups $[S]_U \in \Linp$, a binding $x/t$, and an $\omega$-sharing group $B \in \mgu_\lp([S]_U,x/t)$, there exists
a substitution $[\delta]_U$ such that $[S]_U \rightslice [\delta]_U$ and $B \in
\alpha_\lp(\mgu([\delta]_U,\{x/t\}))$.  First of all, we
prove optimality of $\mgu_\lp([S]_U,x/t)$ in the special case of
$\vars(x/t) \subseteq U$. Next, we extend this result to the general case.
\begin{example}\label{ex:optimality}
Consider $S=\{xu,xv,y\}$ and the binding $x/s(y,y)$. The following is a sharing graph for $x/s(y,y)$ and $S$ whose resultant $\omega$-sharing group is $x^2uvy$.
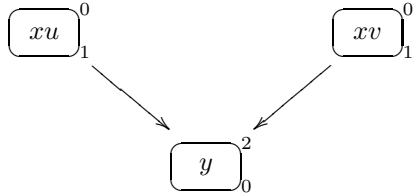
\[
  \xymatrix{
    \ovalee{xu}_1^0 \ar[dr] && \ovalee{xv}_1^0 \ar[dl]\\
    & \ovalee{y}_0^2
  }
\]
We show how to find a substitution $[\delta]_U$ such that the $\omega$-sharing group $x^2uvy \in \alpha_\lp(\mgu([\delta]_U,\{x/s(y,y)\}))$. Let $U=\{u,v,x,y\}$. For each node $n$ of the sharing graph, we consider a different fresh variable $w_n$. Assume that the node labelled with $xu$ in the upper-left corner is node 1, and proceed clockwise to number the other nodes.

For each variable $z\in U\setminus \{x\}$, we associate to $\delta(z)$ a term containing all the variables $w_i$ such that the label of the $i$-th node contains the variable $z$.
Thus, we define $\delta(u)=r(w_1)$ where $w_1$ correspond to the node containing $u$. Analogously, we define $\delta(v)=r(w_2)$ and $\delta(y)=r(w_3)$.

We now define $\delta(x)$ in a different way, namely by replacing in $s(y,y)$ each occurrence of the variable $y$ with a term similar to $\delta(y)$, with the difference that $w_3$ is replaced with the variables  $w_1$ and $w_2$. The choice of $w_1$ and $w_2$ is obvious by looking at the sharing graph, since the first and second node are the sources of the two edges targeted at the node three. Therefore we obtain $\delta(x)=s(r(w_1),r(w_2))$.

Summing up, we have:
\[
\delta=\{u/r(w_1), v/r(w_2), x/s(r(w_1),r(w_2)), y/r(w_3)\} \enspace .
\]
It is easy to check that $[S]_U \rightslice [\delta]_U$ and
\begin{multline*}
\mgu(\delta,\{x/s(y,y)\})=\\
\{u/r(w_1), v/r(w_1), x/s(r(w_1),r(w_1)), y/r(w_1), w_2/w_1, w_3/w_1  \} \enspace ,
\end{multline*}
hence $\alpha_\lp([\mgu(\delta,\{x/s(y,y)\})]_U)=[x^2uvy]_U$.
  \exproofbox
\end{example}

In the above example we have shown how to find a special substitution such that its fresh variables are unified according to the arrows in a sharing graph. The same idea is exploited in the next theorem for proving the optimality of the abstract unification operator $\mgu_\lp$. For any $\omega$-sharing group $X\in \mgu_\lp([S]_U,x/t)$, we provide a substitution $\delta$ obtained as in Example \ref{ex:optimality}, such that $[S]_U$ approximates $[\delta]_U$ and
$X \in  \alpha_\lp(\mgu([\delta]_U,\{x/t)\})$.
\begin{theorem}[Optimality of $\mgu_\lp$]
  \label{th:lpopt}
  The single binding unification $\mgu_\lp([S]_U,x/t)$ is optimal \wrt
  $\mgu$, under the assumption that $\vars(x/t) \subseteq U$, \ie:
\[
\forall B\in\mgu_\omega ([S]_U,x/t)~\exists \delta\in  \Isubst.~[S]_U \rightslice [\delta]_U \text{ and } B\in \alpha_\lp(\mgu([\delta]_U,\{x/t\})) \enspace .
 \]

\end{theorem}

The previous proof requires that $\vars(x/t) \subseteq U$. However, the same construction also works when this condition  does not hold.
\begin{example}
  Given $U=\{x,y\}$ and $S=\{x^2,x^2y\}$, we want to compute $\mgu_\lp([S]_U,x/s(y,z))$. By extending the domain of the variables of interests, we obtain $[S']_V= [x^2,x^2y,z]_{x,y,z}$. One of the sharing graphs for $x/s(y,z)$ and $[S']_V$ is
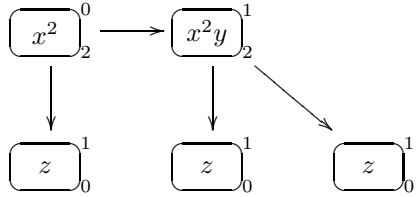
\begin{equation*}
\xymatrix{
  \ovalee{x^2}^0_2 \ar[r]  \ar[d]& \ovalee{x^2y}^1_2 \ar[d] \ar[dr]\\
  \ovalee{z}^1_0 & \ovalee{z}^1_0 & \ovalee{z}^1_0
}
\end{equation*}
Following the proof of the previous theorem, we obtain the substitution
\[
\delta'=\{x/s(r(w_1),r(w_1,w_2,w_2)), y/r(w_2), z/r(w_3,w_4,w_5)\} \enspace,
\]
where $[S']_V \rightslice [\delta']_V$ and $x^4yz^3 \in \alpha_\lp(\mgu([\delta']_V,\{x/s(y,z)\}))$. However  we are looking for  a substitution $\delta$ such that $[S]_U \rightslice [\delta]_U$ and $x^4yz^3 \in \alpha_\lp(\mgu([\delta]_U,\{x/s(y,z)\}))$. Nonetheless, we may choose $\delta=\delta'$ (or, if we prefer, $\delta=\delta'|_{x,y}$) to get the required substitution.
  \exproofbox
\end{example}

This is not a fortuitous coincidence. We may show that it consistently happens every time we apply Theorem \ref{th:lpopt} to an abstract unification where $\vars(x/t) \nsubseteq U$. Therefore, we can prove the main result of the paper.

\begin{theorem}[Optimality of $\mgu_\lp$ with extension]\label{th:optimality_extension}
  The single binding unification $\mgu_\omega$ with extension is optimal \wrt
  $\mgu$.
\end{theorem}

% Note that the operation $\mgu_\lp$ is designed by first extending the
% domain in order to include all the variables in $V$ and then
% performing the operation, and that this construction yields an
% optimal abstraction of the concrete unification. This is not the case
% for other abstract domains, e.g.~$\Sharing$, as shown in
% \cite{AmatoS05tr-1}.

\subsection{A characterization for resultant sharing groups}

The domain $\Linp$ has not been designed to be directly
implemented, but some of its abstractions could.
Providing a simpler
definition for the set of resultant $\omega$-sharing groups could help
in developing the abstract operators for its abstractions. We show
that given a set $S$ of $\omega$-sharing groups and a binding $x/t$,
the set of resultant $\omega$-sharing groups has an elegant algebraic
characterization.

By definition of sharing graph, a set of nodes $N$ labelled with $\omega$-sharing groups of $S$ can be turned into a sharing graph for $S$ and $x/t$ if and only if the condition on the out-degree and in-degree is satisfied and the obtained graph is connected. The condition on the degrees says that for each node $s$ labelled with the sharing group $B_s$, the out-degree of $s$ must be equal to $\chi(B_s,x)$. Symmetrically, the in-degree must be equal to $\chi(B_s,t)$. As a consequence, the sum of the out-degrees of all the nodes $\sum_{s\in N} \chi(B_s,x)$ must be equal to the sum of the in-degrees of all the nodes $\sum_{s\in N}\chi(B_s,t)$. This is because each edge  has a source and a target node. Moreover, in order to be connected, any graph needs at least $\card{N}-1$ edges. Since the number of edges is equal to the sum of in-degrees of all the nodes, it turns out that such a sum must be equal to or greater than $\card{N}-1$. Surprisingly, this is enough to construct a sharing graph from $N$.

\begin{theorem}
  \label{th:algebraic}
  Let $S$ be a set of $\omega$-sharing groups and $x/t$ be a binding. Then $B \in \mgu_\lp(S,x/t)$ iff there exist $n\in \Natzero$, $B_1,\ldots,B_n \in S$ which satisfy the following conditions:
  \begin{enumerate}
    \item $B=\multisum_{1 \leq i \leq n} B_i$;
    \item $\sum_{1 \leq i \leq n} \chi(B_i,x) = \sum_{1 \leq i \leq n}
      \chi(B_i,t) \geq n -1$;
    \item either $n=1$ or $\forall 1 \leq i \leq n.\ \chi(B_i,x) +
      \chi(B_i,t) > 0$.
 \end{enumerate}
\end{theorem}

Following the above theorem, we can give an algebraic characterization
of the abstract unification operator as follows.

\begin{corollary}[Algebraic characterization of $\mgu_\omega$]
  \label{cor:mgu-algebraic}
  Given a set of $\omega$-sharing groups $S$ and a binding $x/t$, we have that
  \[
  \begin{split}
  \mgu_\lp(S,&\ x/t)= (S\setminus \relev(S,x,t)) \cup \\
  & \left\{ \multisum \calS \mid \calS \in
    \mwp(\relev(S,x,t)), \sum_{B \in \calS} \chi(B,x) = \sum_{B
      \in \calS} \chi(B,t) \geq \card{\calS}-1 \right\} ,
  \end{split}
  \]
   where
  \[
    \begin{split}
    \relev(S,x,t) &= \{B\in S.\ \chi(B,x) + \chi(B,t) > 0\} \\
      &=\{ B \in S.\ \supp{B} \cap \vars(x/t) \neq \emptyset \}
      \enspace .
    \end{split}
  \]
\end{corollary}

\begin{example}
  Consider $S=\{xa, xb, z^2, zc\}$ and the equation $x=z$.  Then if we
  choose $\calS=\multil xa, xb, z^2 \multir$, we have
  $\sum_{B \in \calS} \chi(B,x)=2=  \sum_{B \in \calS}  \chi(B,z) \geq \card{\calS}-1$. Therefore $x^2 z^2 ab \in
  \mgu_\lp(S,x/z)$. If we take $\calS=\multil xa, xb, zc, zc \multir$,
  although $\sum_{B \in \calS} \chi(B,x)=2=  \sum_{B \in \calS}  \chi(B,z)$, we have $\card{\calS}-1=3$. This only
  proves that $z^2 c^2 x^2 ab$ cannot be obtained by the multiset $\calS$.
  If we check for every possible multiset over $S$, we have that $z^2
  c^2 x^2 ab \notin \mgu_\lp(S,x/z)$.  \exproofbox
\end{example}

This characterization of the abstract mgu will be the key point for devising the optimal abstract unification operators on the abstractions of $\Linp$.
Let $\alpha$ be the abstraction function from $\Isubst_\sim$ to an abstract domain $A$.
If we are able to factor $\alpha$ through a Galois connection $\langle \alpha': \Linp \ra A, \gamma': A \ra \Linp \rangle$ as $\alpha=\alpha'\circ \alpha_\omega$, then the optimal abstract unification for $\alpha$ is exactly $\alpha' ( \mgu_\omega (\gamma'(\cdot), \cdot))$. However, this expression is helpful when it may be simplified in order to use only objects in $A$.  Our algebraic characterization makes the simplification feasible, as we show in the following section.

\section{Practical domains for program analysis}
\label{sec:practical}

We consider two domains for sharing analysis
with linearity information, namely, the domain proposed in
\cite{King94} and the classical reduced product $\Sharing \times \Lin$.
They are defined as abstractions of $\Linp$ through Galois insertions. 
This allows us to design
optimal abstract operators for both of them, by exploiting the results
introduced so far.
By composing each Galois insertion with $\alpha_\lp$, we get the corresponding abstraction function for substitutions \cite[Sect.~4.2.3.1]{CousotC92lp}.

\subsection{King's domain for linearity and aliasing}
\label{sec:practical1}

We first consider the domain for combined analysis of sharing and
linearity in \cite{King94}. The idea is to enhance the domain $\Sharing$ by annotating each sharing group with linearity information on each variable. For instance, the object $xy^\infty z$ represents the sharing group $xyz$ and the information that $y$ may be non-linear (while $x$ and $z$ are definitely linear).
The objects in this domain can be easily viewed as abstraction of $\omega$-sharing groups. Intuitively, in order to abstract an $\omega$-sharing groups, one simply needs to replace each exponent equal to or greater than $2$ with $\infty$. Let us now formalize the domain as an abstraction of $\Linp$.

An $\omega$-sharing group (which is a multiset $\var \fun \Nat$ whose support is finite) is abstracted into a map $o: \var \fun \{0,1,\infty\}$ such that its support $\supp{o}=\{ v \in
\var \mid o(v) \neq 0\}$ is finite. We call such a map the \emph{2-sharing group}.
We use a polynomial notation for 2-sharing groups as
for $\omega$-sharing groups. For instance, $o=xy^\infty z$ denotes
the 2-sharing group whose support is
$\supp{o}=\{x,y,z\}$, such that $o(x)=o(z)=1$ and $o(y)=\infty$. We denote with $\emptyset$ the 2-sharing group with
empty support. Note that in \cite{King94} the number $2$ is used as an
exponent instead of $\infty$, but we prefer this notation to be
coherent with $\omega$-sharing groups.

We denote $\min\{o(x),2\}$ by $o_m(x)$, where the ordering on $\Nat$ is extended in the obvious way, \ie for all $n\in \Nat$ we have that $n < \infty$. A 2-sharing group $o$ represents the sets $\gamma_\an(o)$
of $\omega$-sharing group given by:
\[
\gamma_\an(o)=\{ B \in \mwp(\var) \mid \supp{o}=\supp{B} \wedge \forall
x \in \supp{o}. o_m(x) \leq B(x) \leq o(x) \}\enspace .
\]
For instance, the 2-sharing group $xy^\infty z$ represents the set of $\omega$-sharing groups $\{xy^{2} z,xy^{3} z,xy^{4} z,xy^{5} z,\ldots\}$. The idea is to use 2-sharing groups to keep track of linearity: if
$o(x)=\infty$, it means that the variable $x$ is not linear in the 2-sharing group $o$. In the rest of this section, we use the term ``sharing group'' as a short form of 2-sharing group, when this does not cause ambiguity.

An $\omega$-sharing group $B$ may be abstracted into the 2-sharing group $\alpha_\an(B)$
given by:
\[
\alpha_\an(B)=\lambda v\in \supp{B}.
  \begin{cases}
    1  \text{ if $B(x)=1$,}\\
    \infty  \text{ otherwise.} 
  \end{cases}
\]
The next proposition shows two useful properties of the maps $\alpha_\an$ and $\gamma_\an$. 

\begin{proposition}
\label{prop:abstraction}
The following properties hold:
\begin{enumerate}
\item \label{eq:abstraction1} $\alpha_\an(\bigmultisum \calS)=
  \Andybin \alpha_\an(\calS)$.
\item \label{eq:abstraction2} $\relev(\gamma_\an(S),x,t)) =
  \gamma_\an( \relev(S,x,t)) $.
\end{enumerate}
\end{proposition}

Since we do not want to represent definite
non-linearity, we introduce an order relation over sharing groups as
follows:
\[
o \leq o' \iff \supp{o}=\supp{o'} \wedge
  \forall x \in \supp{o}.\ o(x) \leq o'(x)  \enspace ,  
\]
and we restrict our attention to downward closed sets of sharing
groups. We denote by $\Andysh(V)$ the set of 2-sharing groups whose support is
a subset of $V$.  The domain we are interested in is the following:
\[
 \ShLinp=\bigl\{ [S]_U \mid S \in \wp_{\downclo}(\Andysh(U)),
  U \in \fwp(\var), S \neq \emptyset \Rightarrow \emptyset \in S \bigr \}
  \enspace ,
\]
where $\wp_{\downclo}(\Andysh(U))$ is the powerset of downward closed
subsets of $\Andysh(U)$ according to $\leq$ and $[S_1]_{U_1} \leq_\an
[S_2]_{U_2}$ iff $U_1=U_2$ and $S_1 \subseteq S_2$.
For instance, the set $\{xy^\infty z\}$ is not downward closed, while $\{xyz,xy^\infty z\}$ is downward closed.
There is a Galois
insertion of $\ShLinp$ into $\Linp$ given by the pair of adjoint maps
$\gamma_\an: \ShLinp \fun \Linp$  and  $\alpha_\an: \Linp \fun \ShLinp$:
\begin{align*}
  \gamma_\an([S]_U)&=\left[ \bigcup \{ \gamma_\an(o) \mid o \in S \} \right]_U
  \enspace ,\\
  \alpha_\an([S]_U)&=\left[ \downclo \{ \alpha_\an(B) \mid B \in S \} \right]_U
  \enspace .
\end{align*}
With an abuse of notation, we also apply $\gamma_\an$ and $\alpha_\an$ to subsets of $\omega$-sharing groups and 2-sharing groups
respectively, by ignoring the set of variables of interest. For instance, $\gamma_\an(\{xyz,xy^\infty z\})= \{xyz, xy^{2} z,xy^{3} z,xy^{4} z,xy^{5} z,\ldots\}$. 
\begin{theorem}\label{th:linp-galois}
 The pair $\langle \alpha_\an, \gamma_\an \rangle$ is a Galois insertion.
\end{theorem}
Now we may define the optimal mgu for $\ShLinp$ and single binding substitutions as follows:
\begin{definition}[Unification for $\ShLinp$]
  Given $[S]_U \in \ShLinp$ and the binding $x/t$, we define
  \[
    \mgu_\an([S]_U,x/t)= \alpha_\an(\mgu_\lp(\gamma_\an([S]_U),x/t)) \enspace .
  \]
\end{definition}
By construction, $\mgu_\an$ is the optimal abstraction of $\mgu_\lp$, hence also of $\mgu$. In the case where $\vars(x/t) \subseteq U$, by using additivity of $\alpha_\an$ we
get:
\begin{multline}
  \label{eq:mguandythefirst}
  \mgu_\an([S]_U,x/t)=\Bigl[\alpha_\an(\gamma_\an(S) \setminus
  \relev(\gamma_\an(S),x,t)) \cup \\
  \alpha_\an \Bigl(\{ \multisum \calS \mid \calS \in
  \mwp(\relev(\gamma_\an(S),x,t)),\\
  \sum_{B \in \calS} \chi(B,x) = \sum_{B \in \calS}
  \chi(B,t) \geq \card{\calS}-1 \}\Bigr)\Bigr]_U \enspace .
\end{multline}

%% The l.u.b.~is given by the downward closure of the unions
%% of the first components. Projection is given by:
%% \begin{equation}
%%   \pi_\an([S_1,U_1],[S_2,U_2])=[ \{ o|_{U_2} \mid o \in S_1 \},U_1
%%   \cap U_2 ] \enspace .
%% \end{equation}
%% where $o|_{X}(v)$ is $o(v)$ if $v \in X$, $0$ otherwise.

%According to the corresponding definition for
%$\omega$-sharing groups, we also need to define the following
%auxiliary function.
%\begin{align}
%  S^*&=\left\{ \Andybin \calS \mid \calS \in \mwp(S) \right\}
%  \enspace .
%\end{align}
%%

Now we want to simplify Eq.~\ref{eq:mguandythefirst}.  In particular we
would like to get rid of the abstraction and concretization maps
and to express the result using only objects and operators in $\ShLinp$.
Therefore, we need to define operations in $\ShLinp$ which correspond
to $\uplus$ and $\chi$ in $\Linp$.

The operation on 2-sharing groups which corresponds to multiset union
on $\omega$-sharing groups, is given by
\[
  o \andybin o' = \lambda v \in \var. o(v) \oplus o'(v) \enspace ,
\]
where $0 \oplus x= x \oplus 0=x$ and $\infty \oplus x=x \oplus
\infty=1 \oplus 1= \infty$. We will use $\Andybin \multil o_1, \ldots,
o_n \multir$ for $o_1 \andybin \cdots \andybin o_n$.  Given a sharing
group $o$, we also define the \emph{delinearization} operator:
\begin{equation}
 \label{eq:delinearization}
 o^2=o \andybin o  \enspace .
\end{equation}
Note that $o^2=\lambda x\in \supp{o}.\infty$. The
operator is extended pointwise to sets and multisets.

A fundamental role is played by the notion of
multiplicity of a sharing group in a term.  While the multiplicity of
an $\omega$-sharing group in a term is a single natural number, every
object in $\ShLinp$ represents a set of $\omega$-sharing groups, hence
its multiplicity should be a set of natural numbers.  Actually, it is
enough to consider intervals.  We define the minimum $\chi_m$ and
maximum $\chi_M$ multiplicity of $o$ in $t$ as follows:
\[
  \chi_m(o,t)=\sum_{v \in \supp{o}} o_m(v) \cdot \occ(v,t) \qquad
  \chi_M(o,t)=\sum_{v \in \supp{o}} o(v) \cdot \occ(v,t) \enspace .
\]
Sum and product on integers are lifted in the obvious way, namely, the
sum is $\infty$ if and only if at least one of the addenda is $\infty$
and $n \cdot \infty = \infty \cdot n = \infty$ for any $n\in
\Natzero$, while $0 \cdot \infty = \infty \cdot 0 = 0$.
 The maximum multiplicity $\chi_M(o,t)$ either is equal to
the minimum multiplicity $\chi_m(o,t)$ or it is infinite.
Note that, if $B$ is an $\omega$-sharing group represented
by $o$, \ie $B \in \gamma_\an(o)$, then $\chi_m(o,t) \leq \chi(B,t)
\leq \chi_M(o,t)$.  Actually, not all the values between $\chi_m(o,t)$
and $\chi_M(o,t)$ may be assumed by $\chi(B,t)$.
\begin{example}
  Let $o=x^\infty$ and $t=f(x,x)$. According, to our definition, $\chi(o,t)=[4,\infty)$. However, it is obvious that if $B \in \gamma_2(o)$, then $\chi(B,t)$ is an even number.   \exproofbox
\end{example}

According to the above definitions,
we define the multiplicity of a multiset of sharing groups as
\[
  \chi(Y,t)=\Big\{n\in \Nat~|~\sum_{o \in Y} \chi_m(o,t) \leq n \leq
  \sum_{o \in Y} \chi_M(o,t) \Big\}
  \enspace .
\]
Even if this is a superset of all the possible values which
can be obtained by combining the multiplicities of all the sharing
groups in $Y$, this definition is
sufficiently accurate to allow us to design the optimal abstract
unification.

We extend in the obvious way the definition of $\relev$ (see Corollary \ref{cor:mgu-algebraic}) from $\omega$-sharing groups to 2-sharing groups, \ie
  \[
\relev(S,x,t)=\{ o \in S \mid \supp{o} \cap \vars(x/t) \neq \emptyset\} \enspace ,
  \]
and we prove the following
\begin{theorem}[Characterization of abstract unification for $\ShLinp$]
  \label{th:pre-mguandy}
  Given $[S]_U \in \Linp$ and the binding $x/t$ with $\vars(x/t) \subseteq U$, we have that
  \[  
   \begin{split}
      \mgu_\an([S]_U,x/t) & = [(S\setminus S') \cup \\
      & \downclo \bigl\{ \Andybin Y \mid Y \in \mwp(S'),
      n \in \chi(Y,x) \cap \chi(Y,t).\ n \geq \card{Y}-1
   \bigr\}]_U \enspace ,
  \end{split}
  \]
  where $S'=\relev(S,x,t)$.
\end{theorem}

\begin{example}
  Let $S=\downclo \{\emptyset, ux^\infty, vx^\infty, x^\infty y, z^\infty \}$
  and $Y=\multil ux^\infty, vx^\infty, xy, z^\infty \multir$. We
  have that $\chi(Y,x)=\{ n \mid n \geq 5 \}$ and $\chi(Y,f(z,z))= \{ n
  \mid n \geq 4 \}$. Since $f(z,z)$ contains two occurrences of $z$,
  the ``actual'' multiplicity of the sharing group $z^\infty $ in
  $f(z,z)$ should be a multiple of $2$. But we do not need to check
  this condition and can safely approximate this set with $\{ n \mid n
  \geq 4 \}$.  This works because we can always choose a
  number which is contained in both $\chi(Y,x)$ and $\chi(Y,t)$ and
  which is an ``actual'' multiplicity. For instance, we can take $n=6
  \in \chi(Y,x) \cap \chi(Y,f(z,z))$ and since we have $6 \geq
  3=\card{Y}-1$, we get that the sharing group $\Andybin Y=uvx^\infty
  yz^\infty$ belongs to $\mgu_\an([S]_U,x/f(z,z))$. This sharing group can
  be generated by the substitution $\{ x/f(f(u,u,y),f(v,v,y)),
  z/f(w,w,w)\}$ when the variables of interest are $\{u,v,x,y,z\}$.  \exproofbox
\end{example}

Theorem \ref{th:pre-mguandy} gives a characterization of the abstract unification over $\ShLinp$. However, it cannot be directly implemented, since one needs to check a certain condition for each element of $\mwp(\relev(S,x,t))$, which is an infinite set. Nonetheless, this is an important starting point to prove correctness and completeness of the abstract unification algorithm which we are going to introduce.

% \begin{note}
% It would be interesting to ask ourself whether it is possible to modify the
% definition of $\chi_m(o,t)$ for a $2$-sharing group $o$ in such a way that
% $\chi_m(o,t)=\chi_m(l(o),t)$. The answer is no. For example, consider the
% substitution $\theta=\{x/y\}$ and $S=\{ \emptyset, ux, ux^\infty, y\}$. The
% sharing group $o=ux^\infty y$ cannot be generated by $\mgu_\an(S,\theta)$: the
% only way we may obtain $o$ is from the multiset $\multil ux^\infty, y \multir$.
% However, $\chi(X,x)=[2, \infty)$ and $\chi(X,y)=\{1\}$ and their intersection
% is empty. If we modify the definition of $\chi$ as stated above, we have
% $\chi(X,x)=[1,\infty)$ and $\chi(X,y)=\{1\}$. Since $\card{X}=2$, we obtain $o$
% as the valid result.
% \end{note}

The characterization in Theorem \ref{th:pre-mguandy} may be used even when $\vars(x/t)\nsubseteq U$, if we first
enlarge the set of variables of interest in order to include all $\vars(x/t)$.
\begin{theorem}[Characterization of abstract unification with extension for $\ShLinp$]
\label{th:mguandy2}
Given $[S]_U$ in $\ShLinp$ and the binding $x/t$, let $V=\{v_1, \ldots, v_n\}$ be $\vars(x/t) \setminus U$. Then,
\[
  \mgu_\an([S]_U,x/t)= \mgu_\an([S \cup \{ v_1, \ldots, v_n \}]_{U \cup V},x/t) \enspace .
\]
\end{theorem}
The previous theorem states that enlarging the set of variables of interest preserves optimality.

\subsection{An algorithm for abstract unification in $\ShLinp$}
\label{sec:andyalgo}

% We first need to give some definitions which will be used in the
% abstract unification algorithm.

In order to obtain an algorithm from the characterization in Theorem \ref{th:pre-mguandy} we need to avoid the use of $\mwp(\relev(S,x,t))$ and to develop a procedure able to compute the resultant sharing groups by inspecting subsets (not multisets!) of $\relev(S,x,t)$. In general, any $X \subseteq \relev(S,x,t)$ yields more than one sharing group, since every element in $X$ may be considered more than once. However, since $\ShLinp$ is downward closed, it is enough to compute the maximal resultant sharing groups.

Given $X \subseteq \relev(S,x,t)$ and the binding $x/t$, assume that we are only interested in those sharing groups whose support is $\supp{\bigmultisum X}$. By joining (multiple copies of) the sharing groups in $X$, any resultant sharing group $o$ is between $\bigmultisum X$ and $\bigmultisum X^2$, \ie $\bigmultisum X\leq o \leq \bigmultisum X^2$, where $X^2$ is the pointwise extension of the delinearization operator (see Eq.~\ref{eq:delinearization}). Note that, if $X$ is badly chosen, it is possible that we are not able to generate any sharing group with this support. In this computation, the notion of multiplicity of a sharing group in a term plays a major role.

For example, given the binding $x/t$, if $\chi_M(o,x) \leq 1$ for each $o \in X$,  then $\bigmultisum X$ is a resultant sharing group only if there is a unique sharing group $o \in X$ such that $\vars(t) \cap \supp{o} \neq \emptyset$.  If there are $o_1, o_2 \in X$ such that $\chi_M(o_1,x)>1$ and
$\chi_M(o_2,t)>1$ then $\bigmultisum X$ is a resultant sharing group. Moreover, we may join two copies of each sharing group in $X$, and therefore also $\bigmultisum{X^2}$ is a resultant sharing group.

Now we can define the notions of linearity and non-linearity on the abstract domain. In addition, we also introduce a new notion of strong non-linearity. Given $X \subseteq \relev(S,x,t)$, we partition $X$ in three subsets $X_x = \{o\in X~|~\chi_M(o,t)=0\}$, $X_t = \{o\in X~|~\chi_M(o,x)=0\}$ and $X_{xt} = X \setminus (X_x \cup
X_t)$.

\begin{definition}
\label{def:andyalgo}
Given a set $S$ of sharing groups and $X \subseteq \relev(S,x,t)$, we say that $X$ is:
\begin{itemize}
\item \emph{linear} for the term $t$ if for all $o \in X$ it holds
  that $\chi_M(o,t) \leq 1$;
\item \emph{non-linear} for the term $t$ if there exists $o \in X$
  such that $\chi_M(o,t) > 1$;
\item \emph{strongly non-linear} for the term $t$ if there exists $o
  \in X$ such that $\chi_M(o,t)=\infty$ or there exists $o \in X_{xt}$ such
  that $\chi_M(o,t)>1$.
\end{itemize}
Analogously, we define linearity and non-linearity of $X$ for the variable $x$.
\end{definition}
Note that, if $t$ is a variable, the non-linear and strongly non-linear cases coincide. We now present the algorithm for computing the abstract unification in $\ShLinp$.

\begin{theorem}[Abstract unification algorithm for $\ShLinp$]
\label{th:mguandy}
Given $[S]_U \in \ShLinp$ and the  binding  $x/t$ with $\vars(x/t) \subseteq U$, we have
  \begin{equation*}
     \mgu_\an([S]_U,x/t)  = [(S\setminus S') \cup
      \downclo \bigcup_{X \subseteq S'} \res(X,x,t)]_U \enspace ,
   \end{equation*}
  where $S'=\relev(S,x,t)$ and $\res(X,x,t)$ is defined as follows:
\begin{enumerate}
\item if $X$ is non-linear for $x$ and $t$,  then $\res(X,x,t)=\{ \Andybin X^2\}$;
\item if $X$ is non-linear for $x$ and linear for $t$,
  $\card{X_x} \leq 1$ and $\card{X_t} \geq 1$, then we have
  $\res(X,x,t)=\{ (\Andybin X_x) \uplus (\Andybin X_{xt}^2) \uplus (\Andybin
  X_t^2)\}$;
\item if $X$ is linear for $x$ and strongly non-linear for $t$,
$\card{X_x} \geq 1$ and $\card{X_t} \leq 1$,
then we have
$\res(X,x,t)=\{ (\Andybin X_x^2) \uplus (\Andybin X_{xt}^2) \uplus (\Andybin X_t)\} $;
\item if $X$ is linear for $x$ and not strongly non-linear for $t$,
  $\card{X_t} \leq 1$, then we have
\[
\begin{split}
\res(X,x,t)=\{(\Andybin Z) \uplus
  (\Andybin X_{xt}^2) \uplus (\Andybin X_t) ~|~ & Z \in \mwp(X_x), \\
   & \card{Z}=\chi_M(X_t,t)=\chi_m(X_t,t), \\
   & \supp{Z} = X_x\} \enspace ;
\end{split}
\]
\item otherwise $\res(X,x,t)=\emptyset$.
\end{enumerate}
\end{theorem}

\begin{example}
\label{ex:mguandy1}
 Let $U=\{u,v,x,y\}$ and consider the set of $2$-sharing groups $S=\{\emptyset,xu,x^\infty,xy,yv\}$ and consider the binding $x/r(y,y)$. Note that $S'=\{xu,x^\infty,xy,yv\}$. Let us compute $\res(X,x,r(y,y))$ for some $X$'s, subsets of $S'$.
\begin{itemize}
 \item $X=\{x^\infty, yv\}$. In this case, $\chi_M( x^\infty,x)=\infty$ and $\chi_M(yv,r(y,y)=2$, hence $X$ is non-linear for $x$ and $r(y,y)$. From the first case of Theorem \ref{th:mguandy}, we have that
$\res(X,x,r(y,y))=\{ \Andybin X^2\}=\{\Andybin \{x^\infty, y^\infty v^\infty \}=\{x^\infty y^\infty v^\infty \}$;
\item $X=\{xu, xy, yv\}$. Then $X$ is linear for $x$ and strongly non-linear for $r(y,y)$, since $ xy \in X_{xt}$ and $\chi_M(xy,r(y,y))=2$. From the third case, it follows that $\res(X,x,t)=\{ (\Andybin \{xu\}^2) \uplus (\Andybin \{xy\}^2) \uplus (\Andybin \{yv\})\} = \{ x^\infty y^\infty u^\infty v \}$;
\item $X=\{xu, yv\}$. Then $X$ is linear for $x$ and not strongly non-linear for $r(y,y)$ (note that $\chi_M(yv,r(y,y))=2>1$ and $yv \in X_{t}$, hence $X$ is non-linear for $r(y,y)$ but it is not strongly non-linear). Since $\chi_M(X_t,r(y,y))=2$, we only need to consider those $Z \in \mwp(X_x)$ such that $\card{Z}=2$. There is only one such set, which is $Z=\multil xu, xu \multir$. Therefore $\res(X,x,r(y,y))=\{ (\Andybin \multil xu, xu \multir ) \uplus (\Andybin \{\}^2) \uplus (\Andybin \{ yv \}) \} = \{ x^\infty y u^\infty v \}$.   \exproofbox
\end{itemize}
\end{example}

Note that, given $X \subseteq S'$, if $x$ does not appear in any sharing group of $S$,
then $\res(X,x,t)\subseteq \{\emptyset\}$. In fact, we can only apply the fourth or fifth case. In the fourth case, we have that $X_x=X_{xt}=\emptyset$, and thus the only $Z\in \mwp(X_x)$ is the empty multiset. Thus, $\card{Z}=0$, which implies that $X_t=\emptyset$, and $\res(X,x,t)=\{\emptyset\}$. In the fifth case, the result is trivially the emptyset. Symmetrically, when none of the variables of $t$ appears in $S$, again we can apply only the fourth or fifth case, and $\res(X,x,t)\subseteq \{\emptyset\}$.

\begin{example}
 Consider $S$ and $U$ as in Example \ref{ex:mguandy1}. We compute $\mgu_\an([S]_U,x/r(y,y))$. We show the value of $\res(X,x,r(y,y))$ for every $X \subseteq S' = \relev(S,x,r(y,y))$ which contains both the variables $x$ and $y$:
\[
 \begin{array}{l|l|l}
   \mathbf{X} & \mathbf{\res(S,x,r(y,y))} & \textbf{case in Theorem \ref{th:mguandy}} \\
   \hline x^\infty, xy & x^\infty y^\infty & 1 \\
   x^\infty, yv & x^\infty y^\infty v^\infty & 1 \\
   x^\infty, xy, yv & x^\infty y^\infty v^\infty & 1\\
   x^\infty, xu, xy & x^\infty y^\infty u^\infty & 1\\
   x^\infty, xu, yv & x^\infty y^\infty u^\infty v^\infty & 1\\
   x^\infty, xu, xy, yv & x^\infty y^\infty u^\infty v^\infty & 1\\
   xu, xy & x^\infty y^\infty u^\infty & 3\\
   xu, yv & x^\infty y u^\infty v & 4\\
   xu, xy, yv & x^\infty y^\infty u^\infty v & 3
 \end{array}
\]
Hence
\begin{multline*}
\mgu_\an([S]_U,x/r(y,y))=\\
\qquad \downclo\{ \emptyset, x^\infty y^\infty,  x^\infty y^\infty v^\infty, x^\infty y^\infty u^\infty,  x^\infty y^\infty u^\infty v^\infty,  x^\infty y u^\infty v, x^\infty y^\infty u^\infty v\}  \mathproofbox \hspace{0.5cm}
\end{multline*}

\end{example}

% We prove (Theorem~\ref{th:algoandy} in Appendix \ref{sec:proofs}) that,
% if $\vars(x/t) \subseteq U$, then $\mgu'_\an$ is the optimal abstraction of $\mgu_\omega$.
The main difference between the algorithm in Theorem \ref{th:mguandy} and the characterization in Theorem \ref{th:pre-mguandy} is that in the former  $X$ is a subset of $S'$ while, in Theorem~\ref{th:pre-mguandy}, $Y$ is a multiset over
$S'$. Since the number of subsets of $S'$ is finite, the characterization in Theorem \ref{th:mguandy}  is an algorithm.

Obviously, a direct implementation of $\mgu_\an$ would be very slow, so  that appropriate data structures and procedures should be developed for a real implementation. Although this is mostly out of the scope of this paper, we show here that the definition of $\mgu_\an([S]_U,x/t)$ may be modified to consider only \emph{maximal} subsets of $\relev(S,x,t)$. This should help in reducing the computational complexity of the abstract operator.

Given $[A]_U \in \ShLinp$, let $\max A$ be the set of maximal elements
of $A$, \ie $\max A=\{a\in A~|~\nexists b\in A. b >_\an a\}$. Given
a sharing group $o$, we define  the \emph{linearized} version of
$o$, denoted by $l(o)$, as
\[
  l(o)(v)=\begin{cases}
    1 & \text{if $v \in \supp{o}$} \enspace ,\\
    0 & \text{otherwise} \enspace .
  \end{cases}
\]
The linearization operator $l$ is extended pointwise to sets of
sharing groups. We show that instead of choosing $X$ as a subset
of $S'$ in the definition of $\mgu_\an$, we may only consider those $X$'s
which are subsets of $\max S'$.

\begin{theorem}
\label{th:closedmgu}
Given $[S]_U \in \ShLinp$ and the binding $x/t$ with $\vars(x/t) \subseteq U$, we have
\[
    \mgu_\an([S]_U,x/t)  = [(S\setminus S') \cup
     \downclo \bigcup_{X \subseteq \max S'} (\res(X,x,t) \cup
     \res'(X,x,t))]_U \enspace ,
\]
where $S'=\relev(S,x,t)$ and
\[
\res'(X,x,t)=\begin{cases}
   \{\Andybin X^2\} & \text{if $X=X_{xt}$ and $l(X)$ is linear for $t$} \enspace ,\\
   \emptyset & \text{otherwise} \enspace .
\end{cases}
\]
\end{theorem}

The next examples compare our optimal abstract unification operator to the original one and show the increase in precision.
\begin{example}
\label{ex:king1}
Let $U=\{u,v,w,x,y\}$. Consider the set of $2$-sharing groups $S=\{\emptyset,xu,xv,xw,y\}$. We compute $\mgu_\an([S]_U,x/r(y,y))$. Since  $\relev(S,x,r(y,y))=S$, we need to consider any $X\subseteq S$.
If $y\notin X$ then clearly $\res(X,x,r(y,y))=\emptyset$. If $y\in X$, since $\chi_M(y,r(y,y))=2$, it follows that $X$ is linear for $x$ and not strongly non-linear for $r(y,y)$.
Thus
\[\mgu_\an([S]_U,x/r(y,y))=[\downclo\{\emptyset,x^\infty u^\infty y,x^\infty uvy,x^2uwy,x^\infty v^\infty y,x^\infty vwy,x^\infty w^\infty y\}]_U\]
On the other hand, computing with the unification algorithm given in \cite{King94}, the result is
\begin{multline*}
\downclo\{\emptyset,x^\infty u^\infty y,x^\infty u^\infty v^\infty y,x^\infty u^\infty w^\infty y,\\
x^\infty v^\infty y,x^\infty v^\infty w^\infty y,x^\infty w^\infty y,x^\infty u^\infty v^\infty w^\infty y\}
\enspace .
\end{multline*}
The old algorithm is not able to infer the linearity which arises when combining two distinct sharing groups from $\{xu,xv,xw\}$ with $\{y\}$. Moreover, it does not assert that the variables $u,v,w$ cannot share a common variable.
\exproofbox
\end{example}

\begin{example}
\label{ex:king2}
Let $U=\{u,x,y,z\}$ and  $S=\{\emptyset,xu,xy,yz\}$. By computing $\mgu_\an([S]_U,x/r(y))$ we obtain $\downclo\{\emptyset,x^\infty y^\infty ,x^\infty uy^\infty z\}$, which shows that $u$ and $z$ are linear after the unification.
This is not the case when computing with the unification algorithm in \cite{King94}, since we obtain $\downclo\{\emptyset,x^\infty y^\infty ,x^\infty u^\infty y^\infty z^\infty ,x^\infty u^\infty y^\infty ,x^\infty y^\infty z^\infty \}$.  Note that, we also improve the groundness information. In fact, in our result, groundness of $u$ implies groundness of $z$.
\exproofbox
\end{example}
Both examples show the increased precision \wrt King's algorithm. In the first example, we obtain optimality thanks to the introduction of the notion of \mbox{(non-)} strong non-linearity. In the second example, we improve the result since we do not need to consider independence between $x$ and $t$, in order to exploit linearity information.

\subsection{Unification for multi-binding substitutions}\label{sec:multibind}

The unification operator on $\ShLinp$ has been defined for single binding substitutions. It is possible to extend this
operator to multi-binding substitutions in the obvious way, namely by iterating the single binding operators.
\[
  \mgu_\an([S]_U,\{ x/t \} \uplus \theta)=
  \mgu_\an(\mgu_\an([S]_U,x/t),\theta) \enspace .
\]
However, defined in such a way, $\mgu_\an$ is not optimal. Consider, for example, $S=\{ \emptyset, xz,yw \}$, $U=\{x,y,z,w\}$,
and the substitution $\theta=\{ x/r(y,y), z/w \}$. We have that $\mgu_\an([S]_U,x/r(y,y))=[\downarrow \{
\emptyset, x^\infty z^\infty yw \}]_U$. Since $x^\infty zyw \leq_\an x^\infty z^\infty yw$,
by applying the third case of $\mgu_\an$ to $Y=\{x^\infty zyw\}$ we get
\[
\mgu_\an([\downarrow \{ \emptyset, z^\infty
x^\infty yw \}]_U,z/w)=[\downarrow \{\emptyset, x^\infty y^\infty z^\infty
w^\infty\}]_U \enspace .
\]
However,
\[
  \begin{split}
    &\alpha_\an(\mgu_\lp(\gamma_\an([\{ \emptyset, xz,yw \}]_U,\theta)))\\
    =\ &\alpha_\an(\mgu_\lp([\{ xz,yw \}]_U,\theta))\\
    =\ &\alpha_\an(\mgu_\lp([\{ wx^2yz^2 \}]_U,\{z/w\}))\\
    =\ &\alpha_\an([\emptymulti]_U)=[\emptyset]_U \enspace ,
  \end{split}
\]
which shows that $\mgu_\an$ is not optimal. Note that, we do not use optimality of $\mgu_\lp$ to prove this result, since correctness is enough.

 The problem is that, to be
able to conclude that the unification of $S$ with $\theta$ is ground,
we need to keep track of the fact that, after the first binding, $w$
is linear and $z$ is definitively non-linear. Since $\ShLinp$ is
downward closed, we are not able to state this property. Note that, in
the case we have presented here, by changing the order of the bindings
we get an optimal result in $\ShLinp$, but this happens just by
accident.
%\marginpar{mostrare un esempio in cui ci\`o non accade?}

Now, consider the substitution $\theta=\{x/r(y,....,y), z/s(y,...,y), u/v\}$
with $S=\{\emptyset, xu, zv, y\}$ and $U=\{u,v,x,y,z\}$. Assume
that $r(y,....,y)$ is an n-ary term, $s(y,...,y)$ is an m-ary term with $n \neq m$
and $n,m \geq 2$.  We have that:
\[
  \begin{split}
    &\mgu_\an([S]_U,x/r(y,....,y)) = [\downclo \{\emptyset,
        x^\infty u^\infty y, zv \}]_U \enspace ,\\
    & \mgu_\an([\downclo \{\emptyset, x^\infty u^\infty y, zv \}]_U , z/s(y,...,y) )=
    [\downclo \{\emptyset, x^\infty u^\infty z^\infty v^\infty y \}]_U \enspace ,\\
    & \mgu_\an([\downclo \{ \emptyset, x^\infty u^\infty z^\infty v^\infty y \}]_U , u/v)=
           [\downclo \{ \emptyset, x^\infty u^\infty z^\infty v^\infty y \}]_U \enspace .
  \end{split}
\]
On the other hand, we have that:
\[
  \begin{split}
    &\alpha_\an(\mgu_\lp(\gamma_\an([\{\emptyset, xu,zv,y \}]_U,\theta)))\\
    =\ &\alpha_\an(\mgu_\lp([\{ xu, zv, y \}]_U,\theta))\\
    =\ &\alpha_\an(\mgu_\lp([\{ x^n u^n y, zv \}]_U,\{ z/s(y,...,y), u/v   \}))\\
    =\ &\alpha_\an(\mgu_\lp([\{ x^n u^n y z^m v^m \}]_U,\{ u/v \}))\\
    =\ &\alpha_\an([\emptymulti]_U)=[\emptyset]_U \enspace .
  \end{split}
\]
However, if $n=m$, we have:
\[
\begin{split}
&\alpha_\an(\mgu_\lp(\gamma_\an([\{ \emptyset, xu,zv,y \}]_U,\theta)))\\
=\ &\alpha_\an([\{ \emptymulti \} \cup \{x^{kn} u^{kn} y^k z^{kn} v^{kn} \mid k \in \Nat \}]_U)\\
=\ &[\downclo \{\emptyset,x^\infty u^\infty z^\infty v^\infty y \}]_U \enspace .
\end{split}
\]
In this case, keeping track of the variables which are definitively non-linear does not
help. It seems that, in order to compute abstract unification one binding at a time,
we need to work in a domain which is able to keep track of the exact multiplicity of
variables in a sharing group. Actually, this is how $\Linp$ works. Obviously, we could try
to develop a different algorithm for unification in $\ShLinp$ which directly works with
multi-binding substitutions. However, since the algorithm for single binding substitutions
is already quite complex, we think this is not worth the effort.

\subsection{The domain $\Sharing \times \Lin$}
\label{sec:practical2}

The reduced product $\ShLin=\Sharing \times \Lin$  has been used
for a long time in the analysis of aliasing properties, since it was
recognized that the precision of these analyses could be greatly
improved by keeping track of the linear variables. Among the papers
which consider the domain $\ShLin$, we refer to \cite{HW92} and
\cite{HillZB04}. Actually, these papers also deal with freeness
properties, which we do not consider here, to further improve precision.
Although the domain $\ShLin$ has been used for many years, the
optimal unification operator is as yet unknown, even for
a single binding substitution. We provide here a new abstract
operator for $\ShLin$, designed from the abstract unification for
$\ShLinp$, and we prove that it is optimal for single binding
substitutions.

The domain $\ShLin$ keeps track of linearity by recording, for each object of $\Sharing$, the set of linear variables. Each element is now a triple: the first component is an object of $\Sharing$, the second component is an object of $\Lin$, that is, the set of variables which are linear in all the sharing groups of the first component, and the third component is the set of variables of interest.
It is immediate that $\ShLin$ is an abstraction of $\ShLinp$ (and thus of $\Linp$). In the following,
we briefly recall the definition of the abstract domain and provide the abstraction function from $\ShLinp$.
\begin{equation*}
  \ShLin = \{ [S,L,U]\mid S \subseteq\wp(U), (S \neq\emptyset \Ra
  \emptyset \in S), L \supseteq U \setminus \vars(S), U \in \fwp(\var)
  \} \enspace ,
\end{equation*}
with the approximation relation $\leq_\shl$ defined as $[S,L,U]
\leq_\shl [S',L',U']$ iff $U=U'$, $S \subseteq S'$, $L \supseteq L'$.
There is a Galois insertion of $\ShLin$ into $\ShLinp$ given by the
pair of maps:
\begin{align*}
  \alpha_\shl([S]_U)&=[\{ \supp{o} \mid o \in S \},\{ x \in U
  \mid \forall o \in S.\ o(x) \leq 1\},U] \enspace ,\\
  \gamma_\shl([S,L,U])&=[\{ B_L \mid B \in S \}]_U \enspace ,\\
  \intertext{where $B_L$ is the 2-sharing group which has the same support of $B$, with linear variables dictated by the set $L$. In formula:}
  B_L&=\lambda v \in \var. \begin{cases}
                               \infty & \text{if $B \in U \setminus L$,}\\
				1 & \text{if $B \in L$,}\\
                                0 & \text{otherwise.}
                       \end{cases}
\end{align*}

The functional composition of $\alpha_\lp$, $\alpha_\an$ and $\alpha_\shl$ gives
the standard abstraction map from substitutions to $\ShLin$.
We still use the
polynomial notation to represent sharing
groups, but now all the exponents are fixed to one.  Note that the
last component $U$ in $[S,L,U]$ is redundant since it can be retrieved
as $L \cup \vars(S)$.  This is because the set $L$ contains all the
ground variables.

\subsection{Abstract unification for $\Sharing \times \Lin$}

In order to obtain a correct and optimal abstract unification over $\ShLin$,
the trivial way is to directly compute $\alpha_\shl(\mgu_\an(\gamma_\shl([S,L,U]),x/t) )$.
However, we prefer to give an unification operator similar
to the other operators
for $\ShLin$ in the literature \cite{HoweK03,BagnaraZH05TPLP,HillZB04}.
As for the domain $\ShLinp$, we now provide the notions of multiplicity and linearity over $\ShLin$.

Given a set $L$ of linear variables,
we define the maximum multiplicity of a sharing group $o$ in a term $t$ as follows:
\[
  \chi^L_M(o,t)=
  \begin{cases}
    \sum_{v \in o} \occ(v,t) & \text{if $o\cap\vars(t) \subseteq L$}\\
    \infty & \text{otherwise}
  \end{cases}
\]
According to the similar definition for 2-sharing groups, given
$[S,L,U] \in \ShLin $, we say that $(S,L)$ is \emph{linear} for a term $t$
when for all $o\in S$ it holds that $\chi_M^L(o,t)\leq 1$.  Note that,
when $t$ is a variable, the definition boils down to check
whether $t\in L$.

Given $X\subseteq \relev(S,x,t)$, we fix the set $L$ of linear variables and
partition $X$ in three subsets $X_x = \{o\in X~|~\chi_M^L(o,t)=0\}$, $X_t = \{o\in
X~|~\chi_M^L(o,x)=0\}$ and $X_{xt} = X \setminus (X_x \cup X_t)$. Moreover,
we need to define the following subsets of $X$:
\[
\begin{array}{rlrl}
  X_t^{=\infty} &= \{ B\in X_t~|~\chi_M^L(B,t) = \infty \}, &
  \hspace{0.8cm} X_t^{\in \Nat} &= \{ B\in X_t~|~\chi_M^L(B,t) \in \Nat \}, \\
  X_t^{=1} &= \{ B\in X_t~|~\chi_M^L(B,t) =1 \}, &
  X_t^{>1} &= \{ B\in X_t~|~\chi_M^L(B,t) >1 \}, \\
  X_{xt}^{=1} &= \{ B\in X_{xt}~|~\chi_M^L(B,t) =1 \}, &
  X_{xt}^{>1} &= \{X_{xt}~|~\chi_M^L(B,t) >1 \}.
\end{array}
\]
Since we do not deal with definite linearity, we need to take into account
the sharing groups which can be obtained by linearizing variables. This may be accomplished by using the set $U$ instead of $L$ when computing the multiplicity. We denote by $X_{xt}^{U}$ the set
\[
  X_{xt}^{U} = \{ B\in X_{xt}~|~\chi_M^U(B,t) =1 \} \enspace ,
\]
which corresponds to the \textit{linearizable} sharing groups.

Moreover, given sets $A_1,\ldots,A_n$ with $n \geq 2$ we denote by
$\bin(A_1,\ldots,A_n)$ the set $\{ \bigcup\{a_1,\ldots,a_n\}~|~a_1\in
A_1,\ldots, a_n\in A_n\}$, by $A^*$ the set $\{\bigcup B~|~B\subseteq
A\}$ and by $A^+$ the set $\{\bigcup B~|~B\subseteq A, B\neq \emptyset
\}$. This notation slightly deviates from most of other
literature on $\Sharing$, where $A^*$ does not include the empty
set.  We prefer to adopt a double notation, namely, $A^*$ and $A^+$,
which is more standard in the rest of the research community.

\begin{definition}[Abstract unification algorithm for\/ $\ShLin$]
\label{def:mgushl}
Given $[S,L,U] \in \ShLin$ and the  binding $x/t$ such that $\vars(x/t) \subseteq U$, we define
\[
\mgu_\shl([S,L,U],x/t) = [ (S\setminus X) \cup K, U' \cup L', U] \enspace ,
\]
where $X=\relev(S,x,t)=\{ B \in S\mid B \cap \vars(x/t) \neq \emptyset\}$ and
$U'=U\setminus\vars( (S\setminus X) \cup
K)$. Here, $K$ is the set of new sharing groups created by the
unification process and $U'$ is the set of variables which do not
appear in any sharing group of the result, i.e. the set of ground
variables. $K$ is defined as follows:
\begin{align}
\intertext{$\bullet$~If $x \in L$:}
\label{eq:sl1}
\begin{split}
K=\ & \bin(X_t^{=\infty}, X_x^+, X_{xt}^* ) \cup \\
  & \bin(X_t\cup \{ \emptyset \}, X_{xt}^{>1}, X_x^+, X_{xt}^*) \cup \\
  & \bin ( \{ \{o\} \cup (\cup Z)~|~ o\in  X_t^{\in \Nat}, Z
  \subseteq  X_x, 1 \leq |Z| \leq \chi_M^L(o,t)\}, (X_{xt}^{=1})^* )
  \cup \\
  & (X_{xt}^{U})^+\enspace .
\end{split}\\
\intertext{$\bullet$~If $x \notin L$:}
  \label{eq:sl2}
\begin{split}
K=\ & \bin( X_t^{>1}\cup X_{xt}^{>1} , X_x \cup X_{xt},X^*)  \cup  \\
  & \bin((X_t^{=1})^+, X_x \cup X_{xt}^{=1}, (X_{xt}^{=1})^* )\cup\\
  & (X_{xt}^{=1})^+ \enspace .
\end{split}
\end{align}
Finally, the set $L'$ of linear variables which are not ground is
\begin{equation}\label{eq:linsharinglin}
  L'=
  \begin{cases}
    L \setminus (\vars(X_x\cup X_{xt}) \cap
    \vars(X_t\cup X_{xt})) & \text{if $(S,L)$ is linear for $x$ and $t$,}\\
    L \setminus \vars(X_x \cup X_{xt}) & \hspace*{-1cm} \text{otherwise, if $(S,L)$ is linear for $x$,}\\
    L \setminus \vars(X_t \cup X_{xt}) & \hspace*{-1cm}\text{otherwise, if $(S,L)$ is linear for $t$,}\\
    L \setminus \vars(X) & \hspace*{-1cm} \text{otherwise.}
  \end{cases}
\end{equation}
\end{definition}

\begin{theorem}[Optimality of $\mgu_\shl$]
  \label{th:optimal-shl}
  The operator $\mgu_\shl$ in Definition~\ref{def:mgushl} is correct and optimal \wrt $\mgu$, when $\vars(x/t)\subseteq U$.
\end{theorem}
% 
% Theorem~\ref{th:optimal-shl} in Appendix \ref{sec:proofs} proves that $\mgu_\shl$ is indeed
% the optimal abstract unification for $\ShLin$.

\begin{example}\label{rational}
  Let $S=\{\emptyset, xv,xy,zw\}, L=\{v,w,x,y\}, U=\{v,w,x,y,z\}$ and consider the binding $x/f(y,z)$. It is easy to check that $(S,L)$ is linear for $x$ but not for $t$. Applying our operator, we obtain $\mgu_\shl([S,L,U],x/f(y,z))=[S',L',U]$ with $S'=\{\emptyset,xy,vwxyz,vwxz\}$ and $L'=\{w\}$. This is more precise that the operators for $\Sharing \times \Lin$ in \cite{HW92}. Actually, even using the optimizations proposed in \cite{HoweK03,HillZB04}, one obtains as result the object
  \[
  [\{vxy,vwxz,xy,wxyz,vwxyz\},\{w\},U] \enspace .
  \]
  The optimization proposed in \cite{BagnaraZH05TPLP} is not applicable as it is, since it requires $\vars(\relev(S,x))$ and $\vars(\relev(S,f(y,z)))$ to be disjoint. Even assuming that this test for independence may be removed as unnecessary, the final result would be the same as above. In both cases, our operator is able to prove that $vxy$ and $wxyz$ are not possible sharing groups.

  Note that, in a domain for rational trees, the sharing group $vxy$ is needed for correctness, since the unification of $\{x/f(f(v,y),c), z/w\}$ with the binding $x/f(y,z)$ succeeds with $\{ x/f(f(v,y),c), z/c, w/c, y/f(v,y)\}$. This means that we are able to exploit the occur-check of the unification in finite trees.  As a consequence, our abstract unification operator is not correct w.r.t.~a concrete domain of rational substitutions \cite{King00jlp}.
%
%   Our results improve over the abstract unification operators of the domains in the literature even in some cases which do not involve the occur-check. For example, if $S=\{\emptyset, xu,xv,xy,z\}$, $L=\{u,v,x,z\}$ and given the binding $x/f(y,z)$, we are able to state that $uvxz$ is not a member of $\mgu_\shl([S,L,\vars(S)],\{ x/f(y,z) \})$, but the domains in \cite{King01tr,BagnaraZH05TPLP,HillZB04} cannot.
\exproofbox
\end{example}

An alternative would be to compute the abstract unification following Theorem \ref{th:mguandy} with $\chi_M$ and $\andybin$ replaced by $\chi^L_M$ and $\cup$ respectively (we can obviously ignore the delinearization operator $(\_)^2$ since $B \cup B=B$). However, we do not pursue further this approach.

In the case $\vars(x/t)\nsubseteq U$, we may proceed as for $\Linp$ and $\ShLinp$: enlarge the set of variables of interest in order to include all $\vars(x/t)$ and compute unification with $\mgu_\shl$.

\begin{definition}[Abstract unification algorithm with extension in $\ShLin$]
\label{def:mgushl2}
Given $[S,L,U] \in \ShLin$ and the binding $x/t$, let $V=\{v_1, \ldots, v_n\}$ be $\vars(x/t) \setminus U$. We define:
\[
  \mgu_\shl([S,L,U],x/t)= \mgu_\shl([S \cup \{ v_1, \ldots, v_n \},L  \cup V, U \cup V],x/t) \enspace .
\]
\end{definition}
\begin{theorem}[Optimality of $\mgu_\shl$ with extension]
  \label{th:optimal-shl2}
  The operator $\mgu_\shl$ in Definition \ref{def:mgushl2} is the optimal abstraction of $\mgu$.
\end{theorem}

Although the abstract operator $\mgu_\shl$ is optimal for the
unification with a single binding, the optimal operator for a
multi-binding substitution cannot be obtained by considering one
binding at a time. This is a consequence of the fact that the
corresponding operator for single binding unification on $\ShLinp$ cannot
be extended to an optimal multi-binding operator by simply
considering one binding at a time.  In fact, all the counterexamples in
Section \ref{sec:multibind} are also counterexamples for $\mgu_\shl$, since
it is the case that $[S]_U = \gamma_\shl(\alpha_\shl([S]_U))$.

\section{Optimal unification in practice}
\label{sec:practice}

In this section, we give some evidence that there are practical advantages in using the optimal unification operators for $\ShLin$. It is far beyond the scope of this paper to provide an experimental evaluation of the new algorithms, but the results in \cite{BagnaraZH05TPLP} give some hints on its possible outcome. \citeN{BagnaraZH05TPLP} introduce an improvement for $\Sharing \times \Lin \times \Free$ exploiting some ideas from King's unification operator for the domain $\ShLinp$. In this way, they improve precision in a few cases and show that efficiency of the analysis  is more likely to be increased than decreased. In fact, even if the final result of the analysis does not change, a more precise operator may reduce the number of sharing groups in the intermediate steps, which helps performance. Hence, we expect the optimal unification for $\ShLin$ to further improve the analysis, both in efficiency and precision. This is more evident if we consider that  \citeN{BagnaraZH05TPLP} measure precision  in terms of the number of independent pairs (as well as definitively ground, free and linear variables) and do not consider set-sharing. However, \citeN{Buenod04} show that set-sharing information may be useful in several application of the analysis, such as parallelization of logic programs. Hence, a greater improvement in precision is to be expected if we consider the full set-sharing property.

We now provide a concrete example of a simple program where our abstract operators give better results than the operators known in the literature.

\subsection{An example: difference lists}

We work with \emph{difference lists}, an alternative data structure to lists for representing a sequence of elements. A difference list is a term of the kind $A \setminus B$ where $A$ and $B$ are lists, which represents the list obtained by removing $B$ from the tail of $A$. For example, using \PROLOG\  notation for lists,  $[1,2,3,4] \setminus  [3,4]$ represents the list $[1,2]$, while $[1,2,3 | x] \setminus x$ and $[1,2,3] \setminus []$ represent the list $[1,2,3]$. The difference lists whose tail is a variable (such as $[1,2,3 | x] \setminus x$) are mostly useful, since they can be concatenated in constant time. An overview of difference lists may be found in \cite{SterlingS94}.

We define the predicate $\difflist/3$, which translates lists to difference lists and vice-versa. The goal $\leftarrow \difflist(l,h,t)$ succeeds when the difference list $h \setminus t$ represents the standard list $l$. For example, $\difflist([],x,x)$ and $\difflist([1,2,3],[1,2,3 | x],x)$ succeed without any further instantiation of variables. In order to improve the precision of the analysis, we keep head and tail of difference lists in separate predicate arguments. The
code for $\difflist/3$, in head normal form, is the following.
\[
\begin{split}
\difflist(l, h, t) &\leftarrow l=[], h=t.\\
\difflist(l, h, t) &\leftarrow l=[x|l'], h=[x|h'], \difflist(l',h',t).
\end{split}
\]
where $l,l'$ (list), $h,h'$ (head), $t$ (tail) and $x$ are variables. We informally compute the goal-independent analysis of $\difflist$ on the domain $\ShLin$, which gives:
\[
\begin{split}
\seml \difflist \semr ^0 =&[\{\emptyset\},\{l,h,t\},\{l,h,t\}] \enspace ,\\
\seml \difflist \semr ^1 =&[\{\emptyset,ht,hl\},\{l,h,t\},\{l,h,t\}] \enspace , \\
\seml \difflist \semr ^2 =& \seml \difflist \semr ^1  \enspace .
\end{split}
\]
The result of the analysis is not affected by our improved unification operator: the standard mgu for $\ShLin$, as given in \cite{HW92}, yields exactly the same result. Now, suppose we want to analyze the goal $\leftarrow \difflist(l,h,h)$. This corresponds to the goal $\leftarrow \difflist(l,h,t), h=t$ in head normal form. Its semantics may be computed, using our operators, as
\[
  \mgu_{sl}([\{\emptyset,ht,hl\},\{l,h,t\},\{l,h,t\}],h/t)=[\{\emptyset,ht\},\{l\},\{l,h,t\}] \enspace .
\]
By projecting over $l$ and $h$, we get $[\{\emptyset,h\},\{l\},\{l,h\}]$.
Hence, the analysis is able to infer that, upon exiting the goal $\leftarrow \difflist(l,h,h)$, the variable $l$ is ground.

% Actually, the concrete semantics of $\difflist(l,h,h)$ is $[\{ l/[] \}_\{l,h\}]$, since the difference list $h \setminus h$ corresponds to the empty list. Note that the result of our analysis is not the best possible analysis, since we state that $h$ may be non-linear, while in the concrete semantics $h$ is linear. To prove also this property we should add freeness information, which we do not consider here.

By using the standard mgu for $\ShLin$ in \cite{HW92}, we get
\begin{equation}
\label{eq:dl-standard}
   [\{\emptyset,ht,htl\},\{l\},\{l,h,t\}] \enspace ,
\end{equation}
hence $l$ is detected to be linear but not ground. The optimizations introduced in \cite{HoweK03,HillZB04,BagnaraZH05TPLP} do not improve this result. This is a consequence of the fact that these optimizations have been developed to be correct also for rational trees. In this case, you cannot infer that $l$ is ground after $\leftarrow \difflist(l,h,h)$ since the substitution in rational solved form $\{l/[v], h/[v | h]\}$ is a correct answer for the same goal.

If we perform the analysis in $\ShLinp$, using our operators we have $\seml \difflist \semr = [\{ \emptyset, hl, ht \}]_{lht}$ and the result for the goal $\leftarrow \difflist(l,h,h)$ is $[\{ \emptyset, h \}]_{lh}$. However, by using the original operator in  \cite{King94}, the semantics of $\difflist$ does not change, but the result for the goal $\leftarrow \difflist(l,h,h)$ is   $[\{\emptyset,h^\infty  ,h^\infty  l^\infty \}]_{lh}$ thus $l$ is not proven to be either ground or linear.

The fact that optimal operators improve groundness information is somehow surprising. Generally, one expects that groundness affects aliasing analysis, but not vice-versa. In fact, it is well known that $\Sharing$ is a refinement \cite{CFGPR97} of the domain $\Def$. However, as far as groundness is concerned, the precision of $\Sharing$ and $\Def$ is the same, \ie the other objects included in $\Sharing$ do not improve groundness analysis \cite{CFW98}. As far as we know, there is no abstract unification operator in the literature, for a domain dealing with sharing, freeness and linearity, which is more precise that $\Def$ for groundness. On the contrary, the example above shows that  $\ShLin$, endowed with the optimal unification, improves over $\Def$. Amazingly, in this example $\ShLin$ is even better than $\Prop$  \cite{ArmstrongMSS94}. In the latter, the abstract semantics of $\difflist$ is $h \leftrightarrow (l \wedge t)$, \ie $h$ is ground iff both $l$ and $t$ are ground. The result of the analysis for the goal  $\leftarrow \difflist(l,h,h)$ is $\exists_t ~ (h  \leftrightarrow (l \wedge t) \wedge h \leftrightarrow t)$. This is equivalent to $h \rightarrow l$ which does not imply groundness of $l$. Actually, $h \rightarrow l$ is the groundness information which may be inferred by \eqref{eq:dl-standard}.

\subsection{Another example for $\ShLinp$}

As far as we know, there is no implementation or experimental evaluation of the domain $\ShLinp$. We think it would be worthwhile to give such an implementation and that there is some evidence that $\ShLinp$ improves over $\ShLin$ also in practice. For instance, we show a simple program where King's domain is more precise than $\ShLin$ with optimal operators.

We provide a variant of the predicate $\difflist/3$, which we call $\difflist'/2$, with only two arguments: head and tail of the difference list are encoded in the second argument as the term $\mathit{head} \setminus \mathit{tail}$.
\[
\begin{split}
\difflist'(l, d) &\leftarrow l=[], d = h \setminus h.\\
\difflist'(l, d) &\leftarrow l=[x|l'], d=[x|h] \setminus t, d'=h \setminus t, \difflist'(l',d').
\end{split}
\]
We informally compute the goal-independent analysis of $\difflist'$ on the domain $\ShLin$, which gives:
\[
\begin{split}
\seml \difflist' \semr ^0 =&[\{\emptyset\},\{d,l\},\{d,l\}] \enspace ,\\
\seml \difflist' \semr ^1 =&[\{\emptyset,dl,d\},\{l\},\{d,l\}] \enspace ,\\
\seml \difflist' \semr ^2 =& \seml \difflist' \semr ^1 \enspace .
\end{split}
\]
The same analysis, computed over $\ShLinp$, gives
\[
\begin{split}
\seml \difflist' \semr ^0 =&[\{\emptyset\}]_{dl} \enspace , \\
\seml \difflist' \semr ^1 =&[\{\emptyset, dl, d, d^\infty\}]_{dl} \enspace , \\
\seml \difflist' \semr ^2 =& \seml \difflist' \semr ^1 \enspace .
\end{split}
\]

Now, suppose we want to analyze the goal $\leftarrow \difflist'(l,d), d=[x_1,x_2|h]\setminus t$, which extracts the first two elements from the difference list $d$. In $\ShLin$ we have the following
\begin{multline*}
  \mgu_{sl}([\{\emptyset,dl,d\},\{l\},\{d,l\}],d/[x_1,x_2|h]\setminus t)=\\
  [\{\emptyset\} \cup \bin(\{dl, d\},\{x_1,x_2,h,t\}^*), \{l\},\{d,l,x_1,x_2,h,t\}] \enspace .
\end{multline*}
Note that the sharing group $dlx_1x_2$ is part of the result. If we repeat the analysis in $\ShLinp$, we have
\begin{multline*}
  \mgu_{\an}([\{\emptyset,dl, d, d^\infty\}]_{dl},d/[x_1,x_2|h]\setminus t)=
  \big[\{\emptyset, dlx_1, dlx_2, dlh, dlt\} ~ \cup \\
  \downclo \bigl\{ \Andybin X \mid X \in \wp(\{d^\infty x_1^\infty, d^\infty x_2^\infty,d^\infty h^\infty,d^\infty t^\infty\}) \bigr\}\big]_{dlx_1x_2ht} \enspace .
\end{multline*}
This result does not contain the sharing group $dlx_1x_2$.

Generally speaking, it is easier to analyze the predicate $\difflist/3$ than $\difflist'/2$. \citeN{CodishMT00} propose a method named \emph{untupling} which is able to automatically recover $\difflist/3$ from $\difflist'/2$.

\section{Related work}
\label{sec:related}

In this paper, we work with a concrete domain of substitutions on finite trees. In the literature, some authors deal with rational trees.

  Since any correct operator for rational trees is also correct for finite trees, we can compare the unification operators for rational trees with ours (of course, this is not entirely fair as far as the precision is concerned). The opposite is not true, since an abstract unification operator for finite trees may be able to exploit the occur-check condition. We have shown in Example \ref{rational} that our optimal operator can exploit the occur-check condition, and thus it is not correct for rational trees.

\subsection{$\Sharing$}
It is well-known that the abstract unification operator of the domain $\Sharing$ alone (\ie without any freeness or linearity information) is optimal. \citeN{CF99} give a formal proof of optimality, considering a slightly different unification operator with two abstract objects and a concrete substitution. Since the two abstract objects are renamed apart, it is equivalent to consider a single abstract object. The basic idea underlying the proof is to exhibit, for each sharing group in the result of the unification, a pair of concrete substitutions generating the resulting sharing group. We follow the same constructive schema in the proof of optimality for $\Linp$ (but we look for a single substitution, due to the different concrete operator). Instead,
to prove optimality for $\ShLin$ and $\ShLinp$, we use a direct approach and show that the abstract unification operator corresponds to the best correct abstraction (\ie $\alpha \circ \mgu_\omega \circ \gamma$) of the unification on $\Linp$ with simple (although tedious) algebraic manipulations.

A different unification operator has been proposed in \cite{AmatoS02agp,AmatoS09sharing} for goal-dependent analysis of $\Sharing$. In this paper, the standard unification operator is splitted into two different operators for forward and backward unification. Both operators are proved to be optimal and the overall analysis is strictly more precise than the analysis performed on $\Sharing$ equipped with the standard operator.

As far as we know, these are the only optimality results for domains encoding aliasing properties.

\subsection{$\Sharing \times \Lin$}
In most of the work combining sharing and linearity, freeness information is included in the abstract domain. In fact, freeness may improve the precision of the aliasing component and it is also interesting by itself, for example in the parallelization of logic programs \cite{HermenegildoR95}. In this comparison, we do not consider the freeness component.

The first work which combines set-sharing with linearity is \cite{Langen90}, followed by \cite{HW92}. The initial unification algorithm has been improved by \citeN{HoweK03} and \citeN{HillZB04} by removing an independence test. This increases the number of cases when linearity information may be exploited. \citeN{BagnaraZH05TPLP} propose a different improvement, adopting an idea by \citeN{King94} for  the domain $\ShLinp$, which simplifies the unification of a linear term with a non-linear one. Example \ref{rational} shows that, even adopting all these improvements, we still obtain a strictly more precise operator. Since our operator is optimal, any further improvement is now impossible.

\citeN{BagnaraHZ-ta-TCS} show that, if we are only interested in pair-sharing information, $\Sharing$ is redundant. They propose a new domain $\PSD$ which is obtained by discharging redundant sharing groups. A sharing group $B$ in a set $S$ is redundant if $\card{B} > 2$ and  $\forall x,y \in B.~\exists C \in S.~\{x,y\} \subseteq C \subset B$. Analyses performed with $\PSD$ are shown to be as precise as those performed with $\Sharing$, if only pair-sharing information is required. \citeN{HillZB04} introduce the domain $\PSD \times \Lin \times \Free$.  Example \ref{rational} shows that our operator is still more precise (of course, without considering the freeness component), because of the sharing group $vxy$ which does not appear in $S'$ and is not redundant for $\PSD$. In any case, \citeN{Buenod04} have shown that classical applications of sharing analyses, such as parallelization of logic programs, are able to exploit information which is encoded in $\Sharing \times \Free$ but not in $\PSD\times \Free$.

An alternative presentation of $\Sharing \times \Lin$, based on \emph{set logic programs}, has been introduced by \citeN{CodishLB00}. However, the proposed operators are not optimal, as shown in \citeNP{HillZB04}.

The domain $\ShLinp$ is introduced by \citeN{King94}, which provides correct operators for abstract unification. However, these operators are not optimal, as Examples \ref{ex:king1} and \ref{ex:king2} show.

\subsection{$\ASub$}

An alternative approach to aliasing analysis is to only record sharing between pairs of variables (and possibly linearity and groundness information). The best known domain of this category is $\ASub$, introduced by \citeN{S86} and formalized by \citeN{CDY91}. The domain $\ASub$ is the reduced product of pair-sharing, $\Lin$ and $\Con$ \cite{JS87}, which is the simplest domain for definite groundness. Recently, \citeN{King00jlp} reformulated the proofs in order to work with rational trees. Moreover, King's algorithms are parametric \wrt\ the groundness domain, allowing to replace $\Con$ with more precise domains such as $\Def$ and $\Prop$.

The domain $\Sharing \times \Lin$ is strictly more precise than $\ASub$, since it embeds more groundness information (equivalent to  $\Def$) and set-sharing information. Since our operator for $\Sharing \times \Lin$ is optimal, we are sure that the analyses performed in $\Sharing \times \Lin$ are strictly more precise than those in $\ASub$.

The following is a counterexample to the optimality of the abstract unification in \cite{King00jlp}, in the case of finite trees, when pair sharing is equipped with $\Def$ or $\Prop$.

% Moreover, the abstract unification for $\ASub$, as defined by \citeN{CDY91}, is not optimal. To the best of our knowledge,
% optimal operators for $\ASub$ have never been defined, at least for finite trees (the abstract unification in \citeN{King00jlp}, which is possibly optimal for rational trees, is not optimal for finite trees).

\begin{example}
\label{ex:asub}
Consider the object $\kappa=(x \leftrightarrow y, \{xy\})$ where the first component is a formula of $\Def$ and  $\Prop$ and $\{xy\}$ is the set of pairs of variables which may possibly share. In this domain, linearity information is embedded in the second component in the following way: if $v$ is not linear, then $vv$ must be included in the second component. Thus, both $x$, $y$ and $z$ are linear in $(x \leftrightarrow y, \{xy\})$. We want to unify $\kappa$ with $x/f(y,z)$. By using the algorithm \cite{King00jlp}, we obtain $(y \leftrightarrow x \wedge x \rightarrow z,\{xy,xz,yz,xx,yy\})$. However, in $\Sharing \times \Lin$ we may represent $\kappa$ with $[S,L,U]=[\{xy,z\},\{x,y,z\},\{x,y,z\}]$ and $\mgu_\shl([S,L,U],x/f(y,z))=[\{xy\},\{z\},\{x,y,z\}]$ which proves that $z$ is ground.
\exproofbox
\end{example}

Actually, \citeN{King00jlp} does not state explicitly how to compute the groundness component of the result, although he says that it must be computed before the linearity and pair-sharing components, in order to improve precision. However, it seems safe to assume that the author's intention was to compute the groundness component using the abstract operators already known, and therefore independently from the pair sharing component. This is what makes our operator more precise, since linearity information may help in tracking ground variables when working over finite trees.

\subsubsection{Alternating paths}

The domain $\ASub$ and its derivatives \cite{King00jlp} use the concept of \emph{alternating path}. Alternating paths may seem the counterparts, for pair-sharing, of sharing graphs. We now investigate this idea, and show to what extent this correspondence is faithful.

We call \emph{carrier graph} a special graph defined by a set of equations $E$. Each distinct occurrence of a variable in $E$ is a node. Edges in the carrier graphs can be of two types:
\begin{itemize}
 \item  edges of \emph{type one} between two variable occurrences if the occurrences are on opposite sides of a single equation in $E$,
\item edges of \emph{type two} linking two (distinct) occurrences of the same variable. 
\end{itemize}
An alternating path is a sequence of edges of alternating type over the carrier graph.

Alternating paths in $\ASub$ (and derivatives) are used to prove correctness of the abstract unification operators. For example, they are used to prove Prop. 3.1 in \cite{King00jlp}. Sharing graphs are used in this paper to prove Theorem \ref{th:algebraic}, which is the starting point to prove correctness and optimality of the unification algorithms for $\ShLinp$ and $\ShLin$. However, sharing graphs are also used to compute the abstract unification in $\Linp$. Even if alternating paths are not used, in the literature, for computing abstract unification, they could. For any object of pair-sharing $o$, which is a set of pairs of variables, consider any substitution $\theta$ in the concretization of $o$. Then, the object $o$ is an abstraction of the set of alternating paths in $\theta$. More precisely, it  represents all the paths which start and end with edges of type one, which we call \emph{admissible paths}. They are abstracted by considering only the start and end variables. In order to unify $o$ with the binding $x/t$, we build a carrier graph with all the occurrences of variables in $o$ and $x/t$. For each pair of variables in $o$, we add an edge of type one. We add edges of type one and two for the binding $x/t$, as explained above. Finally, we add all the type two edges between  any occurrence in $x/t$ and any occurrence of the same variable in $o$. We consider all the admissible paths over the graph so obtained. It is not difficult to check that the result of the unification algorithm for pair-sharing in \cite{King00jlp}, without any additional groundness domain, is the set of all the start and end variables for all these admissible paths.

\begin{example}\label{ex:carriergraph}
Let $S=\{xv\}$ be the set of pairs of variables which share, and consider the binding $x/r(y,y)$. We obtain the carrier graph:
 
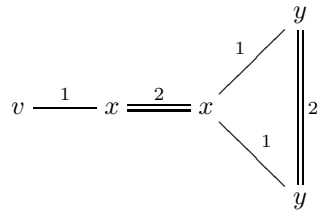
\[
  \xymatrix{
      &&& y \ar@{=}[dd]^{2} \\
      v \ar@{-}[r]^{1} & x \ar@{=}[r]^{2} & x \ar@{-}[ru]^{1} \ar@{-}[rd]^{1} \\
      &&& y
  }
 \]
which gives origin to several alternating paths. Among them, there is an admissible path from $v$ to $v$, which proves that $v$ is not linear after the unification.
\exproofbox
\end{example}

The first difference between alternating paths and sharing graphs is that all the alternating paths are subgraphs of the same carrier graph, while each sharing graph has a different structure, with a different set of nodes. The second difference is that the  information coming from the abstract object and the binding is encoded in a different way. For instance, consider the set $S=\{xy\}$ and the binding $x/z$. We obtain a carrier graph with 4 nodes $x,y,x,z$, two edges $\oneedge{x}{y}$ and $\oneedge{x}{z}$ of type one, and an edge $\twoedge{x}{x}$ of type two. Therefore, the sharing information coming from the initial pair-sharing and the binding is treated symmetrically, and is entirely encoded on the edges. Performing unification on the carrier graph boils down to devising the alternating paths on the graph.
On the contrary, each sharing graph has a set of nodes labelled by $xy$, $x$ and $z$,  with suitable multiplicities. The labels of the nodes encode the initial pair-sharing information, while the binding  affects the multiplicity of nodes.  The process of unification consists of adding the necessary arrows  to get a sharing graph.

If we consider a single alternating path in a carrier graph and the sharing graph for the same pair-sharing information and the same binding, they are obviously related, although not in a straightforward manner. Consider an admissible path and delete all type two edges, collapsing in a single node their start and end nodes (type two edges are used in the carrier graph to avoid the creation of invalid paths, but in a single alternating path they do not add information).  Then, each type one edge coming from the initial pair-sharing information corresponds to a node in the sharing graph, while a type one edge coming from the binding becomes an arrow in the sharing graph.
% Of course, this distinction cannot be derived directly from the alternating path, but we need to consider again the initial information. This is because sharing graphs carry more information about the initial pair-sharing set and binding than type one edge.

\begin{example}
Consider Example \ref{ex:carriergraph}. We depict the (admissible) alternating path from $v$ to $v$, its collapsed version and the corresponding sharing graph.
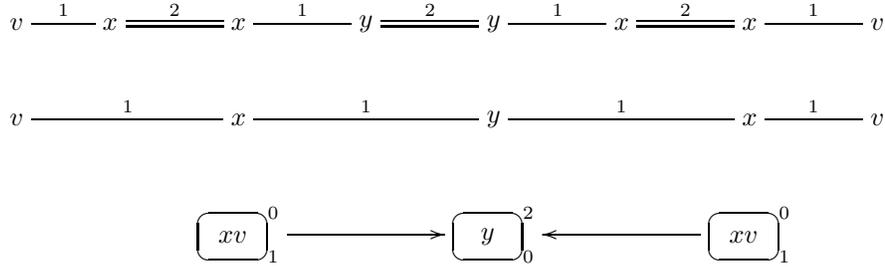
\[
  \xymatrix{v \ar@{-}[r]^{1}  &  x \ar@{=}[r]^{2}  &  x \ar@{-}[r]^{1} & y  \ar@{=}[r]^{2}  &  y \ar@{-}[r]^{1}  &  x \ar@{=}[r]^{2}  &  x \ar@{-}[r]^{1}  & v \\
    v \ar@{-}[rr]^{1}  &&  x  \ar@{-}[rr]^{1} && y   \ar@{-}[rr]^{1}  &&  x  \ar@{-}[r]^{1}  & v \\
     && \ovalee{xv}_1^0  \ar[rr] && \ovalee{y}_0^2  &&  \ovalee{xv}_1^0 \ar[ll] \\
  }
\]
Note that, while in the carrier graph, non-linearity of the variable $x$ is handled by duplicating the variable $y$ which occurs twice, in alternating paths without type two nodes, the duplicated variables are $x$ and $v$, which are connected to $y$. The same holds in the sharing graph, where we have only one node labeled by $y$ and two nodes labelled by $xv$.
\exproofbox
\end{example}

In sharing graphs we also require the multiplicities of a node to be equal to its in- and out-degrees. This makes possible to handle groundness at the same level of sharing and linearity, without requiring a separate domain. Remember that a sharing group $S$ with multiplicity $n$ corresponds, in the concrete domain, to a variable $u$ such that $\theta^{-1}(u)=S$. If the degree of the node labeled with this sharing group is not $n$, this means that one of the occurrences of $u$ is bound to a ground term. This would make ground the entire connected component containing $S$. Hence, in order to correctly and precisely propagate groundness, we just forbid this kind of sharing graphs. On the contrary, the pair-sharing algorithm in \cite{King00jlp}, which focus on a single path in the carrier graph, is not able to extract groundness information without the help of an auxiliary domain.

\subsection{Lagoon and Stuckey's domain}

\citeN{LagoonS02} have recently proposed a different approach to pair-sharing analysis. The authors use multigraphs, called \emph{relation graphs}, to represent sharing and linearity information. The nodes of the multigraph are variables, and two of them may share only if there is a \emph{traversable path} from one variable to the other. Intuitively, each binding generates edges of different types. 
The definition of traversable paths is very similar to that of  alternating paths. A traversable path  is a sequence of edges, such that contiguous edges are always of different types.

This domain should be coupled with a groundness domain, and operators are parametric \wrt the latter one.
The authors show that relation graphs, when coupled with the $\Def$ groundness domain, are more precise than $\Sharing$ and $\ASub$. However, this is not the case for $\Sharing \times \Lin$, at least in the case of finite trees, since the operators in \cite{LagoonS02} are not able to use linearity to improve the precision of the groundness component.

\begin{example}
\label{ex:lagoon-no-better}
As shown in Example \ref{ex:asub}, if we unify  $[S,L,U]=[\{xy,z\},\{x,y,z\},\{x,y,z\}]$ with the binding $x/f(y,z)$, we obtain $\mgu_\shl([S,L,U],x/f(y,z))=[\{xy\},\{z\},\{x,y,z\}]$, proving that $z$ is ground after the unification. In the domains
$\mathrm{\Omega}_\Def$ and $\mathrm{\Omega}_\Prop$ of \cite{LagoonS02}, the abstract object corresponding to $[S,L,U]$ is
\[
  \mu_1=\left( \xymatrix{x \ar@{-}[r] & y}, x \leftrightarrow y \right) \enspace .
\]
% Since the abstraction in $\Def$ of $\theta$ is $x \leftrightarrow y$, we have that $\vars(\theta(x))=\vars(\theta(y))$. Otherwise, it would be possible for $\theta(x)$ to be ground while $\theta(y)$ is not, or viceversa. Moreover, since there are cycles in the relation graph of $\mu_1$, it means that $x$, $y$ and $z$ are bound to linear terms. Therefore
Intuitively, the first element of $\mu_1$ encodes the sharing information, namely, that $x$ and $y$ may share (while $z$ does not share neither with $x$ nor with $y$). The second element of $\mu_1$ is an element of $\Prop$ (and also of $\Def$) and denotes the groundness information that $x$ is ground if and only if $y$ is ground.

The unification of $\mu_1$ with $x/f(y,z)$ in $\mathrm{\Omega}_\Prop$ is realized by abstracting the substitution and composing the two abstract object. The abstraction of $x/f(y,z)$ is
\[
  \mu_2=\left( \vcenterbox{\xymatrix{& y \\x \ar@{-}[ru]
        \ar@{-}[rd] & \\ & z}}, x \leftrightarrow (y \wedge z)
  \right) \enspace ,
\]
The first element says that $x$ shares with both $y$ and $z$, while $y$ and $z$ do not share. The second element says that $x$ is ground if and only if both $y$ and $z$ are ground.

The abstract conjunction is
\[
  \mu_1 \curlywedge \mu_2 = \left( \vcenterbox{\xymatrix{& y \\x \ar@{-}@/^/[ru] \ar@{.}@/_/[ru]
        \ar@{-}[rd] & \\ & z}}, (x \leftrightarrow y) \wedge (x \rightarrow z)
  \right) \enspace ,
\]
where edges drawn in different styles are compatible, namely, that they come from different bindings.
From this result, it is not possible to infer that $z$ is ground after the unification.
%In fact, a possible substitution in the abstract object $\mu_1 \curlywedge \mu_2$ is $\{x/r(z,w,w),y/r(z,w,w,)\}$
\exproofbox
\end{example}
In the actual implementation, \citeN{LagoonS02} use another representation for their domain. Each pair of variables is annotated with a formula denoting the groundness models under which the corresponding pair-sharing may occur. For example, a pair $uv$ annotated with the formula $\bar{u} \wedge \bar{v} \wedge \bar{w} \wedge \bar{z}$ means that $u$ and $v$ may share only if none of $u,v,w,z$ is ground. We conjecture that this domain may be embedded in King's $\ShLinp$. The next example shows how to perform this embedding.

\begin{example}\label{ex:traversablepath}
 We consider the example in Figure $4$ in \cite{LagoonS02}. The variables of interest are $u,v,w,z$.
\[\begin{array}{lcllcl}
 uw &:& \bar{u} \wedge \bar{w} & uz &:& \bar{u} \wedge \bar{z}\\
 vz &:& \bar{v} \wedge \bar{z} & uu &:& \bar{u} \wedge \bar{w}\wedge \bar{z}\\
 uv &:& \bar{u} \wedge \bar{v} \wedge \bar{w}\wedge \bar{z} \qquad\qquad & vv &:& \bar{v} \wedge \bar{w}\wedge \bar{z}\\
 wz &:& \bar{w} \wedge \bar{z} & vw &:& \bar{v} \wedge \bar{w}
\end{array}
\]
For instance, $uv : \bar{u} \wedge \bar{v} \wedge \bar{w}\wedge \bar{z}$ means that $u$ and $v$ may share only if $u,v,w,z$ are not ground, while  $uu : \bar{u} \wedge \bar{w}\wedge \bar{z}$ means that $u$ is (possibly) not linear only if $u,w,z$ are not ground.  Each of these formulas  may be viewed as a condition over $2$-sharing groups. For example $uv : \bar{u} \wedge \bar{v} \wedge \bar{w}\wedge \bar{z}$ means that every $2$-sharing group which contains $u$ and $v$ should also contain $w$ and $z$, while $uu : \bar{u} \wedge \bar{w}\wedge \bar{z}$ means that each $2$-sharing group where $u$ is non-linear should also contain $w$ and $z$. In order to find the object of $\ShLinp$ which corresponds to this example, it is enough to collect all the $2$-sharing groups which satisfy all the conditions enforced by the formulas. In this case, we get $\downclo \{u^\infty v^\infty w z, u^\infty wz, v^\infty wz, uw, vz, wz, uz, vw, u,v,w,z\}$.
\exproofbox
\end{example}

\subsubsection{Traversable paths}

The idea behind traversable paths is very similar to the concept of alternating path and relation graphs are quite similar to carrier graphs. From a carrier graph, we can obtain a relation graph by removing type two edges and introducing a different type of edge for each binding. This works
because the use of non-linear terms is forbidden: a binding like $x / r(y,y)$ has to be replaced by two bindings $x / r(y,z)$ and $y / z$. However, the main difference \wrt traditional pair-sharing (and also $\Linp$) is that \citeN{LagoonS02} do not abstract traversable paths to set of pairs of variables, but they keep in the abstract object the set of all the edges generated during the unification process. In this way, they are able to record that, in order for two variables $x$ and $y$ to share, the only possible path touches another variable $z$.  Hence, if $z$ is ground, $x$ and $y$ cannot share: in this way they recover pair sharing dependence information which would be  lost otherwise.

We could follow the same approach and use multilayer sharing graphs (namely, sets of sharing graphs over the same set of nodes, where each layer represents the unification with a single binding) as abstract objects, without collapsing them to sharing groups. We do not think this would improve precision of the domain very much, since a sharing group is already a much more concrete abstraction of a graph \wrt the set of all the connected pairs of variables.
In fact, already $\Sharing$ can encode the information that, grounding a certain variable $z$, two variables $x$ and $y$ become independent. Moreover, in the Example \ref{ex:traversablepath} we have shown that relation graphs may be encoded into $\ShLinp$.

\subsection{Rational trees}

In the recent years, many authors have studied the behavior of logic programs on \emph{rational trees} \cite{King00jlp,HillZB04}, which formalize the standard implementations of logic languages.  We have shown that our operators, which are optimal for finite trees, are not correct for rational trees, since they exploit the occur-check to reduce the sharing groups generated by the abstract unification (see Example~\ref{rational}).  It would be interesting to adapt our framework to work with rational trees, in order to obtain optimal operators also in this case.
Since a rational tree may contain infinite occurrences of a variable, the notion of $\omega$-sharing group needs to be extended in order to allow infinite exponents. Also, we need to consider infinite sharing graphs (or, at least, a representation of them) and find suitable regularity conditions for them, analogously to the regularity conditions on rational trees.
\begin{example}
Consider the set of $\omega$-sharing groups $S=\{xy,z\}$ and the binding $x/r(z,y)$. On rational trees, unifying $\delta=\{x/y\}$ (such that $[S]_{xyz}  \rightslice [\delta]_{xyz}$) with $x/r(z,y)$ would get the substitution $\{x/r(z,x),y/r(z,y)\}$ in rational solved form. This, intuitively, corresponds to the sharing group $x^\omega y^\omega z$ where the exponent $\omega$ denotes an infinite number of occurrences. A possible (infinite) sharing graph generating this sharing group is the following:
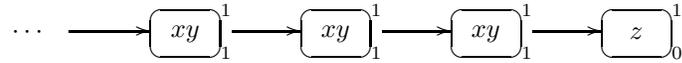
\[
\xybox{
  0 *{\cdots\quad}="a",
  "a"+<\dist,0cm>*{\ovalee{xy}_1^1}="b",
  "b"+<\dist,0cm>*{\ovalee{xy}_1^1}="c",
  "c"+<\dist,0cm>*{\ovalee{xy}_1^1}="d",
  "d"+<\dist,0cm>*{\ovalee{z}_0^1}="e",
  {"a" \ar  "b"},
   {"b" \ar  "c"},
  {"c" \ar  "d"},
  {"d" \ar  "e"},
} 
\]
\exproofbox
\end{example}
\medskip
Although the structure of abstract objects and operators for adapting $\Linp$ to work with rational trees
is more complex, we expect the optimal abstract operators for rational trees on $\ShLinp$ and $\Sharing \times \Lin$ to be simpler than those presented here for finite trees. This is because we do not need to worry about the occur-check condition (embedded in our unification operator) and infinite multiplicities.

\section{Conclusion and Future Works}

We summarize the main results of this paper:
\begin{itemize}
%\item We propose a new domain $\Linp$ as a general framework for
%  investigating sharing and linearity properties. We introduce the
%  notion of \emph{sharing graph} as  a generalization
%  of the concept of alternating path \cite{S86,King00jlp} used
%  for pair sharing analysis and provide optimal abstract
%  operators for $\Linp$.
\item We define a new domain $\Linp$ as a general framework for
  investigating sharing and linearity properties and provide the optimal unification
  operator.
\item We show that $\Linp$ is a useful starting point for studying
  further abstractions. We obtain the optimal operators for
  single binding abstract unification in $\Sharing \times \Lin$ and
  $\ShLinp$, and we show that these are strictly more precise
  than all the other operators in the literature for the same domains.
\item We show, for the first time, an optimality result for a domain which
  combines aliasing and linearity information.
\end{itemize}

Moreover, as a negative result, we prove that the standard schema of the iterative unification algorithm (one binding at a time) does not lead to optimal operators for the domains  $\ShLinp$ and $\Sharing \times \Lin$. As a side result, we show that $\ShLin$ and $\ShLinp$ with optimal operators may be more precise than $\Prop$ for groundness analysis.

Several things remain to be explored: first of all, we need to study the impact on the precision and performance obtained by adopting the new optimal operators and domains. We plan to implement the operators on $\ShLinp$ and $\Sharing \times \Lin$ within the CiaoPP static analyzer \cite{ciao-reference-manual-tr}. Moreover, we plan to analyze the domain $\PSD \times \Lin$ \cite{BagnaraHZ-ta-TCS} in our framework and, possibly, to devise a variant of $\ShLinp$ which enjoys a similar closure property for redundant sharing groups. This could be of great impact on the efficiency of the analysis. Last but not least, we plan to translate our framework to the case of unification over rational trees.

\appendix

\renewcommand\thesection{\Alph{section}}

\section{Proofs of Section \ref{sec:shlinomega}}
\label{sec:mainproof}
In this section we give the proofs of the main results of the paper.

\newtheorem{theoremproof}{Theorem}
\newtheorem{propositionproof}{Proposition}

\renewcommand{\thetheoremproof}{\ref{th:wdapprox}}

\begin{theoremproof}%\label{th:wdapprox}
  The relation $\rightslice$ is well defined.
\end{theoremproof}
\begin{proof}
  It is enough to prove that $\{{\theta_1}^{-1}(v)|_U \mid v \in \var
  \} = \{{\theta_2}^{-1}(v)|_U \mid v \in \var\}$ when $\theta_1
  \sim_U \theta_2$. Assume that $\theta_1 \sim_U \theta_2$, then by
  definition of $\sim_U$ there exists a renaming $\rho$ such that
  $\rho(\theta_1(u))=\theta_2(u)$ for each $u \in U$. Given
  $S=\theta^{-1}_1(v)|_U$, if $w=\rho(v)$ we have
  $\theta_2^{-1}(w)|_U=\theta_1^{-1}(v)|_U=S$. This concludes the
  proof.
\end{proof}

\renewcommand{\thepropositionproof}{\ref{prop:chi}}
\begin{propositionproof}%\label{prop:chi}
  Given a substitution $\theta$, a
  variable $v$ and a term $t$, we have that
  $\chi(\theta^{-1}(v),t)=\occ(v,\theta(t))$. Moreover, given a set of variables $U$, when $\vars(t) \subseteq U$, it holds that $\chi(\theta^{-1}(v)|_U,t)=\occ(v,\theta(t))$.
\end{propositionproof}
\begin{proof}
   Let $B=\theta^{-1}(v)$.
  The proof is by induction on the structure of the term $t$. If $t
  \equiv a$ is a constant, then $\occ(v,\theta(a))=\occ(v,a)=0$ which
  is equal to $\chi(B,a)$ since $\occ(w,a)=0$ for each $w \in \var$.
  If $t \equiv w$ is a variable, then $\occ(v,\theta(w))
  =\theta^{-1}(v)(w)=B(w)$.  At the same time, $\chi(B,t)=B(w)$ since
  $\occ(w,w)=1$ and $\occ(y,w)=0$ for $y \not\equiv w$.  For the
  inductive case, if $t \equiv f(t_1, \ldots, t_n)$, we have
  $\occ(v,t)=\sum_{i=1}^n \occ(v,t_i)= \sum_{i=1}^n \chi(B,t_i)$ by
  inductive hypothesis. Moreover
  \[
\chi(B,t)=\sum_{v \in \supp{B}} (
  B(v) \cdot \sum_{i=1}^{n} \occ(v,t_i))
  = \sum_{i=1}^n \sum_{v \in
    \supp{B}} B(v) \cdot \occ(v,t_i)= \sum_{i=1}^n\chi(B,t_i) \enspace .
  \]
Let $U$ be a set of variables with $\vars(t) \subseteq U$. By definition,
$\chi(\theta^{-1}(v)|_U,t)=\sum_{w \in \theta^{-1}(v)|_U} \occ(w,t)$.
Since $\vars(t) \subseteq U$, for any $w\notin U$ it holds that $\occ(w,t)=0$, and thus
$\chi(\theta^{-1}(v)|_U,t)=\chi(\theta^{-1}(v),t)$.
\end{proof}

\renewcommand{\thepropositionproof}{\ref{prop:thetam1}}
\begin{propositionproof}
%  \label{prop:thetam1}
  Given substitutions $\theta$, $\eta \in \Isubst$ and an $\omega$-sharing group $B$, we have
  \[
    (\eta \circ \theta)^{-1}(B)=\theta^{-1}(\eta^{-1}(B))
    \enspace .
  \]
\end{propositionproof}
\begin{proof*}
  Using the definitions and simple algebraic manipulations, we
  have 
  \begin{eqnarray*}
    && \theta^{-1}(\eta^{-1}(B))\\
%    &=& \theta^{-1}(\lambda v. \chi(B,\eta(v))) \\
    &=& \lambda w. \chi \big( \lambda v. \chi(B, \eta(v)), \theta(w)
    \big)\\
    &=& \lambda w. \sum_{y} \chi(B,\eta(y)) \cdot \occ(y,\theta(w)) \\
    &=& \lambda w. \sum_{y} \left(\sum_{x} B(x) \cdot \occ
    (x,\eta(y)) \right) \cdot  \occ(y,\theta(w))\\
%    &=& \lambda w. \sum_{y} \sum_{x} B(x) \cdot \occ
%    (x,\eta(y))  \cdot  \occ(y,\theta(w))\\
    &=& \lambda w. \sum_{x} B(x) \cdot \sum_{y} \occ
    (x,\eta(y))  \cdot  \occ(y,\theta(w))\\
    &=& \lambda w. \sum_{x} B(x) \cdot \sum_{y} \eta^{-1}(x)(y)
    \cdot  \occ(y,\theta(w))\\
    &=& \lambda w. \sum_{x} B(x) \cdot \chi(\eta^{-1}(x),\theta(w)) \enspace .
  \end{eqnarray*}
  By Prop. \ref{prop:chi}, we have that $\chi(\eta^{-1}(x),\theta(w))=
  \occ(x,\eta(\theta(w))$ and therefore
  \begin{equation*}
    \theta^{-1}(\eta^{-1}(B)) = (\eta \circ \theta)^{-1}(B) \enspace . \mathproofbox
  \end{equation*}
\end{proof*}

\renewcommand{\thetheoremproof}{\ref{thm:correctness_mgu_omega}}
\begin{theoremproof}[Correctness of $\mgu_\lp$]%\label{thm:correctness_mgu_omega}
  The operation $\mgu_\lp$ is correct w.r.t. $\mgu$, \ie
\[
 \forall [S]_U \in \Linp, \delta \in \Isubst.~[S]_U \rightslice [\theta]_{U} \implies
 \mgu_\lp([S]_U,\delta) \rightslice  \mgu([\theta]_{U},\delta) \enspace .
\]

\end{theoremproof}

\begin{proof}
  Given $[S]_U \rightslice [\theta]_{U}$ and $\delta \in \Isubst$,
  we need to prove that $\mgu_\lp([S]_U,\delta) \rightslice
  \mgu([\theta]_{U},\delta)$ or the equivalent property
  $\alpha_\lp(\mgu([\theta]_U,\delta)) \leq_\lp
  \mgu_\lp([S]_U,\delta)$.

  Since $\mgu_\lp$ is defined inductively on the number of
  bindings in $\delta$, it is enough to prove that
  $\mgu_\lp([S]_U,x/t) \rightslice \mgu([\theta]_{U},\{x/t\})$ for a single binding $x/t$. Since composition of correct operators is still correct, it follows that multi-binding unification is correct.

Moreover, when $\vars(x/t) \not\subseteq U$, we exploit the identity $\mgu([\theta]_{U},\{x/t\}) =
  \mgu(\mgu([\theta]_{U}, [\epsilon]_{\vars(x/t)}),\{x/t\})$. When computing
$\mgu([\theta]_U,[\epsilon]_{\vars(x/t)})$ all the variables in
$\vars(x/t) \setminus U$ occurring in $\theta$ are renamed apart
from $x/t$ itself. Therefore each $v \in \vars(x/t) \setminus U$
is free (hence linear) in
$\mgu([\theta]_U,[\epsilon]_{\vars(x/t)})$, \ie
  \[
  \alpha_\lp(\mgu([\theta]_{U}, [\epsilon]_{\vars(x/t)}))= \left[S
    \cup \{ \multil v \multir \mid v \in \vars(x/t) \setminus U
    \}\right]_{U \cup \vars(x/t)} \enspace .
  \]
  Therefore, it is enough to prove that
  $\mgu_\lp([S]_U,x/t) \rightslice \mgu([\theta]_{U},\{x/t\})$ when
  $\vars(x/t) \subseteq U$.  Let $B$ be a
  sharing group in $\alpha_\lp(\mgu([\theta]_U, \{x/t\}))$, we prove
  that $B \in \mgu_\lp([S]_U,x/t)$.

  If $B=\emptymulti$, we consider a multigraph $G$ with only
  one node labelled by $\emptymulti$ and no edges. It is easy to check
  that $G$ is a sharing graph for $S$ (since $\emptymulti \in S$) and
  $x/t$, and that $\res(G)=\emptymulti$. Therefore, in
  the following we consider only the case $B \neq \emptymulti$.

The proof is composed of three parts: first, we look for a (special) substitution $\beta$ obtained by renaming some variables in $\theta$ and such that $\beta$ is still approximated by $S$; second, we define a multigraph $G$ exploiting the variables of $\beta$; third, we show that we can restrict $G$ to a smaller sharing graph whose resultant $\omega$-sharing group is exactly $B$.

  \textbf{First part.} Without loss of generality, we assume that $\dom(\theta)=U$ (this is always possible since, in any class $[\theta]_{U}$, there exists a substitution whose domain is exactly $U$).  Let
  $\theta'=\mgu(\theta,\{x/t\})=\eta \circ \theta$ with
  $\eta=\mgu(\{\theta(x)=\theta(t)\})$ and we have
  $[\theta']_U=\mgu([\theta]_U,[x/t]_U)$. Since
  $\dom(\theta)=U$, we have $\vars(\eta) \cap U=\emptyset$.
  Consider $\eta'$ obtained from $\eta$ by replacing each occurrence
  of a variable in $\rng(\eta)$ with a different fresh variable. This
  means that there exists $\rho \in \subst$ mapping variables to variables such that $\rho(\eta'(x))=\eta(x)$ for each $x \in
  \dom(\eta)$. Namely, we have
  \[
  \rho=\{v_1/v_2 \mid \exists x \in \dom(\eta), \xi \in \Xi \text{ s.t. }
  \eta'(x)(\xi)=v_1 \wedge \eta(x)(\xi)=v_2 \} \enspace .
  \]
  Note that $\rho$ is not a renaming, since it is not bijective.
  We now show that $\beta=\eta' \circ \theta$ has the property that $[S]_U
  \rightslice [\beta]_U$.  For any $C \in
  \alpha([\beta]_U)$, we may distinguish three cases:
  \begin{itemize}
  \item $C=\emptymulti$. In this case $C \in S$ by definition of
    $\Linp$;
  \item $C=\beta^{-1}(w)|_U$ for $w \in \rng(\theta) \setminus
    \dom(\eta)$. In this case $\occ(w,(\eta' \circ
    \theta)(v))=\occ(w,\theta(v))$ for each $v \in \var$, therefore
    $\beta^{-1}(w)|_U=\theta^{-1}(w)|_U \in S$;
  \item $C=\beta^{-1}(w)|_U$ for $w \in \rng(\eta')$. Hence there
    exists $v \in \rng(\theta)$ such that $\occ(w,\eta'(v))=1$ and
    $\occ(w,\eta'(v'))=0$ for each $v' \notin \{v',w\}$.  Hence, for
    each $u \in U$, $\occ(w,\eta'(\theta(u)))=n$ iff
    $\occ(v,\theta(u))=n$ and this implies $C=\theta^{-1}(v)|_U \in
    S$.
  \end{itemize}
  Moreover $\rho(\beta(u))=\theta'(u)$ for each $u \in U$, therefore
  $\theta' \sim_U \rho \circ \beta$.

%%   Since $\theta' \preceq_U \beta \preceq_U \theta$, then
%%   $\theta'=\mgu(\theta, \{x/t\}) \sim_U \mgu(\beta, \{x/t\})$.  In
%%   particular, $\rho(\beta(v))=\theta'(v)$ for each $v \in U$, and
%%   therefore $\rho$ is a solution to the equation $\beta(x)=\beta(t)$.
%%   Note that, by applying standard results of unification theory,
%%   $\beta(x)=\beta(t)$ is equivalent to the set of equations $X=\{ v_1
%%   = v_2 \mid \text{there is a position $p$ such that } \beta(x)(p)=v_1
%%   \wedge \beta(t)(p)=v_2 \}$.

  \textbf{Second part.} Consider the labelled multigraph $G$ such that $N_G=\{ v \mid v
  \in \vars(\beta(U)) \}$, $l_G(v)=\beta^{-1}(v)|_U \in S$ and $E_G=\{
  \xi \mid \beta(x)(\xi) \in \var \}$. Note that if $\beta(x)(\xi) \in
  \var$, then $\beta(t)(\xi) \in \var$, too. Each position $\xi$ in
  $E_G$ is an arrow such that $\src_G(\xi)=\beta(x)(\xi)$ and
  $\tgt_G(\xi)=\beta(t)(\xi)$. Observe that the second condition in
  the definition of  sharing graph for $S$ and $x/t$ is
  satisfied, since $[S]_U \rightslice [\beta]_U$.

  Let us check the third condition.  For each node $v \in N_G$, if
  $\chi(\beta^{-1}(v)|_U,x)=n$ by Prop.  \ref{prop:chi} we have
  $\occ( v,\beta(x))=n$, \ie there are $n$ positions in $\beta(x)$
  corresponding to $v$. Therefore the outdegree of $v$ is $n$.  In
  the same way, we have that $\chi(\beta^{-1}(v)|_U,t)$ is the
  in-degree of $v$.

  \textbf{Third part.} Given $B=\theta'^{-1}(u)|_{U}$, by Prop.~\ref{prop:thetam1} we have
  $B=\beta^{-1}(\rho^{-1}(u))|_{U}$. Since $\theta' \preceq_U \beta
  \preceq_U \theta$, then
  $[\theta']_U=\mgu([\beta]_U,\{x/t\})=[\mgu(\beta,\{x/t\})]_U$.
  Therefore $\rho \circ \beta' \sim_U \theta'=\mgu(\theta, \{x/t\})
  \sim_U \mgu(\beta, \{x/t\}) = \mgu(\beta(x)=\beta(t)) \circ \beta$.
  We call $\delta$ the result of $\mgu(\beta(x)=\beta(t))$, and note
  that $\beta(x)=\beta(t)$ is equivalent to the set of equations $X=\{
  v_1 = v_2 \mid \text{there is a position $\xi$ such that }
  \beta(x)(\xi)=v_1 \wedge \beta(t)(\xi)=v_2 \}$. The relation $\rho
  \circ \beta \sim_U \delta \circ \beta$ means that, if $w_1, w_2 \in
  \beta(U)$ and $\rho(w_1)=\rho(w_2)$ then $\delta(w_1)=\delta(w_2)$.
  The latter implies that there are in $X$ equations of the kind
  $x_1=x_2$, $x_2=x_3$, \ldots, $x_{n-1}=x_n$ with $x_1=w_1$ and
  $x_n=w_2$, \ie that $w_1$ and $w_2$ are connected in the graph $G$.

  Therefore, let $Y=\{ w \mid \rho(w)=u\}=\supp{\rho^{-1}(u)}$. This is
  not empty, since $B \neq \emptymulti$. If
  $\xi$ is an edge such that $\src_G(\xi)\in Y$, then $\tgt_G(\xi) \in
  Y$, since $\beta(x)(\xi)=\beta(t)(\xi) \in X$. The converse also
  holds. Hence, if we restrict the graph $G$ to the set of nodes $Y$,
  we obtain a sharing graph whose resultant $\omega$-sharing group is
  $\biguplus_{w \in Y} \beta^{-1}(w)|_U=\beta^{-1}(\rho^{-1}(u))|_{U}
  =B$.
\end{proof}

\renewcommand{\thetheoremproof}{\ref{th:lpopt}}
\begin{theoremproof}[Optimality of $\mgu_\lp$]
  %\label{th:lpopt}
  The single binding unification $\mgu_\lp([S]_U,x/t)$ is optimal \wrt
  $\mgu$, under the assumption that $\vars(x/t) \subseteq U$, \ie:
\[
\forall B\in\mgu_\omega ([S]_U,x/t)~\exists \delta\in  \Isubst.~[S]_U \rightslice [\delta]_U \text{ and } B\in \alpha_\lp(\mgu([\delta]_U,\{x/t\})) \enspace .
 \]
\end{theoremproof}
\begin{proof}
  Let $X \in  \mgu_\omega(S,x/t)$. By definition of $\mgu_\omega$, there exists a sharing graph
  $\mathcal G$ such that $X \in
  \res(\mathcal G)$. Let $N_{\mG}=\{ n_1, \ldots, n_k \}$.
    We want to define a substitution $\delta$ such that $[S]_U
  \rightslice [\delta]_U$ and $X \in
  \alpha_\lp(\mgu([\delta]_U,\{x/t\}))$.
If $X=\emptymulti$ this is
  trivial, hence we assume that $X \neq \emptymulti$.
The structure of the proof is as follows: first, we define a substitution $\delta$ which unifies with $x/t$; second, we show that $\delta$ is approximated by $[S]_U$, namely, $[S]_U  \rightslice [\delta]_U$; third, we show that $X \in
  \alpha_\lp(\mgu([\delta]_U,\{x/t\}))$.

  \textbf{First part.} We now define a substitution $\delta$ which unifies with $x/t$. For each node $n \in
  N_{\mG}$ we consider a fresh variable $w_{n}$ and we denote by $W$
  the set of all these new variables.

%%  We introduce the following
%%   notation: given a term $t$ we distinguish different occurrences of
%%   the same variable by calling $(y,n)$ the $n$-th occurrence of a
%%   variable $y$ in $t$, where the order is lexicographic. For instance,
%%   a term $f(x,g(y,y,x))$ may be viewed as the term
%%   $f((x,1),f((y,1),(y,2),(x,2)))$.

  For any $y \in U \setminus \{x\}$ we define a term $t_y$ of
  arity $\sum_{n\in N_{\mG}}l_{\mG}(n)(y)$ as follows:
  \[
  t_y= r(\underbrace{w_{n_1},\ldots,w_{n_1}}_{\text{$l_{\mG}(n_1)(y)$
      times}}, \underbrace{w_{n_2},\ldots,w_{n_2}}_{\text{$l_{\mG}(n_2)(y)$
      times}}, \ldots,
  \underbrace{w_{n_k},\ldots,w_{n_k}}_{\text{$l_{\mG}(n_k)(y)$ times}})
  \]
%%   If $y \in \rng(\theta) \setminus U$ we know that there is only one
%%   node $n_y$ such that $y \in l(n_y)$. We define in this case
%%   $t_y=w_{n_y}$. At the end, we have defined a term $t_y$ for each $y
%%   \in (U \cup \vars(\theta)) \setminus \dom(\theta)$.
  We know that there exists a map $f:
  E_{G} \ra \var$ such that, for each variable $y$ and node $n$, the
  set of edges targeted at $n$ and labelled with $y$ by $f$ is
  exactly $l_{\mG}(n)(y) \cdot \occ(y,t)$. Namely, we require
  \begin{equation*}
    \card{ \{ e \in E_{G} \mid f(e)=y \wedge \tgt_{G}(e)=n \}}=
    l_{\mG}(n)(y) \cdot \occ(y,t) \enspace .
  \end{equation*}
  The idea is that each edge targeted at the node $n$ is actually
  targeted at one of the specific variables in $l_{\mG}(n)$. In
  particular, each variable $y \in \supp{l_{\mG}(n)}$ should have
  exactly $l_{\mG}(n)(y) \cdot \occ(y,t)$ edges targeted at it, so
  that the total number of edges pointing $n$ is $\sum_{y \in U}
  l_{\mG}(n)(y) \cdot \occ(y,t)=\chi(l_{\mG}(n),t)$, \ie the
  in-degree of $n$. The map $f$ chooses, for each edge targeted at
  $n$, a variable in $l_{\mG}(n)$ according to the previous idea.

  Now, for each node $n$ and variable $y \in U$, we denote by
  $M_{n,y}$ the set of edges pointing at $y$ in $n$, \ie
  $M_{n,y}=\{ e \in E_{G} \mid \tgt_{G}(e)=n \wedge f(e)=y
  \}$. Thus $M_{n,y}$ may be partitioned in $\occ(y,t)$ sets
  of $l_{\mG}(n)(y)$ elements, denoted by $M_{n,y,\xi}$ such that $\cup
  \{ M_{n,y,\xi} \mid  t(\xi)=y\} = M_{n,y}$.

  We may define some variations of the terms $t_y$ by replacing the
  variables occurring in them with those in the set $M_{n,y,l}$.  In
  particular, for $y \in U \setminus \{x\}$ and any occurrence
  $\xi$ of a variable $y$ in $t$, we define the term $t^{y}_\xi$ of
  arity $\sum_{n\in N_{\mG}} l_{\mG}(n)(y)$ as
  \[
  t_\xi^{y}=r(w(M_{n_1,y,\xi}),w(M_{n_2,y,\xi}),\ldots,w(M_{n_k,y,\xi}))
  \enspace ,
  \]
  where, if $M=\{ e_1, \ldots, e_q \}$, we define $w(M)$ as
  the sequence $w_{n'_1}, \ldots, w_{n'_q}$ where $n'_j=\src_{E_{G}}(e_j)$.

%% When $y \notin U$, there is
%%   only one $n_y$ such that $y \in l(n_y)$ and $M^i_{n_y,y,l}$ is a
%%   singleton for each $l \in [1,\occ(y,t_i)]$.  In this case, we define
%%   \[
%%   t_l^{i,y}=w(M^i_{n_y,y,l})
%%   \]

  Note that $t_y$ and $t_\xi^{y}$ have, in corresponding positions,
  variables related to nodes which are connected through edges.  We
  are now ready to define the substitution $\delta$ in the following
  way:
  \begin{itemize}
  \item $\delta(x)$ is the
    same as $t$ with the difference that each occurrence
    $\xi$ of a variable $y \in t$ is replaced by the term
    $t^{y}_{\xi}$;
  \item for $y \in U \setminus \{x\}$ then $\delta(y)=t_y$;
  \item in all the other cases, i.e $v \notin U$, $\delta(v)=v$.
  \end{itemize}

  \textbf{Second part.} Now we show that $[S]_U \rightslice [\delta]_U$. We need to consider
  all the variables $v \in \var$ and check that $\delta^{-1}(v)|_U \in
  S$. We distinguish several cases.
  \begin{itemize}
  \item If we choose the variable $w_n$ for some $n \in N$, by
    construction $\occ(w_n,t_y)=l_{\mG}(n)(y)$.  Moreover, since $\mG$
    is a sharing graph, there
    are $l_{\mG}(n)(x)$ edges in $E$ departing from $n$ and
    targeted to nodes $m$ such that $\chi(l_{\mG}(m),t) \neq 0$.
    Thus $\sum_{y \in \vars(t), m \in N_{\mG}} |\{e\in
    M_{m,y} | \src_{E_{G}}(e)=n\}| = l_{\mG}(n)(x)$ and
    $\occ(\delta(x),w_n)=l_{\mG}(n)(x)$.  Since for each $v \in U$
    we have that $\occ(\delta(v),w_n)=l_{\mG}(n)(v)$, we obtain the
    required result which is $\delta^{-1}(w_n)|_U=l_{\mG}(n) \in S$.
  \item If we choose a variable $v \in U$ then $v \in \dom(\delta)$
    and $\delta^{-1}(v)=\emptymulti \in S$.
  \item Finally, if $v \notin U \cup W$, then $\delta^{-1}(v)=\multil
    v \multir$ and $\delta^{-1}(v)|_{U}=\emptymulti \in S$.
  \end{itemize}

  \textbf{Third part.} We now show that $X \in \alpha_\lp(\mgu([\delta]_U,\{x/t\}))$. By
  definition of $\mgu$ over $\Isubst_\sim$, we have that
  $\mgu([\delta]_U,\{x/t\})=[\mgu(\delta,\{x/t\})]_U$. We obtain:
  \begin{equation}
    \label{eq:opt-proof1}
    \begin{split}
      \eta = & \mgu(\delta,\{x/t\}) =\\
      &\{x/t\} \circ \mgu \bigl(
      \{ y=t_y \mid y \in U \setminus \{x\} \} \cup
      \{ y=t_\xi^{y} \mid t(\xi)=y\}\bigr) = \\
      &\{x/t\} \circ \{ y/t_y \mid y \in U \setminus \{x\} \} \circ
      \mgu\{t_y=t^{y}_\xi ~|~ t(\xi)=y \} \enspace .
    \end{split}
  \end{equation}

  Let $F$ be the set of equations $\{t_y=t^{y}_j ~|~ t(j)=y \}$.  We show
  that, for any edge $n \ra m \in E_{G}$, it follows from $F$ that
  $w_n=w_m$.  Since $n \ra m \in E_{G}$, then for some $y\in
  \vars(t)$ it holds that $f(n\ra m)=y$.  This implies that $n \ra m \in
  M_{m,y}$ and therefore there exists a position $\xi$ such that
  $n\ra m \in M_{m,y,\xi}$. By definition of $t_\xi^{y}$,
  it means that $w_n \in \vars(t_\xi^{y})$, in the same position where $w_m$ occurs in $t_y$, hence $w_n=w_m$
  follows from $t_y=t^{y}_\xi \in F$.

  We know that $\mG$ is connected, hence for any $n,m\in N_{\mG}$,
  the set of equations in $F$ implies $w_n=w_m$.  We choose a
  particular node $\bar{n} \in N_{\mG}$ and, for what we said before, we
  have $\mgu(F)=\{ w_{n} / w_{\bar{n}} \mid n \in N_{\mG}\setminus\{\bar{n}\} \}$. We show that
  $\eta^{-1}(w_{\bar{n}})|_{ U}=X$.
  \begin{equation*}
    \begin{split}
      & \eta^{-1}(w_{\bar{n}})|_{ U}\\
      =\ &\{x/t\}^{-1}(\{ y/t_y \mid y \in U \setminus \{x\}\}^{-1}
           (\multil w_{n_1},\ldots,w_{n_k}\multir ))|_{ U}\\
      =\ &\{x/t\}^{-1}(\multil w_{n_1},\ldots,w_{n_k} \multir \multisum
      \lambda y \in U \setminus \{x\}.
      \sum_{n \in N_{\mG}} l_{\mG}(n)(y))|_{ U}\\
      =\ &\lambda y \in U \setminus \{x\}. \sum_{n \in N_{\mG}}
      l_{\mG}(n)(y) \ \multisum  \multil x^{\sum_{y \in \var}
      \occ(y,t) \cdot \sum_{n \in N_{\mG}} l_{\mG}(n)(y)} \multir   \\
      =\ &\lambda y \in U \setminus \{x\}.  \sum_{n \in
        N_{\mG}} l_{\mG}(n)(y) \multisum \multil x^{\sum_{n \in N_{\mG}} \chi(l_{\mG}(n),t)} \multir
       \enspace .
    \end{split}
  \end{equation*}
  Since $\mathcal G$ is a sharing graph, the total out-degree $\sum_{n \in N_{\mG}}
  \chi(l_{\mG}(n),t)$, is equal to the total in-degree
  $\sum_{n \in N_{\mG}} \chi(l_{\mG}(n),x)$. Hence
  \begin{equation*}
    \begin{split}
      & \eta^{-1}(w_{\bar{n}})|_{ U}\\
      =\ &\lambda y \in U\setminus \{x\}.  \sum_{n \in N_{\mG}}
      l_{\mG}(n)(y) \multisum \multil x^{\sum_{n \in N_{\mG}} \chi(l_{\mG}(n),x)} \multir \\
      =\ &\lambda y \in U.  \sum_{n \in N_{\mG}} l_{\mG}(n)(y)\\
      =\ &\res(\mathcal G) \enspace .
    \end{split}
  \end{equation*}
  This concludes the proof.
\end{proof}

\renewcommand{\thetheoremproof}{\ref{th:optimality_extension}}
\begin{theoremproof}[Optimality of $\mgu_\lp$ with extension]%\label{th:optimality_extension}
  The single binding unification $\mgu_\omega$ with extension is optimal \wrt
  $\mgu$.
\end{theoremproof}
\begin{proof}
  Let $S'=S \cup \{ \multil v \multir \mid v \in \vars(x/t)
  \setminus U \}$, $V=U \cup \vars(x/t)$ and $X\in \mgu_\omega(S',x/t)$. We want to find
  $[\delta]_U$ such that $[S]_U \rightslice [\delta]_U$ and $X \in
  \alpha_\lp(\mgu([\delta]_U,\{x/t\}))$.

  Following the previous theorem, we find $\delta$ such
  that $X \in \alpha_\lp(\mgu([\delta]_V,\{x/t\}))$ and $[S']_V
  \rightslice [\delta]_V$. We want to prove that $[S]_U \rightslice
  [\delta]_U$ and $\alpha_\lp(\mgu([\delta]_V,\{x/t\})) \leq_\lp
  \alpha_\lp(\mgu([\delta]_U,\{x/t\}))$, so that $[\delta]_U$ is the existential
  substitutions we are looking for.

  We first show that $[S]_U \rightslice [\delta]_U$. Let $v\in 	\var$. Since
  $[S']_V  \rightslice [\delta]_V$, it follows that  $\delta^{-1}(v)|_V \in
  S'$.
\begin{itemize}
 \item If $ \delta^{-1}(v)|_V \in S$, then $ \delta^{-1}(v)|_V= \delta^{-1}(v)|_U$, since $\vars(S)\subseteq U$, and thus
  $\delta^{-1}(v)|_U \in S$.
 \item  If $ \delta^{-1}(v)|_V \notin S$, then $ \delta^{-1}(v)|_V \in \{ \multil v \multir \mid v \in \vars(x/t) \setminus U \}$. Then $\delta^{-1}(v)|_U = \emptymulti \in S$.
\end{itemize}

  Now we distinguish two cases: either $x\in U$ or $x\notin U$.

  If $x\in U$, with the same considerations which led to
  \eqref{eq:opt-proof1}, we have
  \begin{equation*}
    \begin{split}
      \mgu(\{x/t\},\delta) &= \mgu(\{x=t\} \cup \eq(\delta|_U) \cup
      \eq(\delta|_{V \setminus U}))=\\
      &=\mgu(\{x=t\} \cup \eq(\delta|_U) \cup \{ y=t_y \mid y \in \vars(t)
      \setminus U \}
    \end{split}
  \end{equation*}

%%  We also
%%   split $\{ y=t_y \mid y \in \rng(\theta) \setminus U \}$ in $F_1=\{
%%  y=t_y \mid y \in U_1\}$ and $F_2=\{ y=t_y \mid y \in U_2\}$.

  For each $y \in \vars(t) \setminus U$ there exist a
  position $\xi_y$ such that $t(\xi_y)=y$ and
  $\{x/t\} \cup \eq(\delta|_U) \cup \{y=t_y\}$ is equivalent to
  $\{x/t\} \cup \eq(\delta|_U) \cup \{ t^{y}_{\xi_y} = t_y \}$. Note
  that, since $y \notin U$, then $t_y$ (which is actually $\delta(y)$)
  is linear and independent from $x/t$ and the other bindings in
  $\delta$.  Therefore
  \begin{equation*}
    \begin{split}
      &\mgu(\{x=t\} \cup \eq(\delta|_U) \cup \{ y=t_y \mid y \in \vars(t)
      \setminus U \} \\
      =&\mgu(\{x=t\} \cup \eq(\delta|_U) \cup \{ t^{y}_{\xi_y}=t_y \mid y \in \vars(t)
      \setminus U \} \\
      =&\mgu(\{x=t\} \cup \eq(\delta|_U) ) \uplus \beta'
    \end{split}
  \end{equation*}
  where $\beta'=\mgu(\{ t^{y}_{\xi_y}=t_y \mid y \in \vars(t)
      \setminus U  \})$ and
  $\dom(\beta')=\vars(\{ t_y \mid y \in \vars(t)
      \setminus U \})$. It follows that
  \begin{equation*}
    \begin{split}
      &\alpha_\lp([\mgu(\{x=t\} \cup \eq(\delta|_U) ) \uplus \beta']_V) \\
      =&\alpha_\lp([\mgu(\{x=t\} \cup \eq(\delta|_U) )]_V) \\
      =&\alpha_\lp(\mgu([\delta]_U, \{x/t\})) \enspace .
    \end{split}
  \end{equation*}

If $x\notin U$, then
\begin{equation*}
    \begin{split}
      \mgu(\{x/t\},\delta) &= \mgu(\{x=t\} \cup \eq(\delta|_U) \cup  \eq(\delta|_{\vars(t)\setminus U})\cup \eq(\delta|_{\{x\}})\\
 =& \mgu(\{x=t\} \cup \eq(\delta|_U) \cup \{ y=t_y \mid y \in \vars(t)
      \setminus U \} \cup \\
	&\{t_y=t^y_{\xi}\mid t(\xi)=y\}
    \end{split}
  \end{equation*}
Note that $x$ appears in $S'$
  only in the multiset $\multil x \multir$. Moreover, if $n$ is a
  node labelled by $\multil x \multir$, there is only one edge which
  departs from $n$ and there are no edges which arrive in $n$. This
  means that
  \begin{itemize}
  \item $w_n$ does not appear in any $t_y$ for $y \in V \setminus
    \{x\}$,
  \item $\delta(x)$ is linear since given edges $e \neq e'$, we have
    that $\src_{E_{G}}(e) \neq \src_{E_{G}}(e')$.
  \end{itemize}
  As a result, $\delta(x)$ is linear and does not share variables
  with $x/t$ or the other bindings in $\delta$. The last formula
  may be rewritten as
  \begin{equation*}
    \mgu(\{x=t\} \cup \eq(\delta|_U) \cup \{ y=t_y \mid y \in \vars(t) \setminus U \}) \uplus \beta
  \end{equation*}
  where $\beta$ is a substitution such that
  $\dom(\beta)=\vars(\delta(x))\subseteq W$.  It is
  obvious that
  \begin{equation*}
    \begin{split}
    &\alpha_\lp([\mgu(\{x=t\} \cup \eq(\delta|_U) \cup
    \{ y=t_y \mid y \in \vars(t) \setminus U \}) \uplus \beta]_V)\\
    =\ &\alpha_\lp([\mgu(\{x=t\} \cup \eq(\delta|_U) \cup
    \{ y=t_y \mid y \in \vars(t) \setminus U \})]_V)
    \enspace .
    \end{split}
  \end{equation*}
  since $\dom(\beta) \cap V=\emptyset$.

  Let $U_1=\vars(t)\setminus U$, then
  \begin{equation*}
    \begin{split}
      &\mgu(\{x=t\}\cup \eq(\delta|_U) \cup \eq(\delta|_{U_1}))\\
      =&\delta|_U \circ \mgu(\delta|_U(\{x=t\} \cup
      \eq(\delta|_{U_1})))\\
      =&\delta|_U \circ \mgu(\{x=\delta|_U(t)\}) \cup
      \eq(\delta|_{U_1})))\\
       &\text{[since $\vars(\delta|_{U_1}) \cap \vars(\delta|_U)=\emptyset$
        and $x \not \in \vars(\delta|_{U_1})$]}\\
      =&\delta|_U \circ \{x/\delta|_U(t)\} \circ \delta|_{U_1} \\
      &\text{[since $\{x\}\notin \vars(\delta|_{U_1})$]}
    \end{split}
  \end{equation*}
Note that $\delta|_U \circ \{x/\delta|_U(t)\}$ is $\mgu(\delta|_U,\{x/t\})$.
 We call $\gamma = \delta|_U \circ \{x/\delta|_U(t)\}$ and
  we prove that $\alpha_\lp([\gamma]_V) \geq_\lp \alpha_\lp([\gamma
  \circ \delta|_{U_1}]_V)$.

  Consider a variable $v \in \var$. If $v \notin \vars(\delta|_{U_1})$
  there is nothing to prove. If $v \in \rng(\delta|_{U_1})$ we know
  that $v$ does not occur anywhere else in $\delta|_{U_1}$ and
  $\gamma$.  Then $(\gamma \circ \delta|_{U_1})^{-1}(v)=
  \gamma^{-1}(\multil y, v \multir)= \gamma^{-1}(y) \multisum \multil v
  \multir$ for the unique $y$ such that $v \in \vars(\delta|_{U_1}(y))$.
  Therefore, since $v \notin V$, the sharing group over $V$ we obtain
  in $\gamma \circ \delta|_{U_1}$ from $v$ may be obtained in $\gamma$
  from the variable $y$.  If $v \in \dom(\delta|_{U_1})$ then $(\gamma
  \circ \delta|_{U_1})^{-1}(v)=\emptymulti$ which occurs in every
  element of $\Linp$. 
\end{proof}

\renewcommand{\thetheoremproof}{\ref{th:algebraic}}
\begin{theoremproof}
  %\label{th:algebraic}
  Let $S$ be a set of $\omega$-sharing groups and $x/t$ be a binding. Then $B \in \mgu_\lp(S,x/t)$ iff there exist $n\in \Natzero$, $B_1,\ldots,B_n \in S$ which satisfy the following conditions:
  \begin{enumerate}
    \item $B=\multisum_{1 \leq i \leq n} B_i$,
    \item $\sum_{1 \leq i \leq n} \chi(B_i,x) = \sum_{1 \leq i \leq n}
      \chi(B_i,t) \geq n -1$,
    \item either $n=1$ or $\forall 1 \leq i \leq n.\ \chi(B_i,x) +
      \chi(B_i,t) > 0$.
 \end{enumerate}
\end{theoremproof}
\begin{proof}
  We first prove that these conditions are necessary.  Assume that $B$ is a
  resultant sharing group for $S$ and $x/t$, obtained by the sharing
  graph $G$. We show that there exist a finite set $I$ and, for each $i\in I$, a multiset $B_i\in S$, which satisfy the above conditions.

Take $I=N_G$ and $B_i=l_G(i)$ for each $i \in I$, so that
  $B=\multisum_{i \in I} B_i$. Since then in-degree of each
  node is $\chi(B_i,x)$, the sum of the in-degrees of all the nodes is
  $\sum_{i \in I} \chi(B_i,x)$ and the sum of the out-degree is
  $\sum_{i \in I} \chi(B_i,t)$. Both of them must be equal to the
  number of edges in $E_G$. Moreover, each connected graph with
  $\card{I}$ nodes has at least $\card{I}-1$ edges. Finally, if a
  connected graph has more than one node, then every node $i$ has an
  adjacent edge. Therefore, either $\chi(B_i,x)$ or $\chi(B_i,t)$ is
  not zero.

  Now we prove that the conditions are sufficient. Let $I=\{1,\ldots,n\}$. If
  $n=1$ and $\chi(B_i,x) + \chi(B_i,t) = 0$ for the only $i\in I$,
  simply consider a sharing graph with a single node labelled with
  $B_i$ and no edges. Otherwise, we partition the set $I$ in three
  parts:
   \begin{itemize}
    \item $N_x = \{i\in I~|~\chi(B_i,x)=0 \}$;
    \item $N_t = \{i\in I~|~\chi(B_i,t)=0 \}$;
    \item $N = \{i\in I~|~\chi(B_i,x)\neq 0, \chi(B_i,t)\neq 0 \}$;
  \end{itemize}
  Note that this is a partition of $I$ since, by hypothesis, $\forall
  i \in I.\ \chi(B_i,x) + \chi(B_i,t) > 0$. Now we define a connected
  labelled multigraph $G$ whose sets of nodes is $I$ and whose
  labelling function is $\lambda i\in I. B_i$. In order to define the
  edges, we distinguish two cases.
  \begin{description}
  \item [$N\neq \emptyset$:] Let $N=\{b_1, \ldots, b_m\}$ with $m\geq
    1$ and consider the set of edges:
    \[
    \{a \ra b_1~|~ a \in N_t\} \cup \{b_1 \ra c~|~c \in N_x\}\cup
    \{b_i \ra b_{i+1}~|~i \in \{1,\ldots, m-1\} \} \enspace .
   \]
  \item [$N= \emptyset$:] If $N_t=\emptyset$, then also $N_x=\emptyset$
    and there is nothing to prove. We assume that $N_t\neq \emptyset$, and
    thus $N_x\neq \emptyset$. Let $\bar{a} \in N_t$, $\bar{c} \in N_x$ and
    consider the set of edges:
    \[
    \{\bar{a}\ra c~|~c \in N_x\} \cup \{a\ra \bar{c}~|~a \in N_t
    \setminus \{\bar{a}\} \} \enspace .
    \]
  \end{description}
  Note that, in both cases, we obtain a multigraph with the following
  properties:
  \begin{enumerate}
  \item it is connected;
  \item it has exactly $n-1$ edges, i.e. it is a tree (if we do not
    consider the direction of edges);
  \item there is no edge targeted at a node $i$ with $\chi(i,t)=0$ and
    no edge whose source is a node $i$ with $\chi(i,x)=0$.
  \end{enumerate}
  In the rest of the proof, we call \emph{pre-sharing graph} a
  multigraph which satisfies the above properties.

  If $\indeg(i)$ is the in-degree of a node and $\outdeg(i)$ the
  outdegree, we call \emph{unbalancement factor} of the graph the
  value:
  \begin{multline*}
    \sum \{ \outdeg(i) - \chi(B_i,x) \mid i \in I, \outdeg(i) >
    \chi(B_i,x) \} + \\
    + \sum \{ \indeg(i) - \chi(B_i,t) \mid i \in I,
    \indeg(i) > \chi(B_i,t) \} \enspace .
  \end{multline*}
  We prove that given a pre-sharing graph with unbalancement factor
  $k$, we can build another pre-sharing graph with unbalancement
  factor strictly less then $k$. As a result, there is a pre-sharing
  graph with unbalancement factor equal to zero.

  Assume that the graph has unbalancement factor $k$. There is at least an
  unbalanced node. Assume without loss of generality that the
  unbalanced node is $j$ and that $\outdeg(j) > \chi(B_j,x)$.  Since
  $\sum_{i \in I} \chi(B_i,x) \geq n-1$, there exists a node $l$ such
  that $\outdeg(l) < \chi(B_l,x)$. Let $e$ be the unique edge with
  source $j$ such that, if we remove $e$ from the graph, $l$ becomes
  disconnected from $j$.  Since no edge starts from a node $i$ with
  $\chi(B_i,x)=0$, then $\chi(B_j,x)>0$. This means that $\outdeg(j) >
  1$ and there is at least another edge starting from $j$. Assume that it
  is $e':j \ra j'$. Remove this edge and replace it with an edge $l\ra
  j'$. It is obvious that the result is a pre-sharing graph with a
  smaller unbalancement factor than the original one. The case for
  $\indeg(j) > \chi(B_j,t)$ is symmetric.

  Once the unbalancement factor is zero, since $\sum_{i \in I}
  \chi(B_i,x)= \sum_{i \in I} \chi(B_i,t)$ we can freely add other
  edges in such a way to complete the graph w.r.t. the condition on
  the degree of nodes. We obtain a sharing graph $G$ such that
  $\res(G)=B$.
\end{proof}

\section{Proofs of Section \ref{sec:practical}}
\label{sec:proofs}

In this section we give the proofs of correctness and optimality for the abstract unification operators $\mgu_\an$ and $\mgu_\shl$.

\renewcommand{\thepropositionproof}{\ref{prop:abstraction}}
\begin{propositionproof}
The following properties hold:
\begin{enumerate}
\item $\alpha_\an(\bigmultisum \calS)= \Andybin \alpha_\an(\calS)$.
\item $\relev(\gamma_\an(S),x,t)) = \gamma_\an( \relev(S,x,t)) $.
\end{enumerate}
\end{propositionproof}
\begin{proof*}
We begin by proving the first property.
\begin{align*}
  \begin{split}
    & \alpha_\an(\multisum \multil B_1,\ldots,B_n \multir) \\
    =\ & \alpha_\an \Big(\lambda v \in \bigcup_{1\leq i \leq n} \supp{B_i}.
    \sum_{1\leq i \leq n}B_i(v)\Big) \\
    =\ & \lambda v\in \bigcup_{1\leq i \leq n} \supp{B_i}.
  \begin{cases}
    1  \text{ if $\sum_{1\leq i \leq n}B_i(v)=1$}\\
    \infty  \text{ otherwise}
  \end{cases} \\
  =\ & \multisum \multil o_1,\ldots,o_n \multir \text{ where } o_i =
  \lambda v\in\supp{B_i}.
  \begin{cases}
    1  \text{ if $B_i(v)=1$}\\
    \infty  \text{ otherwise}
  \end{cases} \\
  =\  & \multisum \multil \alpha_\an(B_1),\ldots,
  \alpha_\an(B_n) \multir \enspace .
  \end{split}\\
  \intertext{Now we proceed with the proof of the second property.}
  \begin{split}
    &   \relev(\gamma_\an(S),x,t)) \\
    =\  & \bigcup \{\gamma_\an(o)~|~o\in S, \supp{\gamma_2(o)} \cap \vars(x=t)
    \neq  \emptyset \} \\
    =\ & \bigcup \{\gamma_\an(o)~|~o\in S, \supp{o} \cap \vars(x=t)
    \neq  \emptyset \}
    \qquad \text{(since $\supp{o}=\supp{\gamma_2(o)}$)}\\
    =\ & \gamma_\an( \relev(S,x,t))) \enspace . \mathproofbox
  \end{split}
\end{align*}
\end{proof*}

\renewcommand{\thetheoremproof}{\ref{th:linp-galois}}
\begin{theoremproof}
  %\label{th:linp-galois}
  $\langle \alpha_2, \gamma_2 \rangle : \ShLinp \leftrightharpoons
  \Linp$ is a Galois insertion.
\end{theoremproof}
\begin{proof}
  It is obvious that $\alpha_2$ and $\gamma_2$ are monotone functions
  and that they are both join-morphisms. Extensionality of $\gamma_2
  \circ \alpha_2$ follows from the fact that, given an $\omega$-sharing
  group $B$, we have $B \in \gamma_2(\alpha_2(B))$. Finally, given a 2-sharing
  group $o$, we have $\alpha_2(\gamma_2(o))=\lbrace o \rbrace$. This implies
  that $\alpha_2 \circ \gamma_2$ is the identity.
\end{proof}

\renewcommand{\thetheoremproof}{\ref{th:pre-mguandy}}
\begin{theoremproof}
%\label{th:pre-mguandy-proof}
Given $[S]_U \in \Linp$ and the binding $x/t$ with $\vars(x/t) \subseteq U$, we have that
  \[  
   \begin{split}
      \mgu_\an([S]_U,x/t) & = [(S\setminus S') \cup \\
      & \downclo \bigl\{ \Andybin Y \mid Y \in \mwp(S'),
      n \in \chi(Y,x) \cap \chi(Y,t).\ n \geq \card{Y}-1
   \bigr\}]_U \enspace ,
  \end{split}
  \]
  where $S'=\relev(S,x,t)$.
\end{theoremproof}
\begin{proof*}
By using Prop. \ref{prop:abstraction} point \ref{eq:abstraction2},
and since $o \neq o' \Rightarrow \gamma_\an(o) \cap
\gamma_\an(o')=\emptyset$, we get:
\begin{equation*}
  \begin{split}
    &   \alpha_\an( \gamma_\an(S) \setminus \relev(\gamma_\an(S),x,t)) \\
    = \ & \alpha_\an( \gamma_\an(S) \setminus \gamma_\an(
    \relev(S,x,t))) \\
    = \ & \alpha_\an( \gamma_\an(S \setminus \relev(S,x,t))) \\
    = \ & S \setminus \relev(S,x,t) \enspace .
\end{split}
\end{equation*}
Therefore, we get the equality
\begin{multline*}
  \mgu_\an([S]_U,x/t)=\Bigl[S \setminus \relev(S,x,t)  \cup \\
  \alpha_\an \Bigl(\{ \multisum \calS \mid \calS \in
  \mwp(\relev(\gamma_2(S),x,t)), \sum_{B \in \calS} \chi(B,x) = \sum_{B \in
    \calS} \chi(B,t) \geq \card{\calS}-1 \}\Bigr)\Bigr]_U \enspace .
\end{multline*}
Now, with simple algebraic manipulations, we obtain:
{\small
 \begin{align*}
   &\alpha_\an(
   \{ \multisum \calS \mid
      \calS \in \mwp(\relev( \gamma_\an(S),x,t)),
      \sum_{B \in \calS} \chi(B,x) = \sum_{B \in \calS}
      \chi(B,t) \geq \card{\calS}-1 \}) \\
 =\ &\alpha_\an(
   \{ \multisum \calS \mid
      \calS \in \mwp(\gamma_\an(\relev(S,x,t))),
      \sum_{B \in \calS} \chi(B,x) = \sum_{B \in \calS}
      \chi(B,t) \geq \card{\calS}-1 \} )\\
 =\ &\alpha_\an(
   \{ \multisum \multil B_1,\ldots,B_k \multir \mid   k\in \Nat,
      \\
   & \qquad  \forall i. B_i \in \gamma_\an(\relev(S,x,t)), \sum_{1\leq i \leq k} \chi(B_i,x) = \sum_{1\leq i \leq k}
      \chi(B_i,t) \geq k-1 \} )\\
=\ &\alpha_\an(
   \{ \multisum \multil B_1,\ldots,B_k \multir \mid   k\in \Nat,
   \multil o_1,\ldots,o_k \multir \in \mwp(\relev(S,x,t)),
      \\
   & \qquad  \forall i. B_i \in \gamma_\an(o_i), \sum_{1\leq i \leq k} \chi(B_i,x) = \sum_{1\leq i \leq k}
      \chi(B_i,t) \geq k-1 \} )\\
=\ &\alpha_\an(
   \{ \multisum \multil B_1,\ldots,B_k \multir \mid   k\in \Nat,
   \multil o_1,\ldots,o_k \multir \in \mwp(\relev(S,x,t)),
      \\
   & \qquad \forall i. \alpha_\an(B_i) = o_i,  \sum_{1\leq i \leq k} \chi(B_i,x) = \sum_{1\leq i \leq k}
      \chi(B_i,t) \geq k-1 \} )\\
    &  \qquad \qquad \text{(such $o_i$'s do always exist since $\relev(S,x,t)$
is downworld closed)} \\
=\ &\downclo
   \{ \alpha_\an(\multisum \multil B_1,\ldots,B_k \multir) \mid   k\in \Nat,
   \multil o_1,\ldots,o_k \multir \in \mwp(\relev(S,x,t)),
       \\
   & \qquad  \forall i. \alpha_\an(B_i) = o_i, \sum_{1\leq i \leq k} \chi(B_i,x) = \sum_{1\leq i \leq k}
      \chi(B_i,t) \geq k-1 \} \\
 =\ &\downclo
   \{ \multisum \multil \alpha_\an(B_1),\ldots,\alpha_\an(B_k) \multir \mid
    k\in \Nat,  \multil o_1,\ldots,o_k \multir \in \mwp(\relev(S,x,t)),  \\
   & \qquad  \forall i. \alpha_\an(B_i)=o_i, \sum_{1\leq i \leq k} \chi(B_i,x) = \sum_{1\leq i \leq k}
      \chi(B_i,t) \geq k-1 \} \\
&  \qquad \qquad \text{(by Prop. \ref{prop:abstraction} point \ref{eq:abstraction1})}\\
 =\ &\downclo
   \{ \multisum \multil o_1,\ldots,o_k \multir \mid   k\in \Nat,
   \multil o_1,\ldots,o_k \multir \in \mwp(\relev(S,x,t)),    \\
   & \qquad  \forall i. \alpha_\an(B_i) = o_i,  \sum_{1\leq i \leq k} \chi(B_i,x) = \sum_{1\leq i \leq k} \chi(B_i,t) \geq k-1 \} \}\\
 =\ &\downclo
   \{ \multisum \multil o_1,\ldots,o_k \multir \mid   k\in \Nat,
   \multil o_1,\ldots,o_k \multir \in \mwp(\relev(S,x,t)), \\
   & \qquad
     \forall i. \alpha_\an(B_i) = o_i, \forall i. \alpha_\an(B'_i) = o_i,
     \sum_{1\leq i \leq k} \chi(B_i,x)  = \sum_{1\leq i \leq k} \chi(B'_i,t)
     \geq k-1  \} \\
   & \qquad \qquad \text{(we discuss later why this is faithful)}  \\
 =\ &\downclo
   \{ \multisum \multil o_1,\ldots,o_k \multir \mid   k\in \Nat,
   \multil o_1,\ldots,o_k \multir \in \mwp(\relev(S,x,t)), n \geq k-1,\\
   & \qquad   n \in \{ \sum_{1\leq i \leq k}  B_i(x) ~|~
      \forall i. \alpha_\an(B_i) = o_i\} \cap \{ \sum_{1\leq i \leq k}
     \chi(B'_i,t)~|~ \forall i. \alpha_\an(B'_i) = o_i\}  \} \\
 =\ &\downclo
   \{ \multisum \multil o_1,\ldots,o_k \multir \mid   k\in \Nat,
   \multil o_1,\ldots,o_k \multir \in \mwp(\relev(S,x,t)),  n \geq k-1\\
   & \qquad  n \in [\sum_{1\leq i \leq k} {o_i}_m(x),\sum_{1\leq i \leq
k} o_i(x)]  \cap \{ \sum_{1\leq i \leq k}
      \chi(B'_i,t)~|~ \forall i. \alpha_\an(B'_i) = o_i\} \} \enspace .
\end{align*}
}

The move  from a single family $\{B_i\}_{1 \leq i \leq k}$ to different families $\{B_i\}_{1 \leq i \leq k}$ and $\{B'_i\}_{1 \leq i \leq k}$ is possible since, if
\[
 \forall i. \alpha_\an(B_i) = o_i \text{ and } \forall i. \alpha_\an(B'_i) = o_i \text{ and }
     \sum_{1\leq i \leq k} \chi(B_i,x)  = \sum_{1\leq i \leq k} \chi(B'_i,t)
     \geq k-1 \enspace ,
\]
we may define a family $\{C_i\}_{1 \leq i \leq k}$ such that $C_i(x)=B_i(x)$ and $C_i(v)=B'(v)$ for each $v \neq x$. It is immediate to check that the $C_i$'s satisfy the condition
\[
 \forall i. \alpha_\an(C_i) = o_i \text{ and }
     \sum_{1\leq i \leq k} \chi(C_i,x)  = \sum_{1\leq i \leq k} \chi(C_i,t)
     \geq k-1 \enspace .
\]

If we denote with $c(\multil o_1, \ldots, o_k \multir,t)$ the set $\big \{
\sum_{1\leq i \leq k} \chi(B_i,t)~|~ \forall i. \alpha_\an(B_i) = o_i\big\}$, what remains
to prove is that
\begin{align*}
 &\downclo
   \{ \Andybin X \mid X \in \mwp(\relev(S,x,t)),
   n \in \chi(X,x) \cap c(X,t). n \geq \card{X}-1 \}\\
=\ &
\downclo
  \{ \Andybin X \mid X \in \mwp(\relev(S,x,t)),
   n \in \chi(X,x)\cap \chi(X,t), n \geq \card{X}-1 \} \enspace , \notag
\end{align*}
where the only difference is that we replaced $c(X,t)$ with $\chi(X,t)$.

We begin by examining the relationship between $c(X,t)$ and $\chi(X,t)$. First of
all, it is obvious that $c(X,t) \subseteq \chi(X,t)$, therefore we only
need to prove half of the equality.

If there exists $o \in X$ such that $\chi_M(o,t)=\infty$, then $c(X,t)$ is an
infinite set. We call $n$ its least element. Under the same conditions, $\chi(X,
t)$ is the interval $[n,\infty]$. If there is no $o \in X$ such
that $\chi_M(o,t)=\infty$, then $c(X,t)=\chi(X,t)$ and they are both singletons.

In the same way, if there exists some $o \in X$ such that $o(x)=\infty$
then $\chi(X,x)$ is an interval of the kind $[n,\infty)$. However, if there is
no such $o$, then $\chi(X,x)$ is a singleton, whose unique element is $\card{\{o
\in X \mid o(x) =1 \}}$.

Assume that we have $X \in \mwp(\relev(S,x,t))$ such that
there exists $n \in \chi(X,x) \cap \chi(X,t)$ with $n \geq \card{X}-1$. We
want to prove that we may find a multiset $Y \in \mwp(\relev(S,x,t))$ such that
there exists $m \geq \card{Y}-1$ with $m \in \chi(Y,x) \cap c(Y,t)$ and
$\Andybin X \leq \Andybin Y$. This is enough to complete the proof of the
theorem.

We distinguish several cases.
\begin{itemize}
\item $\chi(X,x)$ and $\chi(X,t)$ are both infinite. In this case, $c(X,t)$ is
infinite. Moroever, since $\chi(X,x)$ is an interval, there are infinite natural
numbers in $\chi(X,x) \cap c(X,t)$. We may take $Y=X$.
\item $\chi(X,t)$ is infinite and $\chi(X,x)$ is a singleton $\{v\}$,
then $v=\card{\{o \in X \mid o(x)=1\}} \leq k$. Since it must be $v \geq k-1$,
there are only two choices: either $v=k$ or $v=k-1$. We distinguish the two
subcases.
\begin{itemize}
\item $v=k-1$. In this case, there exists $o \in X$ such that
$\chi_m(o,t)=0$ and $o(x)=1$, otherwise it is not possible that $v \geq
\chi_m(X,t)$. Since $\chi(X,t)$ is infinite, the same holds for $c(X,t)$, hence
we may find an $n \in c(X,t)$ such that $n \geq v$. Consider $Y=X \uplus
(n-v)*\multil o \multir$. We have $\chi(Y,x)=\{ v + (n-v) \}=n$, $c(Y,t)=c(X,t)$
and $\card{Y}=\card{X}+n-v=n+1$. Therefore $n \in c(Y,x) \cap c(Y,t)$ and $n
\geq \card{Y}-1$. $\Andybin Y$ is a valid result, and  $\Andybin X \leq \Andybin Y$.

\item $v=k$.  If there is an $o \in X$ such that $\chi_m(o,t)=0$, the proof
proceeds as in the previous case. Otherwise, $\chi_m(X,t) \geq k$ and
since it should be  $v=k \geq \chi_m(X,t)$, we have $\chi_m(X,t)=k$. Therefore
$k \in c(X,t)$ too, since $\min c(X,t)=\min \chi(X,t)$, and we may take $Y=X$.
\end{itemize}
\item if $\chi(X,t)$ is finite, then $\chi(X,t)=c(X,t)$ and we take $Y=X$. \exproofbox
\end{itemize}
\end{proof*}

\renewcommand{\thetheoremproof}{\ref{th:mguandy2}}
\begin{theoremproof}
Given $[S]_U$ in $\ShLinp$ and the binding $x/t$, let $V=\{v_1, \ldots, v_n\}$ be $\vars(x/t) \setminus U$. Then,
\[
  \mgu_\an([S]_U,x/t)= \mgu_\an([S \cup \{ v_1, \ldots, v_n \}]_{U \cup V},x/t) \enspace .
\]
\end{theoremproof}
\begin{proof}
  First of all, given a finite set of variables $V$, let us define the extension operator $\ext_V: \Linp \ra \Linp$ such that $\ext_V([S]_U) = [S \cup \{ \multil v \multir \mid v \in V \setminus U \}]_{U \cup V}$. Given $V=\vars(x/t) \setminus U=\{ v_1, \ldots, v_n\}$, we have that
  \begin{multline*}
  \mgu_\an([S]_U,x/t)=\alpha_\an(\mgu_\omega(\gamma_\an([S]_U), x/t)) \\= \alpha_\an(\mgu_\omega(\ext_V(\gamma_2([S]_U)),x/t))\enspace .
  \end{multline*}
  We also know that
  \begin{multline*}
   \mgu_\an([S \cup \{v_1, \ldots, v_n\}]_{U \cup V}, x/t)=\\
   \alpha_\an(\mgu_\omega(\gamma_2([S \cup \{  v_1, \ldots, v_n \}]_{U \cup V}), x/t)) \enspace .
  \end{multline*}
  Hence, it is enough to prove that
  \[
    \ext_V(\gamma_2([S]_U)) = \gamma_2([S \cup \{  v_1, \ldots, v_n \}]_{U \cup V} ) \enspace .
  \]
  By definition of $\gamma_2$, we have that
  \begin{align*}
    & \gamma_2([S \cup \{  v_1, \ldots, v_n \}]_{U \cup V}) \\
    =~ & \bigl[\bigcup \{ \gamma_2(o) \mid o \in S \cup \{ v_1, \ldots, v_n \} \}\bigr]_{U \cup V} \\
    =~ & \bigl[\bigcup \{ \gamma_2(o) \mid o \in S \} \cup \{ \multil v_1 \multir, \ldots, \multil v_n \multir \}\bigr]_{U \cup V} \qquad \text{[since $\gamma_2(v_i) = \multil v_i \multir$]}\\
    =~ & \ext_V(\gamma_2([S]_U)) \enspace ,
  \end{align*}
  which completes the proof.
\end{proof}

\renewcommand{\thetheoremproof}{\ref{th:mguandy}}
\begin{theoremproof}
Given $[S]_U \in \ShLinp$ and the  binding  $x/t$ with $\vars(x/t) \subseteq U$, we have
  \begin{equation*}
     \mgu_\an([S]_U,x/t)  = [(S\setminus S') \cup
      \downclo \bigcup_{X \subseteq S'} \res(X,x,t)]_U \enspace ,
   \end{equation*}
  where $S'=\relev(S,x,t)$ and $\res(X,x,t)$ is defined as follows:
\begin{enumerate}
\item if $X$ is non-linear for $x$ and $t$,  then $\res(X,x,t)=\{ \Andybin X^2\}$;
\item if $X$ is non-linear for $x$ and linear for $t$,
  $\card{X_x} \leq 1$ and $\card{X_t} \geq 1$, then we have
  $\res(X,x,t)=\{ (\Andybin X_x) \uplus (\Andybin X_{xt}^2) \uplus (\Andybin
  X_t^2)\}$;
\item if $X$ is linear for $x$ and strongly non-linear for $t$,
$\card{X_x} \geq 1$ and $\card{X_t} \leq 1$,
then we have
$\res(X,x,t)=\{ (\Andybin X_x^2) \uplus (\Andybin X_{xt}^2) \uplus (\Andybin X_t)\} $;
\item if $X$ is linear for $x$ and not strongly non-linear for $t$,
  $\card{X_t} \leq 1$, then we have
\[
\begin{split}
\res(X,x,t)=\{(\Andybin Z) \uplus
  (\Andybin X_{xt}^2) \uplus (\Andybin X_t) ~|~ & Z \in \mwp(X_x), \\
   & \card{Z}=\chi_M(X_t,t)=\chi_m(X_t,t), \\
   & \supp{Z} = X_x\} \enspace ;
\end{split}
\]
\item otherwise $\res(X,x,t)=\emptyset$.
\end{enumerate}
\end{theoremproof}
\begin{proof*}
By Theorem~\ref{th:pre-mguandy}, we only need to show that:
\begin{equation}
\label{eq:algo}
     \downclo \bigl\{ \Andybin Y \mid Y \in \mwp(S'),
      n \in \chi(Y,x) \cap \chi(Y,t).\ n \geq \card{Y}-1
      \bigr\} = \downclo \bigcup_{X \subseteq S'} \res(X,x,t) \enspace ,
\end{equation}
where $S'=\relev(S,x,t)$. We prove the two different inclusions separately.

\textbf{Left to Right Inclusion.}
Let $\bar{o} \in \res(X,x,t)$ for some $X \subseteq \relev(S,x,t)$. We want to prove
that there exist $Y \in \mwp(S')$ and $n \in \chi(Y,x) \cap \chi(Y,t)$ such that
$n \geq \card{Y}-1$ and $\andybin Y = \bar{o}$. We distinguish several cases:
\begin{itemize}
\item if $X$ is non-linear for $x$ and $t$, it is $\andybin X^2=\bar{o}$. We distinguish
two subcases:
\begin{itemize}
\item if $\chi_M(X,t)=\infty$, it is enough to take $Y=X \uplus X$.
\item if $\chi_M(X,t)$ is finite, since $X$ is non-linear for $t$, there exists $o' \in X$ such that $\chi_m(o',t)>1$. Since $S'$ is downward closed, consider $o \in S$ such that $o(x)=\min(o'(x),1)$ and $o(v)=o'(v)$ if $v \neq x$. We show that there exists a natural number $n$ such that, for $Y=X \uplus X \uplus n \multil o\multir$, we have $\chi_m(Y,t) \geq \chi_m(Y,x)$ and $\chi_m(Y,t) \geq \card{Y}-1$. Since $\chi_m(Y,x) \leq 2\chi_m(X,x) + n$, we need to solve the inequalities $2\chi_m(X,t)+n\chi_m(o,t) \geq 2\chi_m(X,x) + n$ and $2\chi_m(X,t)+n\chi_m(o,t) \geq 2\card{X} + n$. Since $\chi_m(o,t)\geq 2$, there always exists a solution for $n$. Since $\chi_M(X,x)=\infty$, we have that $\Andybin Y = \bar{o}$ is in the left hand side of \eqref{eq:algo}.
\end{itemize}

\item if $X$ is non-linear for $x$ and linear for $t$. We need to find
  $m$ such that, if we take $Y=X_x \uplus 2X_{xt} \uplus 2m X_t$, we
  have $\chi_m(Y,t) \geq \chi_m(Y,x)$. In other words, we need to
  solve the disequation $2 \chi_m(X_{xt},t) + 2m\chi_m(X_t,t) \geq
  \chi_m(X_x,x)+2\chi_m(X_{xt},x)$, which is always possible, since $\card{X_t} \geq 1$.
  Since $\card{Y} \leq 1+2\card{X_{xt}}+2m\card{X_t}$ we have $\chi_m(X,t)
  \geq \card{Y}-1$.

\item if $X$ is linear for $x$ and strongly non-linear for $t$, we distinguish
two subcases:
\begin{itemize}
\item $\chi_M(X,t)=\infty$. Let $n=2\chi_m(X_{xt},t)+ \chi_m(X_t,t)$
  and consider any number $m$ such that $2m\card{X_x}+2\card{X_{xt}}
  \geq n$ (such an $m$ always exists since $\card{X_x} \geq 1$). Then,
  consider the multiset $Y=2m X_x \uplus 2X_{xt} \uplus X_t$, and we
  have that $\chi_m(Y,x)=\chi_M(Y,x)=2m\card{X_x} + 2 \card{X_{xt}}
  \geq \chi_m(Y,t)$ by construction. Moreover $\chi_M(Y,t)=\infty$ and
  $\card{Y} \leq 2m\card{X_x}+2\card{X_{xt}}+1$. Then $\Andybin Y \in
  \res(X,x,t)$ is a valid resultant sharing group.
\item $\chi_M(X,t)$ is finite. Let $o \in X_{xt}$ be a sharing group
  such that $\chi_M(o,t)>1$ and $o'$ be a generic sharing group in
  $X_x$.  We need to find two natural numbers $n$ and $m$ such that,
  if we take $Y=2X_x \uplus 2X_{xt} \uplus X_t \uplus m \multil o
  \multir \uplus n \multil o' \multir$, we obtain
  $\chi_m(Y,x)=\chi_m(Y,t)$ (from this immediately follows
  $\chi_M(Y,x)=\chi_M(Y,t)$) and $\chi_m(Y,x) \geq |Y|-1$. This means
  we need to solve the equations:
\[
2\card{X_x} + 2\card{X_{xt}} + m +n = 2\chi_m(X_{xt},t) + \chi_m(X_t,t) + m \chi_m(o,t)
\]
\[
2\card{X_x} + 2\card{X_{xt}} + m +n \geq 2|X_x| + 2|X_{xt}|+ |X_t|+ m + n -1
\]
Since $|X_t| \leq 1$, the second equation is always satisfied. A solution for the first equation always exists, since the greatest common divisor of $\chi_m(o,t)-1$ and $1$ is $1$.
\end{itemize}

\item if $X$ is linear for $x$ and $X$ is not strongly non-linear for
  $t$, consider the multiset $Y=Z \uplus X_{xt} \uplus X_{xt} \uplus
  X_t$. Then $\chi_m(Y,x)=\chi_M(Y,x)=\card{Z} + 2\card{X_{xt}}$ and
  $\chi_m(Y,t)=\chi_M(Y,t)=2\card{X_{xt}} + \chi_m(X_t,t)$. Since
  $\card{Z}=\chi_m(X_t,t)$, we have that $\chi_m(Y,x)=\chi_m(Y,t)$.
  Moreover, $\card{Y}=\card{Z}+2\card{X_{xt}} + \card{X_t} \leq
  \chi_m(X_t,t) +2\card{X_{xt}} +1 = \chi_m(Y,t) +1$.
\end{itemize}

\textbf{Right to left inclusion.}
Let $o=\Andybin X$ where $X \in \mwp(S')$
and there exists $n \geq \card{X}-1$ such that $n \in \chi(X,x) \cap
\chi(X,t)$.  We show that there exists $Y \subseteq S'$ and $o' \in
\res(Y,x,t)$ such that $o' \geq_\an o$. Let $k=\card{X}$. We partition
 $X$ in three multisets $X_x = X|_{\{o \mid \chi_M(o,t)=0\}}$,
$X_t = X|_{\{o \mid \chi_M(o,x)=0\}}$ and $X|_{xt} = X|_{\{o \mid \chi_M(o,t)>0 \wedge
 \chi_M(o,t)>0 \}}$. Note that $X_x$, $X_t$ and $X_{xt}$ here are
multisets and not ordinary set as in the definition of $\mgu_\an$.
We distinguish several cases:
\begin{itemize}
  \item if $\supp{X}$ is linear for $x$ and strongly non-linear for
  $t$, then $\chi_m(X,x)=\chi_M(X,x) \allowbreak =\card{X_x}+\card{X_{xt}} \leq k$.
  Since $\chi_m(X,x) \geq k-1$, there are two cases: either
  $\card{X_x}+\card{X_{xt}}=k-1$ or $\card{X_x}+\card{X_{xt}}=k$,
  which implies that $\card{X_t} \leq 1 $.

  Since $\supp{X}$ is strongly non-linear for $t$, there exists $o''\in
  X_t \uplus X_{xt}$ such that $\chi_M(o'',t) \geq 2$, and thus
  $\chi_m(X,t) \geq 2$. Therefore $\chi_m(X,t) > | X_{xt}|$. Since
  $\chi_m(X,x) = \chi_M(X,x) \geq \chi_m(X,t)$, we have that
  $\card{X_x} \geq 1$. It  follows that $o =\Andybin (X_x \andybin X_{xt} \andybin X_t)
  \leq_\an   (\Andybin \supp{X_x})^2 \uplus (\Andybin \supp{X_{xt}})^2 \uplus
  (\Andybin \supp{X_t}) \in \res(\supp{X},x,t)$.

\item if $\supp{X}$ is linear for $x$ and not strongly non-linear for
  $t$, then, as in the previous case we have $\card{X_t} \leq 1$.
  Since $X$ is not strongly non-linear for $t$,
  $\chi_M(X,t)=\chi_m(X,t)=|X_{xt}|+\chi_M(X_t,t)$.  Moreover,
  $\chi_M(X,x)=\chi_m(X,x)= |X_x| + |X_{xt}|$. By the condition $n \in
  \chi(X,x) \cap \chi(X,t)$, we get $\chi_M(X_t,t)=\card{X_x}$.
  Therefore $o \leq_\an \Andybin \supp{X_x} \andybin (\Andybin \supp{X_{xt}})^2
  \andybin (\Andybin X_t) \in \res(\supp{X},x,t)$.

\item if $\supp{X}$ is non-linear for $x$ and $t$, then
  $o \leq_\an (\Andybin \supp{X})^2 \in \res(\supp{X},x,t)$.

\item if $\supp{X}$ is non-linear for $x$ and linear $t$, the proof is
  symmetric to the one of the first case. \exproofbox
\end{itemize}
\end{proof*}

\renewcommand{\thetheoremproof}{\ref{th:closedmgu}}
\begin{theoremproof}
\label{th:closedmgu-proof}
Given $[S]_U \in \ShLinp$ and the binding  $x/t$ with $\vars(x/t) \subseteq U$, we have
\[
    \mgu_\an([S]_U,x/t)  = [(S\setminus S') \cup
     \downclo \bigcup_{X \subseteq \max S'} (\res(X,x,t) \cup
     \res'(X,x,t))]_U \enspace ,
\]
where $S'=\relev(S,x,t)$ and
\[
\res'(X,x,t)=\begin{cases}
   \{\Andybin X^2\} & \text{if $X=X_{xt}$ and $l(X)$ is linear for $t$} \enspace ,\\
   \emptyset & \text{otherwise} \enspace .
\end{cases}
\]
\end{theoremproof}
\begin{proof}
It clearly holds that:
\begin{equation}
\mgu_\an([S]_U,x/t)  \supseteq [(S\setminus S') \cup
\downclo \bigcup_{X \subseteq \max S'}
( \res(X,x,t) \cup \res'(X,x,t) )]_U
\end{equation}
since, for each $X \subseteq \max S'$, if $\res'(X,x,t)$ is non-empty then
$\Andybin X^2$ may be generated by the characterization in Theorem \ref{th:mguandy}.  It
is enough to take $X'=\{ l(o) \mid o \in X \}$, hence $\Andybin
X'=\Andybin X^2 \in \res(X',x,t)$ according to the last case of
Theorem \ref{th:mguandy}.

We prove the opposite inclusion.  Let $X\subseteq S'$ and assume that
$X\nsubseteq \max S'$. There exists $X' \subseteq \max S'$ obtained by
replacing each $a \in X$ with $b \in \max S'$ such that $a \leq_\an b$.
We have that $\card{X'} \leq \card{X}$ since two different elements in $X$
may be replaced with the same maximal element in $X'$. We want to prove
that either $\res(X,x,t)=\emptyset$, or $\res(X,x,t) \subseteq \downclo
\res(X',x,t)$ or $\res(X,x,t) \subseteq \downclo \res'(X',x,t)$. Therefore, we
assume that $\res(X,x,t) \neq \emptyset$ and compare the linearity
properties (linear, non linear, strongly non-linear) of $X'$ \wrt
those of $X$.

If they coincide, then it follows that $\res(X,x,t) \subseteq \downclo
\res(X',x,t)$.  This happens because both $\res(X,x,t)$ and
$\res(X',x,t)$ are obtained by the same case of Theorem \ref{th:mguandy}. However, note that $X'$ may have less elements than
$X$ and therefore some variable which is non-linear in $\res(X,x,t)$
could be linear in $\res(X',x,t)$. Actually, this never happens since
the elements in $X'$ which are not explicitly delinearized are either
elements of the multiset $Z$ in the third case of Theorem \ref{th:mguandy}
(and therefore may appear multiple times)
or elements of $X_t$ ($X_x$) subject to the condition $\card{X_t} \leq
1$ ($\card{X_x} \leq 1$).

Assume that the linearity properties of $X$ and $X'$ do not coincide.
The only interesting case is when $X$ is linear for $x$ and not
strongly non-linear for $t$. In all the other cases, it is immediate
from the definition that $\res(X,x,t)\subseteq \downclo \res(X',x,t)$.

If $X'$ is not linear for $x$ and for $t$, then it holds
$\res(X,x,t)\subseteq \downclo \res(X',x,t)$ by definition.

If $X'$ is linear for $x$ and strongly non-linear for $t$, then it is
immediate from the definition that $\res(X,x,t)\subseteq \downclo \res(X',x,t)$,
provided that $\card{X_x} \geq 1$. Otherwise, it must be
$\card{X_t}=0$ and therefore, in order to be $\res(X,x,t) \neq
\emptyset$, we have $X=X_{xt}$ and $\chi_M(X,t)=1$, which means $l(X')=X$
is linear for $t$.  It follows that $\res(X,x,t)=\{\Andybin X^2\}=\res'(X',x,t)$.

If $X'$ is not linear for $x$ and  linear for $t$, we show
that $|X_x|\leq 1$.  Assume, by
contradiction, that $|X_x|> 1$. Since $X'$ is linear for $t$ and
$|X_t|\leq 1$, then $\chi_M(X_t,t)=\chi_m(X_t,t)\leq 1$, while
$|\supp{Z}| = |X_x| >1$, which is a contradiction. Thus it must be
$|X_x|\leq 1$.  If $|X_x|= 0$ then $\card{X_t}=0$, hence $\res(X,x,t)=
\{\Andybin X^2 \}$ and $\res(X,x,t)=\res'(X,x,t)$.  If $|X_x|= 1$,
since $X'$ is linear for $t$, it follows that $|Z|= 1$.  Thus
$\res(X,x,t)\subseteq \downclo  \res(X',x,t)$.
\end{proof}

\renewcommand{\thetheoremproof}{\ref{th:optimal-shl}}
\begin{theoremproof}
\label{th:optimal-shl-proof}
  The operator $\mgu_\shl$ in Definition~\ref{def:mgushl} is correct and optimal \wrt $\mgu$, when $\vars(x/t)\subseteq U$.
\end{theoremproof}
\begin{proof*}
  It is enough to prove that $\mgu_\shl$ is correct and optimal \wrt
  $\mgu_\an$, namely, that:
  \[
  \mgu_\shl([S,L,U]),x/t)=\alpha_\shl(
  \mgu_\an(\gamma_\shl([S,L,U]),x/t ) ) \enspace .
  \]
  Let $\gamma_\shl([S,L,U])=[T]_U$. By Theorem \ref{th:closedmgu-proof},
  it holds that:
  \begin{equation*}
  \begin{split}
    &\alpha_\shl( \mgu_\an(\gamma_\shl([S,L,U]),x/t)\\
    =\ &\alpha_\shl([(T\setminus T') \cup \downclo \bigcup_{Y \subseteq \max
      T'} (\res(Y,x,t) \cup \res'(Y,x,t) \})]_U \\
    =\ &\alpha_\shl\big([T\setminus T']_U \sqcup_2 \bigsqcup_{Y \subseteq \max
      T'} ([\downclo \res(Y,x,t)]_U \sqcup_2 [\downclo \res'(Y,x,t)\}]_U)\big) \enspace ,
  \end{split}
  \end{equation*}
  where $T'=\relev(T,x,t)$ and $\sqcup_2$ is the lowest upper bound in $\ShLinp$.
  By additivity of $\alpha_\shl$, this is equivalent to
  \begin{equation}
    \label{eq:shlp1}
    \alpha_\shl([T\setminus T']_U) \lub_\shl \biglub_{Y \subseteq \max
      T'} (\alpha_\shl([\res(Y,x,t)]_U) \lub_\shl
    \alpha_\shl([\res'(Y,x,t)]_U)) \enspace .
  \end{equation}
  Let $X$, $L'$, $U'$ and $K$ as in Definition \ref{def:mgushl}, we have that
  $\mgu_\shl([S,L,U],x/t)$ is equivalent to
  \begin{equation}
    \label{eq:shlp2}
     [ (S\setminus X) \cup K, U' \cup L', U] \enspace .
  \end{equation}
  We need to prove that equations \eqref{eq:shlp1} and
  \eqref{eq:shlp2} do coincide. In the rest of the paper, we assume that
  the result of \eqref{eq:shlp1} is $[S'',L'',U]$.

  \textbf{Sharing.} We first prove that the $\Sharing$ components of the two equations are
  equal, i.e. that $S''=(S \setminus X) \cup K$. Given $B \in S''$,
  there are several cases. If $B=\supp{o}$ with $o \in T \setminus
  T'$, then $B \in S \setminus X$.

  If $B=\supp{o}$, for $o \in \res'(Y,x,t)$ with $Y \subseteq
  \max T'$, then $B= \bigcup \{\supp{o} \mid o \in Y\}$ with
  $Y=Y_{xt}$ and $l(Y)$ is linear for $t$.
  If $x\in L$ then is generated by $(X_{xt}^{U})^+$, since $l(Y)$ is
  linear for $t$. If $x\notin L$ there are two cases: if $Y$ is linear
  for $t$ then it is generated by $(X_{xt}^{=1})^+$, otherwise by
  $\bin( X_t^{>1}\cup X_{xt}^{>1} , X_x \cup X_{xt},X^*)$. Thus $B \in
  K$.

  Now, assume that $B = \supp{o}$ with $o \in \res(Y,x,t)$ and $\emptyset \neq Y
  \subseteq \max T'$. Then $B=\bigcup W$ where $W=\{\supp{o} \mid o
  \in Y \}$. Since $Y$ is made of maximal elements and $[T]_U=
  \gamma_\an([S,L,U])$, we have that  $Y$ is linear for $x$ iff $x \in L$.
  For the same reason, $Y$ is linear for $t$ iff
  $(W,L)$ is linear for $t$.  As a consequence, if $Y$ is non-linear
  for $t$, then $(X,L)$ is non-linear for $t$.

  We proceed by cases:
  \begin{description}
  \item[$Y$ non-linear for $x$ and $t$.] Then $\res(Y,x,t)=\{ \Andybin
    Y^2\}$.  Since $(X,L)$ is non-linear for $x$ and $t$, we have
    $X_t^{>1}\cup X_{xt}^{>1} \neq \emptyset$ and $X_x \cup X_{xt}\neq
    \emptyset$. Thus $B \in \bin( X_t^{>1}\cup X_{xt}^{>1} , X_x\cup
    X_{xt}, X^*) \subseteq K$.
  \item[$Y$ non-linear for $x$ and linear for $t$.] By hypothesis
    $\card{Y_x} \leq 1$ and $\card{Y_t} \geq 1$, hence $o=(\Andybin
    Y_x) \uplus (\Andybin Y_{xt}^2) \uplus (\Andybin Y_t^2)$ and
    \[
    B \in \bin((X_t^{=1})^+, X_x \cup X_{xt}^{=1},(X_{xt}^{=1})^*)
    \subseteq K \enspace .
    \]
    In particular, $B \in \bin((X_t^{=1})^+, X_x,(X_{xt}^{=1})^*)$ when
    $\card{Y_x}=1$, otherwise $B \in \bin((X_t^{=1})^+, X_{xt}^{=1},(X_{xt}^{=1})^*)$.
  \item[$Y$ linear for $x$ and strongly non-linear for $t$.]  In this
    case we have that $o= (\Andybin Y_x^2) \uplus (\Andybin Y_{xt}^2) \uplus
    (\Andybin Y_t)$ with $\card{Y_x} \geq 1$ and $\card{Y_t} \leq 1$.
    By definition of strong non-linearity, we have two cases:
    \begin{itemize}
    \item there exists $o \in Y_{xt}$ such that $\chi_M(o,t) >1$: in
      this case
      \[B \in \bin(X_t\cup \{ \emptyset \}, X_{xt}^{>1},
      X_x^+, X_{xt}^* ) \subseteq K \enspace ;
      \]
    \item there exists $o\in Y_t$ such that $\chi_M(o,t) =\infty$: in
      this case
      \[
      B \in \bin(X_t^{=\infty}, X_x^+, X_{xt}^* ) \subseteq K \enspace .
      \]
    \end{itemize}
  \item[$Y$ linear for $x$ and non strongly non-linear for $t$.] In
    this case
    \[
    o = (\Andybin Z') \uplus (\Andybin Y_{xt}^2) \uplus
    (\Andybin Y_t) \enspace ,
    \]
    with $\card{Y_t}=1$, for some $Z' \in \mwp(Y_x)$ such that
    $\card{Z'}=\chi_m(Y_t,t)$ and $\supp{Z'}=Y_x$. It is obvious
    that
    \begin{equation*}
      B \in \bin (\{ \{o\} \cup (\cup Z)~|~ o\in X_t^{\in \Nat}, Z
      \subseteq X_x, 1 \leq |Z| \leq \chi_M^L(o,t)\}, (X_{xt}^{=1})^* ) \subseteq K \enspace ,
    \end{equation*}
    by choosing $Z=\{ \supp{o} \mid o \in Z' \}$.
  \end{description}

  This proves that if $B \in S''$, then $B \in (S \setminus X) \cup
  K$. Now, we need to prove the converse implication. If $B \in S
  \setminus X$, then $B=\supp{o}$ for some $o \in T$, and it is
  obvious that $o \in T \setminus T'$, hence $B \in S''$.

  Therefore, assume that $B \in K$, and consider the case when $x \in L$
  and $B \in \bin(X_t^{=\infty}, X_x^{+}, X_{xt}^*)$. We have that $B=A
  \cup (\cup A') \cup (\cup A'')$ for some $A \in X_t^{=\infty}$, $A'$
  non-empty subset of $X_x$ and $A'' \subseteq X_{xt}$. We may find
  $o' \in \max T'$, $Y', Y'' \subseteq \max T'$ such that
  $\supp{o'}=A$, $\supp{Y'}=A'$ and $\supp{Y''}=A''$. We have that
  $Y'''=\{o'\} \cup Y' \cup Y''$ is linear for $x$ and strongly
  non-linear for $t$ (due to the element $o'$), with $\card{Y'''_x}
  \geq 1$ and $\card{Y'''_t} \leq 1$. Therefore, we may apply the
  definition of $\res$ to obtain $\res(Y''',x,t)=\{o\}$ with $\supp{o}=B$,
  hence $B \in S''$.

  With similar reasonings, we may prove that for every $B \in K$, we
  have $B \in S''$. In particular: the second line of \eqref{eq:sl1}
  corresponds to the case we choose a $Y'''$ which is linear for $x$
  and strongly non-linear for $t$, due to an element $o \in Y'''_{xt}$
  which $\chi_M(o,t)>1$; the third line of \eqref{eq:sl1} corresponds
  to the case $Y'''$ is linear for $X$ and is not strongly non-linear
  for $t$; the first line of
  $\eqref{eq:sl2}$ corresponds to the case $Y'''$ is non-linear for
  both $x$ and $t$; the second line of \eqref{eq:sl2} corresponds to
  the case $Y'''$ is linear for $t$ and non-linear for $x$.

  Finally, if $x \notin L$ and $B \in (X_{xt}^{=1})^+$, it is possible
  that $B$ cannot be obtained as $\res(Y''',x,t)$ for any $Y'''
  \subseteq \max T'$.  However, $B$ may be obtained as
  $\res'(Y''',x,t)$, choosing $Y'''$ as in the previous cases. The
  same happens if $x \in L$and $B \in (X_{xt}^U)^+$.

  \textbf{Linearity.}
   We want to prove that $L''=L' \cup U'$. First of all, let us define
   $L''_g=U \setminus \vars(\mgu_\an([T]_U,x/t))$ the set of ground
   variables in $\mgu_\an([T]_U,x/t)$, hence $L''_g
   \subseteq L''$.  We are going to prove that $U'=L''_g$ and $L'
   \setminus U'=L'' \setminus L''_g$. The first equality trivially
   follows from the fact that the sharing component of $\mgu_\shl$ is
   optimal, hence a variable occurs in a sharing group of $S \setminus S \cup K$ iff it
   occurs in a 2-sharing group of $\mgu_\an([T]_U,x/t)$.

   Now, we consider a variable $v \in U \setminus U'$, and prove that
   $v \in L'$ iff $v \in L''$. There are several cases. If we assume
   that $v \notin L$, by \eqref{eq:linsharinglin} we have $v \notin L'$.
   Moreover, if $Y \in \max T'$ and $v \in \supp{Y}$, by maximality of
   $Y$ we have $Y(v)=\infty$. Hence, by Theorem
   \ref{th:closedmgu-proof}, we have $v \notin L''$. If we assume that $v \notin
   X$, by \eqref{eq:linsharinglin} we have $v \in L'$ iff $v \in L$.
   Since $\vars(X)=\vars(T)$, we also have $v \in L''$ iff $v \in L$
   and therefore $v \in L'$ iff $v \in L''$.

   The only case it remains to prove is $v \in \vars(X) \cap L$ which,
   combined with the condition $v \notin U'$, gives $v \in \vars(K)
   \cap L$. First of all, note that if $v \in \vars(X_{xt})$ then $v
   \notin L'$ (by definition of $L'$) and $v \notin L''$ (since $X_{xt}$)
   appears delinearized in every 2-sharing group resulting from $\res$ or
   $\res'$. If $v \notin \vars(X_{xt})$, we distinguish four subcases:
   \begin{itemize}
   \item $x \in L$ and $(S,L)$ linear for $t$. Given $Y \subseteq \max
     T'$, checking the forth case of Theorem \ref{th:mguandy} when
     $\chi_M(Y_t,t)=1$, we have that $\res(Y,x,t)$ is not linear for $v$
     iff $v \in \vars(Y_{xt})$ or $v \in \vars(Y_{x}) \cap
     \vars(Y_{t})$.  Note that there exists $Y \subseteq \max T'$ s.t.
     $v \in \vars(Y_{xt}) \cup (\vars(Y_x) \cap \vars(Y_t))$ iff $v
     \in \vars(T'_{xt}) \cup (\vars(T'_x) \cap \vars(T'_t))$.  Finally
     $v \in L''$ iff $v \in \vars(T'_{xt}) \cup (\vars(T'_x) \cap
     \vars(T'_t))$ iff $v \in (X_{xt} \cup (X_x \cap X_t))$ iff $v \in
     L'$.
   \item $x \in L$ and $(S,L)$ not linear for $t$. Given $Y \subseteq
     \max T'$, checking the third and forth cases (when
     $\chi_M(Y,t)>1$) of Theorem \ref{th:mguandy}, we have that
     $\res(Y,x,t)$ non-linear for $v$ implies $v \in \vars(Y_{xt})$ or
     $v \in \vars(Y_x)$, which is equivalent to $v \in X_{xt} \cup
     X_x$, i.e.  $v \notin L'$. On the other hand, if $v \in X_{x}$,
     we distinguish the cases:
     \begin{itemize}
     \item $(S,L)$ strongly non-linear for $t$. There exists $o \in
       T'$ such that $\chi_M(o,t)=\infty$ or $o \in T'_{xt}$ such that
       $\chi_M(o,t)>1$. Moreover, there exists $o' \in T'_{x}$ such
       that $v \in \supp{o'}$. If we take $Y=\{o,o'\}$, we have that
       $\res(Y,x,t)$ is not linear for $v$, hence $v \notin L''$.
     \item $(S,L)$ is not strongly non-linear for $t$. There exists $o
       \in T'_{t}$ such that $1 < \chi_M(o,t) < \infty$. Moreover,
       there exists $o' \in T'_x$ such that $v \in \supp{o'}$.  If we
       take $Y'=\{o,o'\}$, by the fourth case in the definition of
       $\res$, we have $\res(Y,x,t)$ is not linear for $v$, i.e. $v
       \notin L''$.
     \end{itemize}
   \item $x \notin L$ and $(S,L)$ linear for $t$. If $v \notin L''$
     then $v \in \vars(Y_{xt})$ or $v \in \vars(Y_t)$. This implies $v
     \in X_{xt} \cup X_t$, i.e. $v \notin L'$.  On the other hand, if
     $v \in X_t$, there exist $o \in T'_x$ such that
     $\chi_M(o,x)=\infty$ and $o' \in T'_t$ such that $v \in
     \supp{o'}$. By definition of $\res$, we have that
     $\res(\{o,o'\},x,y)$ is not linear for $v$, hence $v \notin L''$.
   \item $x \notin L$ and $(S,L)$ non-linear for $t$. Since $L'=L
     \setminus X$, it is obvious that $v \notin L'$. Moreover, there
     exist $o \in T'$ such that $\chi_M(o,x)=\infty$, $o' \in T'$ such
     that $\chi_M(o,t)>1$ and $o'' \in T'$ such that $v \in
     \supp{o''}$.  Note that it is possible that $o=o'=o''$. By
     definition, we have $\res(\{o,o',o''\},x,t)$ is not linear for
     $v$, hence $v \notin L''$. \exproofbox
   \end{itemize}
\end{proof*}

\renewcommand{\thetheoremproof}{\ref{th:optimal-shl2}}
\begin{theoremproof}
  \label{th:optimal-shl2-proof}
  The operator $\mgu_\shl$ in Definition \ref{def:mgushl2} is the optimal abstraction of $\mgu$.
\end{theoremproof}
\begin{proof}
   First of all, given a finite set of variables $V$, let us define the extension operator $\ext_V: \ShLinp \ra \ShLinp$ such that $\ext_V([S]_U) = [S \cup \{ v \mid v \in V \setminus U \}]_{U \cup V}$. Given $V=\vars(x/t) \setminus U$, we have that
  \[
  \alpha_\shl(\mgu_\an(\gamma_\shl([S,L,U], x/t)) = \alpha_\shl(\ext_V(\mgu_\an(\gamma_\shl([S,L,U])),x/t))\enspace .
  \]
  By Theorem \ref{th:optimal-shl-proof} we have that
  \begin{multline*}
   \mgu_\shl([S,L,U],x/t)=
   \mgu_\shl([S \cup V,L \cup V, U \cup V], x/t)=\\
   \alpha_\shl(\mgu_\an(\gamma_\shl([S \cup V, L \cup V, U \cup V],x/t)))  \enspace .
  \end{multline*}
  Hence, it is enough to prove that
  \[
    \ext_V(\gamma_\shl([S,L,U]) = \gamma_2([S \cup V, L \cup V, U \cup V]) \enspace .
  \]
  By definition of $\gamma_2$, we have that
  \begin{align*}
    & \gamma_\shl([S \cup V, L \cup V, U \cup V])= \\
    =~ & [ \{ B_{L \cup V} \mid B \in S\} \cup \{ B_{L \cup V} \mid B \in V\} ]_{U \cup V}\\
    =~ & [ \{ B_{L} \mid B \in S\} \cup V]_{U \cup V} \quad \text{[since $v_{L \cup V} = v$]}\\
    =~ &  \ext_V(\gamma_\shl([S,L,U]) \enspace ,
  \end{align*}
  which completes the proof.
\end{proof}

%\bibliographystyle{acmtrans}
%\bibliography{asbiblio}

\end{document}